\documentclass[journal]{IEEEtran}

\ifCLASSINFOpdf
  \usepackage[pdftex]{graphicx}
\else
\fi

\usepackage{amsmath}
\usepackage{amssymb}
\usepackage{amsthm}
\usepackage{mathtools}

\usepackage{algorithm, algpseudocode}

\newtheorem{theorem}{Theorem}
\newtheorem*{theorem*}{Theorem}
\newtheorem{lemma}[theorem]{Lemma}
\newtheorem*{lemma*}{Lemma}
\newtheorem{corollary}[theorem]{Corollary}
\newtheorem*{corollary*}{Corollary}
\newtheorem{proposition}[theorem]{Proposition}
\newtheorem*{proposition*}{Proposition}
\newtheorem{definition}[theorem]{Definition}
\newtheorem*{definition*}{Definition}
\newtheorem{remark}[theorem]{Remark}
\newtheorem*{remark*}{Remark}

\newtheorem*{example*}{Example}

\DeclareMathOperator*{\lexmin}{lex\,min}
\DeclareMathOperator*{\arglexmin}{arg\,lex\,min}

\ifCLASSOPTIONcompsoc
  \usepackage[caption=false,font=normalsize,labelfont=sf,textfont=sf]{subfig}
\else
  \usepackage[caption=false,font=footnotesize]{subfig}
\fi

\begin{document}

\title{Prioritized Inverse Kinematics: Nonsmoothness, Trajectory Existence, Task Convergence, Stability}

\author{Sang-ik An and Dongheui Lee% <-this % stops a space
	\thanks{This work was supported in part by Technical University of Munich - Institute for Advanced Study, funded by the German Excellence Initiative. {\it(Corresponding Author: Dongheui Lee.)}}% <-this % stops a space
	\thanks{The authors are with the Human-Centered Assistive Robotics, Department of Electrical and Computer Engineering, Technical University of Munich, D-80333 Munich, Germany. D. Lee is also with Institute of Robotics and Mechatronics, German Aerospace Center (DLR) (email: sangik.an@tum.de; dhlee@tum.de).}%
}

\maketitle

\begin{abstract}
In this paper, we study various theoretical properties of a class of prioritized inverse kinematics (PIK) solutions that can be considered as a class of (output regulation or tracking) control laws of a dynamical system with prioritized multiple outputs.
We first develop tools to investigate nonsmoothness of PIK solutions and find a sufficient condition for nonsmoothness.
It implies that existence and uniqueness of a joint trajectory satisfying a PIK solution cannot be guaranteed by the classical theorems.
So, we construct an alternative existence and uniqueness theorem that uses structural information of PIK solutions.
Then, we narrow the class of PIK solutions down to the case that all tasks are designed to follow some desired task trajectories and discover a few properties related to task convergence.
The study goes further to analyze stability of equilibrium points of the differential equation whose right hand side is a PIK solution when all tasks are designed to reach some desired task positions.
Finally, we furnish an example with a two-link manipulator that shows how our findings can be used to analyze the behavior of a joint trajectory generated from a PIK solution.
\end{abstract}


\IEEEpeerreviewmaketitle

\section{Introduction}

\IEEEPARstart{P}{riority} is a strategy to distribute a limited resource to multiple tasks. 
In the context of the prioritized inverse kinematics (PIK), the resource is the available degrees of freedom (DOF) of a mechanism and the distribution is carried out by the consecutive projections of the joint velocity to the null spaces of the higher priority tasks. 
The PIK problem has been studied intensively for decades in the robotic society and the study has been used and expanded in many areas such as constrained PIK \cite{Mansard2009}\cite{Kanoun2012}\cite{Escande2014}, task switching \cite{Lee2012}\cite{Petric2013}\cite{An2015}, prioritized control \cite{Khatib1987}\cite{Hsu1989}\cite{Sadeghian2014}\cite{Ott2015}, prioritized optimal control \cite{Nakamura1987a}\cite{Geisert2017}, learning prioritized tasks \cite{Saveriano2015}\cite{Calinon2018}\cite{Silverio2018}, etc.
Despite the large amount of studies on this topic, only few studies are found that reveal the theoretical aspects of the PIK problem.
Antonelli \cite{Antonelli2009} analyzed convergence of the task trajectories when all tasks are designed to reach some desired task positions.
Bouyarmane and Kheddar \cite{Bouyarmane2018} showed that the PIK solution found by the multi-objective optimization with the lexicographical ordering can be approximated in any accuracy by the multi-objective optimization with the weighted-sum scalarization and found some stability properties of the approximated PIK solution.
An and Lee proposed a generalization \cite{An2019} of the PIK problem by specifying three properties (dependence, uniqueness, and representation) for an objective function to be proper for the PIK problem.

In this paper, we study nonsmoothness, trajectory existence, task convergence, and stability of the PIK problem.
In Section \ref{sec:preliminary}, we define mathematical notations and prove some basic lemmas that will be used throughout the paper.
In Section \ref{sec:prioritized_inverse_kinematics}, we extend the definition of the PIK problem found in \cite{An2019} by relaxing the representation property of a proper objective function and define a class of PIK solutions of interest that can be considered as a class of (regulation or output tracking) control laws of a dynamical system that has prioritized multiple outputs.
In Section \ref{sec:nonsmoothness_of_pik_solutions}, we develop tools to investigate nonsmoothness of PIK solutions and show that existence and uniqueness of a joint trajectory satisfying a PIK solution cannot be guaranteed by the classical theorems such as the Peano's existence theorem or the contraction mapping theorem.
It motivates us to construct an alternative existence and uniqueness theorem that uses structural information of PIK solutions in Section \ref{sec:trajectory_existence}.
In many practical cases, the goal of prioritized tasks is to follow some desired task trajectories.
In Section \ref{sec:task_convergence}, we narrow the class of PIK solutions down to those practical cases and discover a few task convergence properties.
Our findings are better than the one Antonelli \cite{Antonelli2009} showed in the sense that they hold for every positive feedback gains and does not assume that the desired task trajectory is a constant function.
In Section \ref{sec:stability}, we analyze stability of equilibrium points of the differential equation whose right hand side is a PIK solution when the goal of the prioritized tasks is to reach some desired task positions.
In Section \ref{sec:example}, we provide an example with a two-link manipulator that shows how the properties we found can be used to analyze the behavior of the joint trajectory.
Finally, we give the concluding remarks in Section \ref{sec:conclusion}.

\section{Preliminary}
\label{sec:preliminary}

Let $X$ and $Y$ be finite dimensional Banach spaces over the field $\mathbb{R}$. 
$B_X$ is the closed unit ball in $X$. 
A function $f:D\subset X\to Y$ is said to be pointwise Lipschitz at $x_0\in D$ if there exist $r>0$ and $L\ge0$ such that $\|f(x)-f(x_0)\| \le L\|x-x_0\|$ provided $x \in x_0 + r B_X = \{x_0+rx'\mid x'\in B_X\}\subset D$. 
$0/L_p/L/1_p/1$-continuity represents continuity/pointwise Lipschitz continuity/local Lipschitz continuity/differentiability/continuous differentiability. 
Let $\bullet\in\mathcal{I} = \{0,L_p,L,1_p,1\}$. 
$\bullet$-discontinuity is the negation of $\bullet$-continuity. 
$C_\Omega^\bullet(D,Y)$, or simply $C_\Omega^\bullet$ if there is no confusion, is the set of all functions from $D$ to $Y$ that are $\bullet$-continuous at every $x\in\Omega\subset D$, equivalently on $\Omega$. 
Obvious inclusions are $C_\Omega^1 \subset C_\Omega^{1_p}\cap C_\Omega^L \subset C_\Omega^{1_p}\cup C_\Omega^L \subset C_\Omega^{L_p} \subset C_\Omega^0 \subset Y^D$ where $Y^D$ is the set of all functions from $D$ to $Y$. 
We also write $C_{\Omega}^\bullet(D)$, or simply $C_\Omega^\bullet$ if there is no confusion, to denote the set of all functions from $D$ to a finite dimensional Banach space that are $\bullet$-continuous on $\Omega\subset D$. 
For the sake of convenience, $C_{\{x\}}^\bullet = C_x^\bullet$ and $C_D^\bullet = C^\bullet$. 
The derivative of $f$ is denoted as $\mathsf{D}f$. 
If the variable of $X = \mathbb{R}$ is time, then we also write $\mathsf{D}f = \dot{f}$.
We recall that $f:[a,b]\subset\mathbb{R}\to Y$ is said to be absolutely continuous on $[a,b]$ if for every $\epsilon>0$ there exists $\delta>0$ such that if any finite collection of pairwise disjoint subintervals $\{[t_i',t_i'']\subset [a,b]\mid i\in I\}$ satisfies $\sum_{i\in I}|t_i''-t_i'|<\delta$, then $\sum_{i\in I}\|f(t_i'') - f(t_i')\|<\epsilon$.
$f$ is absolutely continuous on $[a,b]$ if and only if $f$ is differentiable almost everywhere on $[a,b]$ in the sense of Lebesgue, $\dot{f}$ is Lebesgue integrable on $[a,b]$, and $f(t) = f(a) + \int_a^t\dot{f}(s)ds$ for all $t\in[a,b]$.
All measures and integrals are Lebesgue's, so we do not mention it afterwards.
$\mathrm{AC}([a,b],Y)$ is the set of all absolutely continuous functions from $[a,b]$ to $Y$.
$L^1([a,b],Y)$ is the set of all integrable functions from $[a,b]$ to $Y$.
$\mathrm{AC}([a,\infty),Y)$ is the set all functions from $[a,\infty)$ to $Y$ that are absolutely continuous for every compact subinterval of $[a,\infty)$.
The fact that the countable union of sets of measure zero has measure zero implies that $f\in\mathrm{AC}([a,\infty),Y)$ if and only if $f$ is differentiable almost everywhere on $[a,\infty)$, $\dot{f}$ is integrable on every compact subinterval of $[a,\infty)$, and $f(t) = f(a) + \int_a^t\dot{f}(s)ds$ for all $t\in[a,\infty)$.
For $f:X\to\mathbb{R}$, we recall that $\limsup_{x\to x_0}f(x) = \lim_{\delta\to0} \sup \{f(x)\mid x\in (x_0+\delta B_X)\setminus\{x_0\}\}$ always exists in the extended real number system $[-\infty,\infty] = \mathbb{R}\cup\{+\infty,-\infty\}$. 
By convention, $-\infty < y < \infty$ for all $y\in\mathbb{R}$. 
If $y\in Y$, we denote $y:D\to Y$ to represent the constant function $f(x)\equiv y$. 
$\mathrm{int}(\Omega)$, $\mathrm{cl}(\Omega)$, and $\mathrm{bd}(\Omega)$ are interior, closure, and boundary of $\Omega$, respectively. 
For $a,b\in\mathbb{N}$, we use the shorthand notation $\overline{a,b} = \mathbb{N}\cap[a,b]$. 
For a set $S$, $2^S$ is the power set and $|S|$ is the cardinal number. 
The distance function $d:X\times2^X\to[0,\infty]$ is defined as $d(x,\Omega) = \inf\{\|x-x'\|\mid x'\in\Omega\}$ if $\Omega\neq\emptyset$ and $d(x,\emptyset) = \infty$.
$0^0 = 1$ by convention.

\begin{lemma}
	\label{lem:relation_between_pointwise_and_local_Lipschitz_continuity}
	If $f:D\to Y$ is pointwise Lipschitz on a convex set $\Omega\subset D$ with a uniform Lipschitz constant $L\ge0$, then $f$ is Lipschitz on $\Omega$ with $L$. 
\end{lemma}
\begin{proof}
	This is a simplified version of the proof found in \cite[Lemma 2.3]{Durand-Cartagena2010}.
	Let $x_1,x_2\in\Omega$. Define $x:\mathbb{R}\to X$ as $x(t) = x_1 + t(x_2-x_1)$. $t\in[0,1]$ implies $x(t) \in \Omega$ because $\Omega$ is convex. Thus, for each $t\in[0,1]$ there exists $r_t>0$ such that $\|f(x')-f(x(t))\|\le L\|x'-x(t)\|$ provided $x'\in x(t)+r_tB_X$. Also, for each $t\in[0,1]$ there exists $\delta_t>0$ such that $I_t = (t-\delta_t,t+\delta_t) \subset x^{-1}(x(t)+r_tB_X)$. Since $\{I_t\}_{t\in[0,1]}$ is an open cover of the compact set $[0,1]$, there is a finite subcover $\{I_{t_i}\}_{i=1}^N$. Without loss of generality, $I_{t_i}\not\subset I_{t_j}$ if $i\neq j$. Rearrange $\{t_i\}$ and choose $\tau_i\in I_{t_i}\cap I_{t_{i+1}}$ for $i\in\overline{1,N-1}$ such that $0\le t_1 \le \tau_1 \le \cdots \le \tau_{N-1} \le t_N \le 1$. Then, $\|f(x_1)-f(x_2)\| \le \|f(x(0)) - f(x(t_1))\| + \|f(x(t_1)) - f(x(\tau_1))\| + \cdots + \|f(x(t_N)) - f(x(1))\| \le L\|x_1-x_2\|$.
\end{proof}

\begin{lemma}[Gronwall's Inequality]
	\label{lem:gronwall_inequality}
	Let $-\infty<t_0<t_1<\infty$, $a,b\in L^1([t_0,t_1],\mathbb{R})$, and $\phi\in\mathrm{AC}([t_0,t_1],\mathbb{R})$. If $\dot{\phi}(t) \le a(t)\phi(t) + b(t)$ for almost all $t\in[t_0,t_1]$, then $\phi(t) \le \phi(t_0)e^{\int_{t_0}^t a(s)ds} + \int_{t_0}^t b(s)e^{\int_s^ta(r)dr}ds$
	for all $t\in[t_0,t_1]$.
\end{lemma}

Let $d,m,n\in\mathbb{N}$. $\mathbb{R}^d = \mathbb{R}^{d\times1}$ and $\mathbb{R}^{1\times d}$ are the sets of column vectors and row vectors, respectively. 
We assume the standard basis and do not distinguish matrices and linear transformations. 
For $\mathbf{x} = (x_1,\dots,x_d)\in\mathbb{R}^d$, $\mathbf{y} \in \mathbb{R}^{1\times d}$, $\mathbf{A} = [a_{ij}]\in \mathbb{R}^{m\times n}$, and $\mathbf{M} = [m_{ijk}] = ([m_{ij1}],\dots,[m_{ijd}]) \in \mathbb{R}^{m\times n}\times\cdots\times\mathbb{R}^{m\times n} = \mathbb{R}^{m\times n\times d}$, we define norms as $\|\mathbf{x}\| = \langle \mathbf{x},\mathbf{x}\rangle^{1/2} = (\mathbf{x}^T\mathbf{x})^{1/2}$, $\|\mathbf{y}\| = \langle \mathbf{y}, \mathbf{y}\rangle^{1/2} = (\mathbf{y}\mathbf{y}^T)^{1/2}$, $\|\mathbf{A}\| = \sup_{\|\mathbf{z}\|=1}\|\mathbf{A}\mathbf{z}\|$, $\|\mathbf{A}\|_F = (\sum_{i,j}a_{ij}^2)^{1/2}$, $\|\mathbf{M}\| = (\sum_{k=1}^d\|[m_{ijk}]\|^2)^{1/2}$, and $\|\mathbf{M}\|_F = (\sum_{i,j,k}m_{ijk}^2)^{1/2}$. $\mathbf{M}^T = [m_{jik}] = ([m_{ij1}]^T,\dots,[m_{ijd}]^T) \in \mathbb{R}^{n\times m\times d}$. 
$\mathbf{M}\mathbf{x} = \sum_{k=1}^d[m_{ijk}]x_k \in \mathbb{R}^{m\times n}$. 
For $\mathbf{B}\in\mathbb{R}^{a\times m}$ and $\mathbf{C}\in\mathbb{R}^{n\times b}$, $\mathbf{B}\mathbf{M}\mathbf{C} = (\mathbf{B}[m_{ij1}]\mathbf{C},\dots,\mathbf{B}[m_{ijd}]\mathbf{C})\in\mathbb{R}^{a\times b\times d}$. $B_d$ is the closed unit ball in $\mathbb{R}^d$. 
$\mathbf{I}_d \in\mathbb{R}^{d\times d}$ is the identity matrix. 
$\mathcal{R}(\mathbf{A}) = \{\mathbf{A}\mathbf{b}\mid\mathbf{b}\in\mathbb{R}^n\}$ and $\mathcal{N}(\mathbf{A}) = \{\mathbf{b}\in\mathbb{R}^n\mid\mathbf{A}\mathbf{b}=\mathbf{0}\}$. 
We recall that $\mathcal{N}(\mathbf{A})^\perp = \{\mathbf{c}\in\mathbb{R}^n\mid \mathbf{c}^T\mathbf{b} = 0,\,\mathbf{b}\in\mathcal{N}(\mathbf{A})\} = \mathcal{R}(\mathbf{A}^T)$.
If $\mathbf{A}$ is square ($m = n$), then $\det(\mathbf{A}) = |\mathbf{A}|$ is the determinant of $\mathbf{A}$.
$\mathbf{A}^+\in\mathbb{R}^{n\times m}$ is the (Moore-Penrose) pseudoinverse of $\mathbf{A}$. 
For $\lambda\in[0,\infty]$, we define the (extended) damped pseudoinverse of $\mathbf{A}$ with the damping constant $\lambda$ as
\begin{equation*}
\mathbf{A}^{*(\lambda)} = \begin{dcases*}
\mathbf{A}^T(\mathbf{A}\mathbf{A}^T + \lambda^2\mathbf{I}_m)^+, & $\lambda \in [0,\infty)$ \\
\mathbf{0}, & $\lambda = \infty$
\end{dcases*}
\end{equation*}
where $\mathbf{A}^{*(0)} = \mathbf{A}^+$ has zero damping and $\mathbf{A}^{*(\infty)} = \mathbf{0}$ has infinite damping.
We may simply write $\mathbf{A}^*$ for $\mathbf{A}^{*(\lambda)}$ if $\lambda$ is defined ahead.
The next Lemma can be easily proven from $\|\mathbf{A}\| \le \|\mathbf{A}\|_F \le \sqrt{\mathrm{rank}(\mathbf{A})}\|\mathbf{A}\|$ and $\|\mathbf{M}\| \le \|\mathbf{M}\|_F \le \sqrt{\min\{m,n\}}\|\mathbf{M}\|$.

\begin{lemma}
	\label{lem:derivative_of_matrix_and_its_entries}
	$\mathbf{A} = [a_{ij}] \in C_{\mathbf{x}_0}^\bullet(\mathbb{R}^d,\mathbb{R}^{m\times n})$ if and only if $a_{ij}\in C_{\mathbf{x}_0}^\bullet(\mathbb{R}^d,\mathbb{R})$ for all $i\in\overline{1,m}$ and $j\in\overline{1,n}$.
\end{lemma}

\begin{lemma}
	\label{lem:ContinuityOfPseudoinverse-modified}
	Let $\mathbf{A}\in C_{\mathbf{x}_0}^\bullet(\mathbb{R}^d,\mathbb{R}^{m\times n})$. Then, $\mathbf{A}^+\in C_{\mathbf{x}_0}^\bullet(\mathbb{R}^d,\mathbb{R}^{n\times m})$ if and only if $\lim_{\mathbf{x}\to \mathbf{x}_0}\mathrm{rank}(\mathbf{A}(\mathbf{x})) = \mathrm{rank}(\mathbf{A}(\mathbf{x}_0))$.
\end{lemma}
\begin{proof}
	Fix $c \in (0,1)$ and define $\mathbf{A}_0 = \mathbf{A}(\mathbf{x}_0)$, $\Delta = \mathbf{A} - \mathbf{A}_0$, and
	\begin{align*}
	\Delta_1 &= (\mathbf{A}_0\mathbf{A}_0^+)\Delta(\mathbf{A}_0^+ \mathbf{A}_0) \\
	\Delta_2 &= (\mathbf{A}_0\mathbf{A}_0^+)\Delta(\mathbf{I}_n - \mathbf{A}_0^+ \mathbf{A}_0) \\
	\Delta_3 &= (\mathbf{I}_m - \mathbf{A}_0\mathbf{A}_0^+)\Delta(\mathbf{A}_0^+ \mathbf{A}_0) \\
	\Delta_4 &= (\mathbf{I}_m - \mathbf{A}_0\mathbf{A}_0^+)\Delta(\mathbf{I}_n - \mathbf{A}_0^+ \mathbf{A}_0) \\
	\Omega_c &= \{\mathbf{x} \in \mathbb{R}^d \mid \|\mathbf{A}_0^+\|\|\Delta_1(\mathbf{x})\| < 1 - c\}.
	\end{align*}
	
	$\impliedby$: There exists $r>0$ such that $\mathrm{rank}(\mathbf{A}(\mathbf{x})) = \mathrm{rank}(\mathbf{A}_0)$ for all $\mathbf{x} \in \mathbf{x}_0+r B_d \subset \Omega_c$. We can prove the case $\bullet\in\{0,L_p\}$ from
	\begin{align*}
		\|\mathbf{A}^+(\mathbf{x}) - \mathbf{A}_0^+\| &\le \|\mathbf{A}_0^+\|\left( \beta_1 + \gamma\sum_{i=2}^4\sqrt{\frac{\beta_i^2}{1 + \beta_i^2}} \right)(\mathbf{x}) \\
			&\le \|\mathbf{A}_0^+\| \left(\beta_1 + \gamma \sum_{i=2}^4 \beta_i\right)(\mathbf{x}) \\
			&\le \frac{c+3}{c^2}\|\mathbf{A}_0^+\|^2\|\mathbf{A}(\mathbf{x})-\mathbf{A}_0\|
	\end{align*}
	for all $\mathbf{x}\in \mathbf{x}_0 + r B_d$ where $\gamma = (1 - \|\mathbf{A}_0^+\|\|\Delta_1\|)^{-1}$ and $\beta_i = \gamma\|\mathbf{A}_0^+\|\|\Delta_i\|$ \cite{Stewart1969}. If $\bullet = L$, then there exist $0<r_1\le r$ and $L_1\ge0$ such that $\|\mathbf{A}(\mathbf{x}_1) - \mathbf{A}(\mathbf{x}_2)\| \le L_1 \|\mathbf{x}_1-\mathbf{x}_2\|$ for all $\mathbf{x}_1,\mathbf{x}_2 \in \mathbf{x}_0 + r_1 B_d$. Let $0<r_2<r_1$. Since $\lim_{\mathbf{x'}\to\mathbf{x}}\mathrm{rank}(\mathbf{A}(\mathbf{x'})) = \mathrm{rank}(\mathbf{A}(\mathbf{x}))$ for all $\mathbf{x}\in\mathbf{x}_0+r_2B_d$, $\mathbf{A}^+ \in C_{\mathbf{x}_0 + r_2B_d}^{L_p}$ by the case $\bullet\in\{0,L_p\}$. Since $\mathbf{x}\mapsto\|\mathbf{A}^+(\mathbf{x})\|$ is continuous on the compact set $\mathbf{x}_0+r_2B_d$, $L = \frac{c+3}{c^2}L_1\max_{\mathbf{x}\in\mathbf{x}_0+r_2B_d}\|\mathbf{A}^+(\mathbf{x})\|\in[0,\infty)$ exists such that $\mathbf{A}^+\in C_{\mathbf{x}_0+r_2B_d}^{L_p}$ with the uniform Lipschitz constant $L$. Then, $\mathbf{A}^+$ is Lipschitz on $\mathbf{x}_0+r_2B_d$ with $L$ by Lemma \ref{lem:relation_between_pointwise_and_local_Lipschitz_continuity}. If $\bullet = 1_p$, then $\mathsf{D}\mathbf{A}_0 = \mathsf{D}\mathbf{A}(\mathbf{x}_0) = (\mathsf{D}_1\mathbf{A}(\mathbf{x}_0),\dots,\mathsf{D}_d\mathbf{A}(\mathbf{x}_0)) \in\mathbb{R}^{m\times n\times d}$ exists where $\mathsf{D}_i\mathbf{A} = \partial\mathbf{A}/\partial x_i$ is the partial derivative of $\mathbf{A}$ with respect to the $i$-th component of $\mathbf{x} = (x_1,\dots,x_d)$. By \cite{Golub1973}, we have $\mathsf{D}\mathbf{A}^+(\mathbf{x}_0) = -\mathbf{A}_0^+ \mathsf{D}\mathbf{A}_0 \mathbf{A}_0^+ + \mathbf{A}_0^+ \mathbf{A}_0^{+ T} \mathsf{D}\mathbf{A}_0^T (\mathbf{I}_m - \mathbf{A}_0\mathbf{A}_0^+) + (\mathbf{I}_n - \mathbf{A}_0^+ \mathbf{A}_0) \mathsf{D}\mathbf{A}_0^T \mathbf{A}_0^{+ T}\mathbf{A}_0^+$.
	If $\bullet = 1$, then there exists $0<r_3<r$ such that $\mathsf{D}\mathbf{A}(\mathbf{x})$ exists for all $\mathbf{x} \in \mathbf{x}_0+r_3B_d$ and is continuous at $\mathbf{x}_0$. Since $\lim_{\mathbf{x'}\to\mathbf{x}}\mathrm{rank}(\mathbf{A}(\mathbf{x'})) = \mathrm{rank}(\mathbf{A}(\mathbf{x}))$ for all $\mathbf{x}\in\mathbf{x}_0+r_3B_d$, $\mathsf{D}\mathbf{A}^+(\mathbf{x})$ exists for all $\mathbf{x}\in\mathbf{x}_0+r_3B_d$ by the case $\bullet = 1_p$. $\mathbf{A}$, $\mathbf{A}^+$, and $\mathsf{D}\mathbf{A}$ are all continuous at $\mathbf{x}_0$ and so is $\mathsf{D}\mathbf{A}^+$. 
	
	$\implies$: If there exists a sequence $\{\mathbf{x}_i\}$ converging to $\mathbf{x}_0$ with $\mathrm{rank}(\mathbf{A}(\mathbf{x}_i)) \neq \mathrm{rank}(\mathbf{A}_0)$ for all $i\in\mathbb{N}$, then $\|\mathbf{A}^+(\mathbf{x}_i) - \mathbf{A}_0^+\| \ge \|\mathbf{A}^+(\mathbf{x}_i)\| - \|\mathbf{A}_0^+\| \ge \|\Delta(\mathbf{x}_i)\|^{-1} - \|\mathbf{A}_0^+\| \to \infty$ as $i\to\infty$ \cite{Stewart1969}.
\end{proof}

\begin{lemma}
	\label{lem:Af_continuous_iff_f_continuous}
	Let $\mathbf{f}:\mathbb{R}^d\to\mathbb{R}^n$, $\mathbf{A}\in C_{\mathbf{x}_0}^\bullet(\mathbb{R}^d,\mathbb{R}^{n\times n})$, and $\lim_{\mathbf{x}\to \mathbf{x}_0}\mathrm{rank}(\mathbf{A}(\mathbf{x})) = \mathrm{rank}(\mathbf{A}(\mathbf{x}_0)) = n$. Then, $\mathbf{A}\mathbf{f}\in C_{\mathbf{x}_0}^\bullet$ if and only if $\mathbf{f}\in C_{\mathbf{x}_0}^\bullet$.
\end{lemma}
\begin{proof}
	The `if' part is obvious by Lemma \ref{lem:derivative_of_matrix_and_its_entries}. Let $\mathbf{A}\mathbf{f}\in C_{\mathbf{x}_0}^\bullet$. There exists $r>0$ such that $\mathrm{rank}(\mathbf{A}(\mathbf{x})) = n$ and $\mathbf{A}^+(\mathbf{x}) = \mathbf{A}^{-1}(\mathbf{x})$ provided $\mathbf{x}\in \mathbf{x}_0+rB_d$. $\mathbf{A}^+\in C_{\mathbf{x}_0}^\bullet$ by Lemma \ref{lem:ContinuityOfPseudoinverse-modified}. Since $\mathbf{f}(\mathbf{x}) = (\mathbf{A}^+\mathbf{A}\mathbf{f})(\mathbf{x})$ for all $\mathbf{x}\in\mathbf{x}_0+rB_d$, $\mathbf{f} \in C_{\mathbf{x}_0}^\bullet$.
\end{proof}

\begin{lemma}
	\label{lem:orthogonalization_of_J}
	Let $\mathbf{J}\in\mathbb{R}^{m\times n}$ with $m\le n$. There exist a lower triangular matrix $\mathbf{C}_e = [c_{ij}] \in\mathbb{R}^{m\times n}$ and an orthogonal matrix $\hat{\mathbf{J}}_e \in \mathbb{R}^{n\times n}$ such that $\mathbf{J} = \mathbf{C}_e\hat{\mathbf{J}}_e$; $c_{aa} \ge 0$ for $a\in\overline{1,m}$; and $c_{ab} = 0$ for $a\in\overline{1,m}$ if $c_{bb} = 0$.
\end{lemma}
\begin{proof}
	This proof can be found in \cite{An2019}.
	Run the next algorithm.
	\begin{algorithmic}[1]
		\State $\mathbf{V} = \begin{bmatrix} \mathbf{v}_1 & \cdots & \mathbf{v}_m \end{bmatrix} \gets \mathbf{J}^T$
		\State $\tilde{\mathbf{Q}}_r = \begin{bmatrix} \mathbf{q}_1 & \cdots & \mathbf{q}_m \end{bmatrix} \gets \mathbf{0} \in \mathbb{R}^{n\times m}$
		\State $\mathbf{R}_r = [r_{ij}] \gets \mathbf{0}\in\mathbb{R}^{m\times m}$
		\For{$a = 1$ \textbf{to} $m$}
		\State $r_{aa} = \|\mathbf{v}_a\|$
		\If{$r_{aa} > 0$}
		\State $\mathbf{q}_a = \mathbf{v}_a / r_{aa}$
		\If{$a<m$}
		\For{$b = a + 1$ \textbf{to} $m$}
		\State $r_{ab} = \langle\mathbf{q}_a, \mathbf{v}_b\rangle$
		\State $\mathbf{v}_b = \mathbf{v}_b - r_{ab}\mathbf{q}_a$
		\EndFor
		\EndIf
		\EndIf
		\EndFor
	\end{algorithmic}
	Then, we get the decomposition $\mathbf{J}^T = \tilde{\mathbf{Q}}_r\mathbf{R}_r$ where $\mathbf{R}_r$ is upper triangular; $r_{aa}\ge0$ for $a\in\overline{1,m}$; $r_{ab} = 0$ for $b\in\overline{1,m}$ if $r_{aa} = 0$; $\mathbf{q}_a = \mathbf{0}$ if and only if $r_{aa} = 0$; and nonzero columns of $\tilde{\mathbf{Q}}_r$ are orthonormal. Define $\tilde{\mathbf{Q}}\in\mathbb{R}^{n\times n}$ and $\mathbf{R}\in\mathbb{R}^{n\times m}$ as $\tilde{\mathbf{Q}} = \tilde{\mathbf{Q}}_r$ and $\mathbf{R} = \mathbf{R}_r$ if $m = n$ and $\tilde{\mathbf{Q}} = \begin{bmatrix} \tilde{\mathbf{Q}}_r & \mathbf{0}\end{bmatrix}$ and $\mathbf{R} = \begin{bmatrix} \mathbf{R}_r^T & \mathbf{0} \end{bmatrix}^T$ if $m < n$. There are only $r = \mathrm{rank}(\mathbf{J})$ nonzero columns in $\tilde{\mathbf{Q}}$. Let $\{\mathbf{p}_1,\dots,\mathbf{p}_{n-r}\}\subset\mathbb{R}^n$ be an orthonormal basis of $\mathcal{N}(\mathbf{J})$. Construct $\mathbf{Q}\in\mathbb{R}^{n\times n}$ by replacing zero columns of $\tilde{\mathbf{Q}}$ with $\mathbf{p}_1,\dots,\mathbf{p}_{n-r}$. Since $\mathcal{R}(\tilde{\mathbf{Q}}) = \mathcal{R}(\mathbf{J}^T)\perp\mathcal{N}(\mathbf{J})$, $\mathbf{Q}$ is orthogonal and we get the full QR decomposition $\mathbf{J}^T = \mathbf{Q}\mathbf{R} = \hat{\mathbf{J}}_e^T\mathbf{C}_e^T$.
\end{proof}

\begin{lemma}
	\label{lem:bound_of_extended_damped_pseudoinverse}
	Let $\mathbf{A}\in\mathbb{R}^{m\times n}$, $l = \min\{m,n\}$, $r = \mathrm{rank}(\mathbf{A})$, $\mu,\nu\in[0,\infty)$, and
	\begin{equation*}
	\lambda^2 = \begin{dcases*}
	0, & $\mu = 0$ \\
	\frac{\mu^{2}}{|\mathbf{A}\mathbf{A}^T|^\nu}, & $\mu>0$, $|\mathbf{A}\mathbf{A}^T|^\nu > 0$ \\
	\infty, & $\mu>0$, $|\mathbf{A}\mathbf{A}^T|^\nu = 0$.
	\end{dcases*}
	\end{equation*}
	Then, $\|\mathbf{A}^{*(\lambda)}\| \le \min\{M_1,M_2\}$
	where
	\begin{align*}
	M_1 &= \begin{dcases*} 0, & $\mathbf{A} = \mathbf{0}$ \\ \frac{1}{\sigma_r}, & $\mathbf{A} \neq \mathbf{0}$, \end{dcases*}
	&M_2 &= \begin{dcases*} 
	\infty, & $\mu = 0$ \\
	\frac{1}{2\mu}\prod_{i=1}^l\sigma_i^\nu, & $\mu>0$,
	\end{dcases*}
	\end{align*}
	and $\sigma_1 \ge \cdots \ge \sigma_l\ge0$ are singular values of $\mathbf{A}$. 
	If $\mu>0$, then $\|\mathbf{A}^{*(\lambda)}\| \le \frac{1}{2\mu}\|\mathbf{A}\|^{\nu l}$.
\end{lemma}
\begin{proof}
	Let $M = \min\{M_1,M_2\}$. If $\mathbf{A} = \mathbf{0}$, then $\mathbf{A}^{*(\lambda)} = \mathbf{0}$ and $\|\mathbf{A}^{*(\lambda)}\| = 0 = M_1 = M$. Let $\mathbf{A} \neq \mathbf{0}$. If $\mu = 0$, then $\mathbf{A}^{*(\lambda)} = \mathbf{A}^+$ and $\|\mathbf{A}^{*(\lambda)}\| = \sigma_{\max}(\mathbf{A}^+) = 1/\sigma_r = M_1 = M$. Let $\mu>0$. If $\nu = 0$, then $\mathbf{A}^{*(\lambda)} = \mathbf{A}^T(\mathbf{A}\mathbf{A}^T + \mu^2\mathbf{I}_m)^{-1}$ and $\|\mathbf{A}^{*(\lambda)}\| = \sigma_{\max}(\mathbf{A}^T(\mathbf{A}\mathbf{A}^T + \mu^2\mathbf{I}_m)^{-1}) = \max\{\sigma_1/(\sigma_1^2+\mu^2),\dots,\sigma_r/(\sigma_r^2+\mu^2)\}.$ Since $\sigma_i/(\sigma_i^2+\mu^2) < \sigma_i/\sigma_i^2 = 1/\sigma_i \le 1/\sigma_r = M_1$ and $\sigma_i/(\sigma_i^2+\mu^2) \le \sup\{\sigma/(\sigma^2+\mu^2)\mid\sigma\in(0,\infty)\} = 1/(2\mu) = M_2$ for $i\in\overline{1,r}$, we have $\|\mathbf{A}^{*(\lambda)}\| \le M$. Let $\nu>0$. If $|\mathbf{A}\mathbf{A}^T| = 0$, then $\mathbf{A}^{*(\lambda)} = \mathbf{0}$ and $\|\mathbf{A}^{*(\lambda)}\| = 0 = M_2 = M$. If $|\mathbf{A}\mathbf{A}^T| = \prod_{i=1}^l\sigma_i^2>0$, then $\lambda^2 = \mu^{2}/\prod_{i=1}^l\sigma_i^{2\nu}$, $\mathbf{A}^{*(\lambda)} = \mathbf{A}^T(\mathbf{A}\mathbf{A}^T + \lambda^2\mathbf{I}_m)^{-1}$, and $\|\mathbf{A}^{*(\lambda)}\| = \max\{\sigma_1/(\sigma_1^2+\lambda^2),\dots,\sigma_l/(\sigma_l^2+\lambda^2)\} \le \min\{1/\sigma_l,1/(2\lambda)\} = M$. If $\mu>0$, then $\|\mathbf{A}^{*(\lambda)}\| \le M_2 \le 1/(2\mu)\prod_{i=1}^l\sigma_1^\nu = 1/(2\mu)\|\mathbf{A}\|^{\nu l}$.
\end{proof}
\begin{lemma}
	\label{lem:smoothness_of_extended_damped_pseudoinverse}
	Let $\mathbf{x}_0\in\mathbb{R}^d$, $\mathbf{A}\in C_{\mathbf{x}_0}^\bullet(\mathbb{R}^d,\mathbb{R}^{m\times n})$, $\mu\in(0,\infty)$, and $\nu\in\mathbb{N}\cup\{0\}$. Define $\lambda:\mathbb{R}^n\to(0,\infty]$ as
	\begin{equation*}
	\lambda^2(\mathbf{x}) = \begin{dcases*}
	\frac{\mu^{2}}{|(\mathbf{A}\mathbf{A}^T)(\mathbf{x})|^\nu}, & $|(\mathbf{A}\mathbf{A}^T)(\mathbf{x})|^\nu > 0$ \\
	\infty, & $|(\mathbf{A}\mathbf{A}^T)(\mathbf{x})|^\nu = 0$.
	\end{dcases*}
	\end{equation*}
	Then, $\mathbf{A}^{*(\lambda)} \in C_{\mathbf{x}_0}^\bullet(\mathbb{R}^d,\mathbb{R}^{n\times m})$.
\end{lemma}
\begin{proof}
	Denote $f = |\mathbf{A}\mathbf{A}^T|^\nu$ and $\mathbf{G} = f\mathbf{A}\mathbf{A}^T + \mu^2\mathbf{I}_m$. We recall that the determinant of $(\mathbf{A}\mathbf{A}^T)(\mathbf{x}) = [m_{ij}]$ can be written as 
	\begin{equation*}
	|(\mathbf{A}\mathbf{A}^T)(\mathbf{x})| = \sum_{(j_1,\dots,j_m)}\prod_{\alpha<\beta}\mathrm{sgn}(j_\beta-j_\alpha)\prod_{i=1}^mm_{ij_i}
	\end{equation*}
	where $\mathrm{sgn}(x) = 1$ if $x\in(0,\infty)$, $\mathrm{sgn}(x) = -1$ if $x\in(-\infty,0)$, $\mathrm{sgn}(0) = 0$, and the sum extends over all ordered $m$-tuple of integers $(j_1,\dots,j_m)$ with $1\le j_\alpha \le m$. It is immediate that $f,\mathbf{G}\in C_{\mathbf{x}_0}^\bullet$. Since $\mathrm{rank}(\mathbf{G}(\mathbf{x})) = m$ for all $\mathbf{x}\in\mathbb{R}^d$, $\mathbf{G}^{-1}\in C_{\mathbf{x}_0}^\bullet$ by Lemma \ref{lem:ContinuityOfPseudoinverse-modified}. Therefore, $\mathbf{A}^{*(\lambda)} = f\mathbf{A}^T\mathbf{G}^{-1} \in C_{\mathbf{x}_0}^\bullet$.
\end{proof}

\section{Prioritized Inverse Kinematics}
\label{sec:prioritized_inverse_kinematics}

A preconditioned kinematic system with multiple tasks or a kinematic system for short is a 6-tuple $\mathsf{S} = (l,\mathbf{m},n,\mathbf{F},\mathbf{R},\mathbf{r})$ defined as follows:
\begin{itemize}
	\item $l\in\mathbb{N}\setminus\{1\}$ is the number of tasks;
	\item $\mathbf{m} = (m_1,\dots,m_l)\in\mathbb{N}^l$ where $m_a\in\mathbb{N}$ is the dimension of the $a$-th task space;
	\item $n\in\mathbb{N}$ is the dimension of the joint space and $m = m_1 + \cdots + m_l \le n$ is assumed;
	\item $\mathbf{x} = (t,\mathbf{q}) \in \mathbb{R}\times\mathbb{R}^n = X$ where $X$ is the domain of $\mathsf{S}$;
	\item $\mathbf{F} = \begin{bmatrix} \mathbf{f}_t & \mathbf{F}_q\end{bmatrix} = \begin{bmatrix} \mathbf{f}_{t1} & \mathbf{F}_{q1} \\ \vdots & \vdots \\ \mathbf{f}_{tl} & \mathbf{F}_{ql} \end{bmatrix} = \begin{bmatrix} \mathbf{F}_1 \\ \vdots \\ \mathbf{F}_l\end{bmatrix} : X\to\mathbb{R}^{m\times(n+1)}$ where $\mathbf{F}_a:X\to\mathbb{R}^{m_a\times(n+1)}$ is the $a$-th velocity mapping function with $\mathbf{f}_{ta}:X\to\mathbb{R}^{m_a}$ and $\mathbf{F}_{qa}:X\to\mathbb{R}^{m_a\times n}$ that maps the joint velocity $\dot{\mathbf{q}}$ into the $a$-th task velocity $\mathbf{f}_{ta} + \mathbf{F}_{qa}\dot{\mathbf{q}}$;
	\item $\mathbf{R}:X\to\{\mathbf{M}\in\mathbb{R}^{n\times n}\mid \det(\mathbf{M})\neq 0\}$ is the (right) preconditioner function;
	\item $\mathbf{r} = (\mathbf{r}_1,\dots,\mathbf{r}_l):X\to\mathbb{R}^m$ where $\mathbf{r}_a:X\to\mathbb{R}^{m_a}$ is the $a$-th reference function.
\end{itemize}
The assumption $m \le n$ is not that restrictive because if $m > n$, then we can redefine $\mathsf{S}$ to $\tilde{\mathsf{S}} = (l,\mathbf{m},\tilde{n},\tilde{\mathbf{F}},\tilde{\mathbf{R}},\tilde{\mathbf{r}})$ by introducing dummy variables or virtual joints $q_{n+1},\dots,q_m$ such that $\tilde{\mathbf{x}} = (t,q_1,\dots,q_n,q_{n+1},\dots,q_m)$, $\tilde{\mathbf{F}}(\tilde{\mathbf{x}}) = \begin{bmatrix} \mathbf{F}_t(\mathbf{x}) & \mathbf{F}_q(\mathbf{x}) & \mathbf{0} \end{bmatrix}$, $\tilde{\mathbf{R}}(\tilde{\mathbf{x}}) = \begin{bmatrix} \mathbf{R}(\mathbf{x}) & \mathbf{0} \\ \mathbf{0} & \mathbf{I}_{m-n} \end{bmatrix}$, and $\tilde{\mathbf{r}}(\tilde{\mathbf{x}}) = \mathbf{r}(\mathbf{x})$.
In a special case that there exists the $a$-th forward kinematic function $\mathbf{f}_a\in C^{1_p}(X,\mathbb{R}^{m_a})$ satisfying $\mathsf{D}\mathbf{f}_a = \begin{bmatrix} \frac{\partial\mathbf{f}_a}{\partial t} & \frac{\partial\mathbf{f}_a}{\partial\mathbf{q}}\end{bmatrix} = \begin{bmatrix} \mathsf{D}_t\mathbf{f}_a & \mathsf{D}_q\mathbf{f}_a \end{bmatrix} = \begin{bmatrix} \mathbf{f}_{ta} & \mathbf{F}_{qa}\end{bmatrix} = \mathbf{F}_a$, we can write the $a$-th task velocity as $\dot{\mathbf{f}}_a = \mathsf{D}_t\mathbf{f}_a + \mathsf{D}_q\mathbf{f}_a\dot{\mathbf{q}}$.
Let $\mathbf{r}_a' = \mathbf{r}_a - \mathbf{f}_{ta}$, $\mathbf{r}' = \mathbf{r} - \mathbf{f}_t$, $\mathbf{J}_a = \mathbf{F}_{qa}\mathbf{R}^{-1}$, and $\mathbf{J} = \mathbf{F}_q\mathbf{R}^{-1}$. 
$\mathbf{R}$ is introduced in consideration of the preconditioning of $\mathbf{F}_{q1},\dots,\mathbf{F}_{ql}$. 
A specific choice of $\mathbf{R}$ and its effect is discussed in \cite{An2014}. 
One may let $\mathbf{R} = \mathbf{I}_n$ to ignore this part. 
We say that $\mathsf{S}$ is $\bullet$-continuous if $\mathbf{F},\mathbf{R},\mathbf{r}\in C^\bullet$. 
We define $\mathbb{S}$ as the set of all kinematic systems and $\mathbb{S}^\bullet = \{\mathsf{S}\in\mathbb{S}\mid \textrm{$\mathsf{S}$ is $\bullet$-continuous}\}$. 

The $a$-th task of $\mathsf{S}$ is the 2-tuple $\mathsf{T}_a = (\mathbf{r}_a,\mathbf{F}_a)$ where $\mathbf{r}_a$ represents the desired behavior of the task velocity $\mathbf{f}_{ta} + \mathbf{F}_{qa}\dot{\mathbf{q}}$. 
Thus, the goal of $\mathsf{T}_a$ is to find the joint velocity $\dot{\mathbf{q}}^*$ that minimizes the $a$-th residual $\mathbf{e}_a^\mathrm{res} = \mathbf{r}_a - \mathbf{f}_{ta} - \mathbf{F}_{qa}\dot{\mathbf{q}} = \mathbf{r}_a' - \mathbf{J}_a\mathbf{R}\dot{\mathbf{q}}$ in some sense to be explained later.
$\mathrm{rank}(\mathbf{J}_a(\mathbf{x})) = \mathrm{rank}(\mathbf{F}_{qa}(\mathbf{x})) \le m_a$ is the maximum available DOF for $\mathsf{T}_a$ that is needed to achieve the goal of $\mathsf{T}_a$ at $\mathbf{x}$. 
In total, there are $\mathrm{rank}(\mathbf{J}(\mathbf{x})) = \mathrm{rank}(\mathbf{F}_q(\mathbf{x}))$ available DOF for $\mathsf{T}_1,\dots,\mathsf{T}_l$ and $\mathrm{rank}(\mathbf{J}(\mathbf{x})) \le \sum_{a=1}^l\mathrm{rank}(\mathbf{J}_a(\mathbf{x})) \le m$ by singularity. 
So, the available DOF is the limited common resource necessary for all tasks and we need a strategy how to distribute it. 
We assign priority to tasks $\mathsf{T}_1,\dots,\mathsf{T}_l$ to make prioritized tasks $\mathsf{T}_1\prec\cdots\prec\mathsf{T}_l$ by demanding the next two properties:
\begin{description}
	\item[(P1)] $\mathsf{T}_a$ does not influence $\mathsf{T}_1,\dots,\mathsf{T}_{a-1}$;
	\item[(P2)] $\mathsf{T}_a$ uses the maximum available DOF needed to achieve the goal of $\mathsf{T}_a$ under (P1).
\end{description}
(P2) claims that doing nothing or unnecessary things does not preserve priority. 
The goal of $\mathsf{T}_1\prec\cdots\prec\mathsf{T}_l$ is to find $\dot{\mathbf{q}}^*$ that minimizes $\mathbf{e}_a^\mathrm{res}$ for $a\in\overline{1,l}$ in some sense under (P1) and (P2).  
Then, the PIK problem can be considered as a problem to find a control law that regulates multiple prioritized outputs of the dynamical system
\begin{align*}
	\dot{\mathbf{q}} &= \mathbf{u} \\
	\mathbf{e}_a^\mathrm{res} &= \mathbf{r}_a'(t,\mathbf{q}) - \mathbf{J}_a(t,\mathbf{q})\mathbf{R}(t,\mathbf{q})\mathbf{u},\quad a\in\overline{1,l} \\
	\mathsf{T}_1 &\prec \cdots \prec \mathsf{T}_l
\end{align*}
where $\mathbf{q}\in\mathbb{R}^n$ is the state, $\mathbf{u}\in\mathbb{R}^n$ is the control input, and $\mathsf{T}_1\prec\cdots\prec\mathsf{T}_l$ represents the priority relations between multiple outputs $\mathbf{e}_a^\mathrm{res}\in\mathbb{R}^{m_a}$ for $a\in\overline{1,l}$ in this case.

We define equivalence relations on $\mathbb{S}$ and $\mathbb{S}^\bullet$. 
Let $\mathsf{S} = (l,\mathbf{m},n,\mathbf{F},\mathbf{R},\mathbf{r})\in\mathbb{S}$. 
$\mathbf{F}$ is defined by a mechanism and an environment, $\mathbf{R}$ is constructed from $\mathbf{F}_q$, and $\mathbf{r}$ is designed by a scenario. 
Usually, multiple scenarios are applied for a mechanism in an environment, so we need to consider various $\mathbf{r}$ given $\mathbf{F}$ and $\mathbf{R}$. 
We say that $\tilde{\mathsf{S}} = (\tilde{l},\tilde{\mathbf{m}},\tilde{n},\tilde{\mathbf{F}},\tilde{\mathbf{R}},\tilde{\mathbf{r}})\in\mathbb{S}$ is equivalent to $\mathsf{S}$ on $\mathbb{S}$ and denote $\mathsf{S}\sim\tilde{\mathsf{S}}$ if $(l,\mathbf{m},n,\mathbf{F},\mathbf{R}) = (\tilde{l},\tilde{\mathbf{m}},\tilde{n},\tilde{\mathbf{F}},\tilde{\mathbf{R}})$. 
The equivalence class of $\mathsf{S}$ in $\mathbb{S}$ is denoted as $[\mathsf{S}] = \{\tilde{\mathsf{S}}\in\mathbb{S}\mid\mathsf{S}\sim\tilde{\mathsf{S}}\}$. 
The equivalence relation $\stackrel{\bullet}{\sim}$ on $\mathbb{S}^\bullet$ and the equivalence class $[\mathsf{S}]^\bullet$ of $\mathsf{S}$ in $\mathbb{S}^\bullet$ are defined similarly.
Note that $[\mathsf{S}]^\bullet\subset[\mathsf{S}]$ for $\mathsf{S}\in\mathbb{S}^\bullet$. 
Let $\mathbb{S}_e\subset\mathbb{S}$ be an equivalence class in $\mathbb{S}$ and $\mathsf{S}\in\mathbb{S}_e$. 
Obviously, $\mathbb{S}_e = [\mathsf{S}]$. 
So, we write $\mathsf{S}\in[\mathsf{S}]\subset\mathbb{S}$ to denote an arbitrary equivalence class $[\mathsf{S}]$ in $\mathbb{S}$ and an arbitrary kinematic system $\mathsf{S}$ of $[\mathsf{S}]$. 
$\mathsf{S}\in[\mathsf{S}]^\bullet\subset\mathbb{S}^\bullet$ has the similar meaning. 
Every member of $[\mathsf{S}]$ or $[\mathsf{S}]^\bullet$ shares the same $l$, $\mathbf{m}$, $n$, $\mathbf{F}$, $\mathbf{R}$, and $\mathbf{J}$. 
We orthogonalize rows of $\mathbf{J}$ by performing the full QR decomposition of $\mathbf{J}^T(\mathbf{x})$ at each $\mathbf{x}\in X$ as in Lemma \ref{lem:orthogonalization_of_J}
\begin{equation*}
\underbrace{\begin{bmatrix} \mathbf{J}_1 \\ \vdots \\ \mathbf{J}_l \end{bmatrix}}_{\mathbf{J}(\mathbf{x})\in\mathbb{R}^{m\times n}} = \underbrace{\begin{bmatrix} \mathbf{C}_{11} & \cdots & \mathbf{0} & \mathbf{0} \\ \vdots & \ddots & \vdots & \vdots \\ \mathbf{C}_{l1} & \cdots & \mathbf{C}_{ll} & \mathbf{0} \end{bmatrix}}_{\mathbf{C}_e(\mathbf{x}) = [\mathbf{C}_{ij}(\mathbf{x})] \in \mathbb{R}^{m\times n}} \underbrace{\begin{bmatrix} \hat{\mathbf{J}}_1 \\ \vdots \\ \hat{\mathbf{J}}_{l+1} \end{bmatrix}}_{\hat{\mathbf{J}}_e(\mathbf{x})\in\mathbb{R}^{n\times n}}.
\end{equation*}
Define orthogonal-projector-valued functions $\mathbf{P}_a:X\to\mathbb{R}^{n\times n}$ for $a\in\overline{1,l}$ as
\begin{equation*}
\mathbf{P}_a(\mathbf{x}) = ((\mathbf{C}_{aa}\hat{\mathbf{J}}_a)^+(\mathbf{C}_{aa}\hat{\mathbf{J}}_a))(\mathbf{x}) = (\hat{\mathbf{J}}_a^T\mathbf{C}_{aa}^+\mathbf{C}_{aa}\hat{\mathbf{J}}_a)(\mathbf{x}).
\end{equation*}
Then, $\mathbf{C}_{ab}\hat{\mathbf{J}}_b = \mathbf{J}_a\mathbf{P}_b$ by Lemma \ref{lem:orthogonalization_of_J} and the $a$-th residual can be written as
\begin{equation*}
\mathbf{e}_a^\mathrm{res} = \mathbf{r}_a' - \mathbf{J}_a\mathbf{R}\dot{\mathbf{q}} = \mathbf{r}_a' - \mathbf{J}_a\sum_{b=1}^a\mathbf{P}_b\mathbf{R}\dot{\mathbf{q}}.
\end{equation*}

We may represent a goal of a task of a kinematic system as an optimization problem. 
Since we are considering various references given $\mathbf{F}$ and $\mathbf{R}$, the optimization problem should be defined for all equivalent kinematic systems. 
Let $\pi_1,\dots,\pi_l:X\times\mathbb{R}^n\times[\mathsf{S}]\to[0,\infty]$ be objective functions that describe how to minimize $\mathbf{e}_1^\mathrm{res},\dots,\mathbf{e}_l^\mathrm{res}$ for each $\mathbf{x}\in X$ and $\mathsf{S}\in[\mathsf{S}]$ such that the goal of $\mathsf{T}_a$ at $\mathbf{x}$ is to find $\dot{\mathbf{q}}^*$ that minimizes $\pi_a(\mathbf{x},\mathbf{R}(\mathbf{x})\dot{\mathbf{q}},\mathsf{S})$ with respect to $\dot{\mathbf{q}}$. 
Let $\mathbf{y} = \mathbf{R}(\mathbf{x})\dot{\mathbf{q}}$. 
Since $\mathbf{R}$ is invertible everywhere, to find such $\dot{\mathbf{q}}^*$ is equivalent to find $\mathbf{y}^* = \mathbf{R}(\mathbf{x})\dot{\mathbf{q}}^*$ that minimizes $\pi_a(\mathbf{x},\mathbf{y},\mathsf{S})$ with respect to $\mathbf{y}$. 
Then, the goal of $\mathsf{T}_1\prec\cdots\prec\mathsf{T}_l$ at $\mathbf{x}$ is to find $\mathbf{y}^*$ that minimizes $\pi_a(\mathbf{x},\mathbf{y},\mathsf{S})$ with respect to $\mathbf{y}$ for $a\in\overline{1,l}$ under (P1) and (P2). 
Not every objective function is proper in the context of the PIK.
For example, if $\pi_a(\mathbf{x},\mathbf{y},\mathsf{S}) = \|\mathbf{y}\|$ for all $a\in\overline{1,l}$, then we have a trivial solution $\dot{\mathbf{q}}^* = \mathbf{R}^{-1}(\mathbf{x})\mathbf{y}^* = \mathbf{0}$ for all $(\mathbf{x},\mathsf{S})$ that is not appropriate.
An and Lee \cite{An2019} proposed three properties for a vector-valued objective function $\boldsymbol{\pi} = (\pi_1,\dots,\pi_l):X\times\mathbb{R}^n\times[\mathsf{S}]\to[0,\infty]^l$ to be proper for the PIK problem:
\begin{description}
	\item[(O1)] $\forall(a,\mathbf{x},\mathbf{y},\mathsf{S})\in\overline{1,l}\times X\times\mathbb{R}^n\times[\mathsf{S}]$, $\pi_a(\mathbf{x},\mathbf{y},\mathsf{S}) = \pi_a(\mathbf{x},\sum_{b=1}^a\mathbf{P}_b(\mathbf{x})\mathbf{y},\mathbf{r}_a(\mathbf{x}))$;
	\item[(O2)] $\forall(a,\mathbf{x},\mathsf{S})\in\overline{1,l}\times X\times[\mathsf{S}]$, there exists a unique minimizer $\mathbf{y}_a^*$ of $\pi_a(\mathbf{x},\sum_{b=1}^{a-1}\mathbf{y}_b^* + \mathbf{y},\mathbf{r}_a(\mathbf{x}))$ subject to $\mathbf{y}\in\mathcal{R}(\mathbf{P}_a(\mathbf{x}))$;
	\item[(O3)] $\forall(a,\mathbf{x})\in\overline{1,l}\times X$, the mapping $\mathbf{r}_a(\mathbf{x})\mapsto\mathbf{y}_a^*$ of $\mathcal{R}((\mathbf{C}_{aa}\hat{\mathbf{J}}_a)(\mathbf{x}))$ into $\mathcal{R}(\mathbf{P}_a(\mathbf{x}))$ is one-to-one and onto.
\end{description}
We say that $\boldsymbol{\pi}$ is {\it strongly proper} for $[\mathsf{S}]$ if $\boldsymbol{\pi}$ has properties (O1) to (O3); {\it weakly proper} for $[\mathsf{S}]$ if $\boldsymbol{\pi}$ has properties (O1) and (O2) only; and {\it proper} for $[\mathsf{S}]$ if it is either strongly proper or weakly proper.
We also say that $\boldsymbol{\pi}$ is (strongly or weakly) proper for $[\mathsf{S}]^\bullet$ if the domain of $\boldsymbol{\pi}$ is restricted to $X\times\mathbb{R}^n\times[\mathsf{S}]^\bullet$.

The minimization of a proper objective function under (P1) and (P2) can be written as the multi-objective optimization with the lexicographical ordering \cite{An2019}. 
Consider multiple objective functions $\phi_a:\mathbb{R}^n\to[0,\infty]$ for $a\in\overline{1,l}$ and a constraint set $\Omega\subset\mathbb{R}^n$. 
The problem
\begin{equation*}
\lexmin_{\mathbf{y}\in\Omega}(\phi_1(\mathbf{y}),\dots,\phi_l(\mathbf{y}))
\end{equation*}
is to find an optimal solution $\mathbf{y}^*\in\Omega$ satisfying
\begin{align*}
\phi_a(\mathbf{y}^*) = \min\{&\phi_a(\mathbf{y})\mid\mathbf{y}\in\Omega \textrm{ and } \\
&\phi_b(\mathbf{y}) = \phi_b(\mathbf{y}^*) \textrm{ for } b \in\overline{1,a-1}\}
\end{align*}
for $a\in\overline{1,l}$. 
We say that a map $\mathbf{u}:X\times[\mathsf{S}]\to\mathbb{R}^n$ is a strongly-prioritized / weakly-prioritized / prioritized inverse kinematics (SPIK / WPIK / PIK) solution of $[\mathsf{S}]$ if there exists a strongly-proper / weakly-proper / proper objective function $\boldsymbol{\pi}$ for $[\mathsf{S}]$ satisfying
\begin{align*}
\mathbf{u}(\mathbf{x},\mathsf{S}) &= \mathbf{R}^{-1}(\mathbf{x})\mathbf{v}(\mathbf{x},\mathsf{S}) \\
\mathbf{v}(\mathbf{x},\mathsf{S}) &= \arglexmin_{\mathbf{y}\in\mathbb{R}^n}(\pi_1(\mathbf{x},\mathbf{y},\mathsf{S}),\cdots,\pi_l(\mathbf{x},\mathbf{y},\mathsf{S}),\|\mathbf{y}\|^2/2)
\end{align*}
for every $(\mathbf{x},\mathsf{S})$. 
We also say that $\mathbf{u}$ is a SPIK or WPIK or PIK solution of $[\mathsf{S}]^\bullet$ if the domain of $\mathbf{u}$ is restricted to $X\times[\mathsf{S}]^\bullet$.
The $\boldsymbol{\pi}$-PIK solution of $[\mathsf{S}]$ or $[\mathsf{S}]^\bullet$ is the PIK solution determined by the proper objective function $\boldsymbol{\pi}$ for $[\mathsf{S}]$ or $[\mathsf{S}]^\bullet$.
In this paper, we study a class of PIK solutions of $[\mathsf{S}]$ that can be written as
\begin{equation}
	\label{eqn:class_of_pik_solutions}
	\mathbf{u} = \mathbf{R}^{-1}\hat{\mathbf{J}}^T\mathbf{C}_D^T\underbrace{\begin{bmatrix} \mathbf{L}_{11} & \mathbf{0} & \cdots & \mathbf{0} \\ \mathbf{L}_{21} & \mathbf{L}_{22} & \cdots & \mathbf{0} \\ \vdots & \vdots & \ddots & \vdots \\ \mathbf{L}_{l1} & \mathbf{L}_{l2} & \cdots & \mathbf{L}_{ll} \end{bmatrix}}_{\mathbf{L} = [\mathbf{L}_{ij}]:X\to\mathbb{R}^{m\times m}}\mathbf{r}'
\end{equation}
where $\hat{\mathbf{J}}:X\to\mathbb{R}^{m\times n}$ is the top ($m\times n$) block of $\hat{\mathbf{J}}_e$, $\mathbf{C}_D = \mathrm{diag}(\mathbf{C}_{11},\dots,\mathbf{C}_{ll}):X\to\mathbb{R}^{m\times m}$ is block diagonal whose diagonal blocks are $\mathbf{C}_{11},\dots,\mathbf{C}_{ll}$ starting from the top left corner, and $\mathbf{L} = [\mathbf{L}_{ij}]:X\to\mathbb{R}^{m\times m}$ is block lower triangular with $\mathbf{L}_{ab}:X\to\mathbb{R}^{m_a\times m_b}$.

We introduce four PIK solutions of this type.
Define objective functions $\boldsymbol{\pi}_\alpha = (\pi_{\alpha1},\dots,\pi_{\alpha l})$ for $\alpha\in\overline{1,4}$ as:
\begin{align*}
	\pi_{1a}(\mathbf{x},\mathbf{y},\mathsf{S}) &= \|\mathbf{r}_a'(\mathbf{x}) - \mathbf{J}_a(\mathbf{x})\mathbf{y}\|^2 + \lambda_a^2(\mathbf{x})\|\mathbf{P}_a(\mathbf{x})\mathbf{y}\|^2 \\
	\pi_{2a}(\mathbf{x},\mathbf{y},\mathsf{S}) &= \|\mathbf{J}_a^*(\mathbf{x})\mathbf{r}_a'(\mathbf{x}) - \mathbf{P}_a(\mathbf{x})\mathbf{y}\|^2 \\
	\pi_{3a}(\mathbf{x},\mathbf{y},\mathsf{S}) &= \|\mathbf{r}_a'(\mathbf{x}) - (\mathbf{C}_{aa}\hat{\mathbf{J}}_a)(\mathbf{x})\mathbf{y}\|^2 + \lambda_a^2(\mathbf{x})\|\mathbf{P}_a(\mathbf{x})\mathbf{y}\|^2 \\
	\pi_{4a}(\mathbf{x},\mathbf{y},\mathsf{S}) &= \|\mathbf{J}_a^T(\mathbf{x})\mathbf{r}_a'(\mathbf{x}) - \mathbf{P}_a(\mathbf{x})\mathbf{y}\|^2
\end{align*}
where the damping functions $\lambda_a:X\to[0,\infty]$ for $a\in\overline{1,l}$ are arbitrary and $\mathbf{J}_a^*(\mathbf{x})$ is the damped pseudoinverse of $\mathbf{J}_a(\mathbf{x})$ with the damping constant $\lambda_a(\mathbf{x})$ at $\mathbf{x}$. 
A choice of the damping functions could be
\begin{equation}
\label{eqn:damping_function}
\lambda_a^2(\mathbf{x}) = \begin{dcases*} 
0, & $\mu_a = 0$ \\
\frac{\mu_a^2}{|\mathbf{M}_a(\mathbf{x})|^\nu}, & $\mu_a>0,\,|\mathbf{M}_a(\mathbf{x})|^\nu > 0$ \\
\infty, & $\mu_a>0,\,|\mathbf{M}_a(\mathbf{x})|^\nu = 0$
\end{dcases*}
\end{equation}
where $\mathbf{M}_a = \mathbf{C}_{aa}\mathbf{C}_{aa}^T$ and $\mu_1,\dots,\mu_l,\nu\in[0,\infty)$.
It is not difficult to show that all objective functions satisfy (O1) and (O2) and $\boldsymbol{\pi}_4$ satisfies additionally (O3) by following the procedure shown in \cite{An2019}.
If $\lambda_a(\mathbf{x})\in[0,\infty)$ for all $(a,\mathbf{x})$, then $\boldsymbol{\pi}_\alpha$ for $\alpha\in\overline{1,3}$ also satisfies (O3).
Therefore, $\boldsymbol{\pi}_\alpha$ is proper for $[\mathsf{S}]$.
Let $\mathbf{C}:X\to\mathbb{R}^{m\times m}$ be the left ($m\times m$) block of $\mathbf{C}_e$ and $\mathbf{C}_L = \mathbf{C} - \mathbf{C}_D$.
Define $\mathbf{D}_a,\mathbf{H}_a:X\to\mathbb{R}^{m_a\times m_a}$ as:
\begin{align*}
	\mathbf{D}_a(\mathbf{x}) &= \begin{dcases*} 
			(\mathbf{C}_{aa}\mathbf{C}_{aa}^T + \lambda_a^2\mathbf{I}_{m_a})^+(\mathbf{x}), & $\lambda_a^2(\mathbf{x}) \in [0,\infty)$ \\
			\mathbf{0}, & $\lambda_a^2(\mathbf{x}) = \infty$
		\end{dcases*} \\
	\mathbf{H}_a(\mathbf{x}) &= \begin{dcases*} 
			(\mathbf{J}_a\mathbf{J}_a^T + \lambda_a^2\mathbf{I}_{m_a})^+(\mathbf{x}), & $\lambda_a^2(\mathbf{x}) \in [0,\infty)$ \\
			\mathbf{0}, & $\lambda_a^2(\mathbf{x}) = \infty$
		\end{dcases*}
\end{align*}
and let $\mathbf{D} = \mathrm{diag}(\mathbf{D}_1,\dots,\mathbf{D}_l)$ and $\mathbf{H} = \mathrm{diag}(\mathbf{H}_1,\dots,\mathbf{H}_l)$.
One can easily check that the damped pseudoinverse of $\mathbf{C}_{aa}(\mathbf{x})$ with the damping constant $\lambda_a(\mathbf{x})$ can be written as $\mathbf{C}_{aa}^* = \mathbf{C}_{aa}^T\mathbf{D}_a$.
Let $\mathbf{C}_D^\circledast = \mathrm{diag}(\mathbf{C}_{11}^*,\dots,\mathbf{C}_{ll}^*)$.
Then, we can formulate the $\boldsymbol{\pi}_\alpha$-PIK solution of $[\mathsf{S}]$ in the form of \eqref{eqn:class_of_pik_solutions} with
\begin{equation*}
	\mathbf{L} = \begin{dcases*}
		\mathbf{D}(\mathbf{I}_m + \mathbf{C}_L\mathbf{C}_D^\circledast)^{-1}, & $\alpha = 1$ \\
		\mathbf{H}, & $\alpha = 2$ \\
		\mathbf{D}, & $\alpha = 3$ \\
		\mathbf{I}_m, & $\alpha = 4$.
	\end{dcases*}
\end{equation*}

\section{Nonsmoothness of PIK Solutions}
\label{sec:nonsmoothness_of_pik_solutions}

Orthogonalization plays an important role when we derive PIK solutions, so nonsmoothness in the orthogonalization process is directly related to nonsmoothness of PIK solutions. 
Primitive questions are in what condition orthogonalization becomes nonsmooth and when nonsmooth orthogonalization induces nonsmooth PIK solutions. 
We discuss about them by defining \textit{purely nonsmooth orthogonal projector} and \textit{smooth minimum basis subset}.

\begin{definition}[Purely Nonsmooth Orthogonal Projector]
	Let $\mathrm{OP} = \{ \mathbf{P}:X\to\mathbb{R}^{n\times n}\mid \mathbf{P}(\mathbf{x}) = \mathbf{P}^T(\mathbf{x}) = \mathbf{P}^2(\mathbf{x}),\,\forall\mathbf{x}\in X\}$. 
	We say $\mathbf{P}\in\mathrm{OP}$ is $\bullet$-discontinuous at $\mathbf{x}_0$ \textit{purely} or $\mathbf{P}\notin C_{\mathbf{x}_0}^\bullet$ \textit{purely} or $\mathbf{P}\in\mathrm{OP}\setminus C_{\mathbf{x}_0}^\bullet$ \textit{purely} if $\mathbf{P}\not\in C_{\mathbf{x}_0}^\bullet$ and it cannot be written as $\mathbf{P} = (\mathbf{P}) + [\mathbf{P}]$ for every $(\mathbf{P})\in\mathrm{OP}\cap C_{\mathbf{x}_0}^\bullet$ and $[\mathbf{P}]\in\mathrm{OP}\setminus C_{\mathbf{x}_0}^\bullet$ satisfying $(\mathbf{P})[\mathbf{P}] = \mathbf{0}$ and $(\mathbf{P})(\mathbf{x}_0) \neq \mathbf{0}$. 
\end{definition}

For $\mathbf{P}\in\mathrm{OP}$, we can consider a decomposition, which we call \textit{$\bullet$-discontinuity decomposition}, $\mathbf{P} = (\mathbf{P})_{\mathbf{x}_0}^\bullet + [\mathbf{P}]_{\mathbf{x}_0}^\bullet$ where $(\mathbf{P})_{\mathbf{x}_0}^\bullet \in \mathrm{OP}\cap C_{\mathbf{x}_0}^\bullet$; $[\mathbf{P}]_{\mathbf{x}_0}^\bullet = \mathbf{0}$ or $[\mathbf{P}]_{\mathbf{x}_0}^\bullet\in\mathrm{OP}\setminus C_{\mathbf{x}_0}^\bullet$ purely; and $(\mathbf{P})_{\mathbf{x}_0}^\bullet[\mathbf{P}]_{\mathbf{x}_0}^\bullet = \mathbf{0}$. 
The discontinuity decomposition always exists. 
If $\mathbf{P}\in C_{\mathbf{x}_0}^\bullet$, then we can choose $(\mathbf{P})_{\mathbf{x}_0}^\bullet = \mathbf{P}$ and $[\mathbf{P}]_{\mathbf{x}_0}^\bullet = \mathbf{0}$. 
If $\mathbf{P}\notin C_{\mathbf{x}_0}^\bullet$ purely, then we can choose $(\mathbf{P})_{\mathbf{x}_0}^\bullet = \mathbf{0}$ and $[\mathbf{P}]_{\mathbf{x}_0}^\bullet = \mathbf{P}$. 
If $\mathbf{P}\notin C_{\mathbf{x}_0}^\bullet$ not purely, then $\mathbf{P} = (\mathbf{P})_1 + [\mathbf{P}]_1 = (\mathbf{P})_1 + (\mathbf{P})_2 + [\mathbf{P}]_2 = \cdots = \sum_{i=1}^j(\mathbf{P})_i + [\mathbf{P}]_j$, $j\le n$, by definition until we find $[\mathbf{P}]_{\mathbf{x}_0}^\bullet = [\mathbf{P}]_a \notin C_{\mathbf{x}_0}^\bullet$ purely because $\mathbf{P}\notin C_{\mathbf{x}_0}^\bullet$, $\mathrm{rank}((\mathbf{P})_i(\mathbf{x}_0))\ge1$, and $\mathrm{rank}(\mathbf{P}(\mathbf{x}_0)) = \sum_{i=1}^j\mathrm{rank}((\mathbf{P})_i(\mathbf{x}_0)) + \mathrm{rank}([\mathbf{P}]_a(\mathbf{x}_0)) \le n$ by $(\mathbf{P})_i[\mathbf{P}]_i = \mathbf{0}$. 
An obvious property of the discontinuity dicomposition is that $(\mathbf{P})_{\mathbf{x}_0}^\bullet$ has a local constant rank at $\mathbf{x}_0$ because $\|(\mathbf{P})_{\mathbf{x}_0}^\bullet(\mathbf{x}) - (\mathbf{P})_{\mathbf{x}_0}^\bullet(\mathbf{x}_0)\|_F \ge |\|(\mathbf{P})_{\mathbf{x}_0}^\bullet(\mathbf{x})\|_F - \|(\mathbf{P})_{\mathbf{x}_0}^\bullet(\mathbf{x}_0)\|_F| = |(\mathrm{rank}((\mathbf{P})_{\mathbf{x}_0}^\bullet(\mathbf{x})))^{1/2} - (\mathrm{rank}((\mathbf{P})_{\mathbf{x}_0}^\bullet(\mathbf{x}_0)))^{1/2}|$.

Orthogonalization of rows of $\mathbf{J}\in C^\bullet(X,\mathbb{R}^{m\times n})$ with $m\le n$ by Lemma \ref{lem:orthogonalization_of_J} can be written elementwise as
\begin{equation*}
	\underbrace{\begin{bmatrix} \mathbf{j}_1 \\ \vdots \\ \mathbf{j}_m \end{bmatrix}}_{\mathbf{J}: X\to\mathbb{R}^{m\times n}} = \underbrace{\begin{bmatrix} c_{11} & \cdots & 0 & \cdots & 0 \\ \vdots & \ddots & \vdots & \ddots & \vdots \\ c_{m1} & \cdots & c_{mm} & \cdots & 0 \end{bmatrix}}_{\mathbf{C}_e = [c_{ij}]: X\to\mathbb{R}^{m\times n}}\underbrace{\begin{bmatrix} \hat{\mathbf{j}}_1 \\ \vdots \\ \hat{\mathbf{j}}_n\end{bmatrix}}_{\hat{\mathbf{J}}_e: X\to\mathbb{R}^{n\times n}}.
\end{equation*}
For each $\mathbf{x}\in X$, $\{\hat{\mathbf{j}}_1(\mathbf{x}),\dots,\hat{\mathbf{j}}_n(\mathbf{x})\}$ is an orthonormal basis of $\mathbb{R}^{1\times n}$ in which only $\mathrm{rank}(\mathbf{J}(\mathbf{x}))$ orthonormal vectors can be uniquely determined. 
Let $\mathcal{A} = \{\mathbf{j}_1,\dots,\mathbf{j}_m\}$ and $\mathcal{B} = \{\hat{\mathbf{j}}_1,\dots,\hat{\mathbf{j}}_n\}$. 
Note that $\mathcal{A}$ and $\mathcal{B}$ are sets of vector-valued functions, $|\mathcal{A}| \le m$, and $|\mathcal{B}| = n$. 
We define $\mathbf{J}_{a:b} = \begin{bmatrix} \mathbf{j}_a^T & \cdots & \mathbf{j}_b^T \end{bmatrix}^T$, $\mathcal{A}_{a:b} = \{\mathbf{j}_a,\dots,\mathbf{j}_b\}$, and $\mathbf{C}_{a:a',b:b'}$ as the block of $\mathbf{C}_e = [c_{ij}]$ with the top left entry $c_{ab}$ and the bottom right entry $c_{a'b'}$. 
$\hat{\mathbf{J}}_{a:b}$ and $\mathcal{B}_{a:b}$ are defined similarly. 
Define an orthogonal-projector-valued function $\mathbf{P}:2^\mathcal{B}\times X\to\mathbb{R}^{n\times n}$ and a set-valued function $\mathcal{F}:2^\mathcal{B}\times X\to2^\mathcal{B}$ as:
\begin{align*}
	\mathbf{P}(\mathcal{S},\mathbf{x}) &= \mathbf{P}_\mathcal{S}(\mathbf{x}) = \begin{dcases*} \mathbf{0}, & $\mathcal{S}=\emptyset$ \\ \sum_{\hat{\mathbf{j}}\in\mathcal{S}}\hat{\mathbf{j}}^T(\mathbf{x})\hat{\mathbf{j}}(\mathbf{x}), & $\mathcal{S}\neq\emptyset$ \end{dcases*} \\
	\mathcal{F}(\mathcal{S},\mathbf{x}) &= \mathcal{F}_\mathcal{S}(\mathbf{x}) = \begin{dcases*} \emptyset & $\mathcal{S} = \emptyset$ \\ \{\hat{\mathbf{j}}_a \in \mathcal{S}\cap\mathcal{B}_{1:m} \mid c_{aa}(\mathbf{x}) \neq 0 \} & $\mathcal{S} \neq \emptyset$ \end{dcases*}
\end{align*}
where $\mathbf{P}_\mathcal{B} = \mathbf{I}_n$. 
We can easily check that if $\mathcal{S},\mathcal{T}\subset\mathcal{B}$, then $\mathbf{P}_{\mathcal{S}\cap \mathcal{T}} = \mathbf{P}_\mathcal{S}\mathbf{P}_\mathcal{T} = \mathbf{P}_\mathcal{T}\mathbf{P}_\mathcal{S}$, $\mathbf{P}_{\mathcal{S}\cup \mathcal{T}} = \mathbf{P}_{\mathcal{S}\setminus \mathcal{T}} + \mathbf{P}_{\mathcal{T}\setminus \mathcal{S}} + \mathbf{P}_{\mathcal{S}\cap \mathcal{T}}$, and $\mathbf{P}_{\mathcal{S}\setminus\mathcal{T}} = \mathbf{P}_\mathcal{S} - \mathbf{P}_{\mathcal{S}\cap\mathcal{T}}$.
\begin{lemma}
	\label{lem:functions_property}
	For $a,b\in\overline{1,m}$ with $a\le b$, $\mathbf{P}_{\mathcal{F}(\mathcal{B}_{a:b})} = (\mathbf{C}_{a:b,a:b}\hat{\mathbf{J}}_{a:b})^+(\mathbf{C}_{a:b,a:b}\hat{\mathbf{J}}_{a:b})$.
\end{lemma}
\begin{proof}
	Let $\mathbf{x}\in X$. 
	If $\mathcal{F}(\mathcal{B}_{a:b},\mathbf{x}) = \emptyset$, then $\mathbf{C}_{a:b,a:b}(\mathbf{x}) = \mathbf{0}$ by Lemma \ref{lem:orthogonalization_of_J}, so $\mathbf{P}_{\mathcal{F}(\mathcal{B}_{a:b},\mathbf{x})}(\mathbf{x}) = ((\mathbf{C}_{a:b,a:b}\hat{\mathbf{J}}_{a:b})^+(\mathbf{C}_{a:b,a:b}\hat{\mathbf{J}}_{a:b}))(\mathbf{x}) = \mathbf{0}$. 
	Assume $\mathcal{F}(\mathcal{B}_{a:b},\mathbf{x}) \neq \emptyset$. 
	Let $\mathbf{C}_{a:b,a:b} = \begin{bmatrix} \mathbf{c}_a & \cdots & \mathbf{c}_b\end{bmatrix}$. 
	Then, $(\mathbf{C}_{a:b,a:b}\hat{\mathbf{J}}_{a:b})(\mathbf{x}) = \sum_{i\in\overline{a,b},\,c_{ii}(\mathbf{x})\neq0}(\mathbf{c}_i\hat{\mathbf{j}}_i)(\mathbf{x}) = (\underline{\mathbf{C}}_{a:b}\underline{\hat{\mathbf{J}}}_{a:b})(\mathbf{x})$ where $\underline{\mathbf{C}}_{a:b}(\mathbf{x})\in\mathbb{R}^{b-a+1\times|\mathcal{F}(\mathcal{B}_{a:b},\mathbf{x})|}$, $\mathrm{rank}(\underline{\mathbf{C}}_{a:b}(\mathbf{x})) = |\mathcal{F}(\mathcal{B}_{a:b},\mathbf{x})| \le b-a+1$, and $\underline{\hat{\mathbf{J}}}_{a:b}(\mathbf{x}) \in \mathbb{R}^{|\mathcal{F}(\mathcal{B}_{a:b},\mathbf{x})|\times n}$. 
	Thus, $\mathbf{P}_{\mathcal{F}(\mathcal{B}_{a:b},\mathbf{x})}(\mathbf{x}) = \sum_{i\in\overline{a,b},\,c_{ii}(\mathbf{x})\neq0}(\hat{\mathbf{j}}_i^T\hat{\mathbf{j}_i})(\mathbf{x}) = (\underline{\hat{\mathbf{J}}}_{a:b}^T\underline{\hat{\mathbf{J}}}_{a:b})(\mathbf{x}) = ((\underline{\mathbf{C}}_{a:b}\underline{\hat{\mathbf{J}}_{a:b}})^+(\underline{\mathbf{C}}_{a:b}\underline{\hat{\mathbf{J}}}_{a:b}))(\mathbf{x}) = ((\mathbf{C}_{a:b,a:b}\hat{\mathbf{J}}_{a:b})^+(\mathbf{C}_{a:b,a:b}\hat{\mathbf{J}}_{a:b}))(\mathbf{x})$.
\end{proof}

Let $a,a'\in\overline{1,m}$, $b,b'\in\overline{1,n}$, $a\le a'$, $b\le b'$, and $\mathbf{x}_0\in X$. 
We are looking for the condition of $\mathbf{C}_{a:a',b:b'}\hat{\mathbf{J}}_{b:b'} = \mathbf{J}_{a:a'}\mathbf{P}_{\mathcal{B}_{b:b'}} \not\in C_{\mathbf{x}_0}^\bullet$. 
By Lemma \ref{lem:derivative_of_matrix_and_its_entries}, $\mathbf{J}_{a:a'}\mathbf{P}_{\mathcal{B}_{b:b'}} \not\in C_{\mathbf{x}_0}^\bullet$ if and only if there exists $\mathbf{j} \in \mathcal{A}_{a:a'}$ such that $\mathbf{P}_{\mathcal{B}_{b:b'}}\mathbf{j}^T\not\in C_{\mathbf{x}_0}^\bullet$. 
Since $\mathbf{j}^T\in C_{\mathbf{x}_0}^\bullet$ for all $\mathbf{j}\in\mathcal{A}_{a:a'}$, $\mathbf{P}_{\mathcal{B}_{b:b'}}\mathbf{j}^T\not\in C_{\mathbf{x}_0}^\bullet$ only if $\mathbf{P}_{\mathcal{B}_{b:b'}}\not\in C_{\mathbf{x}_0}^\bullet$. 
The orthogonal projector $\mathbf{P}_{\mathcal{B}_{b:b'}}$ is generated from the basis subset $\mathcal{B}_{b:b'}\subset\mathcal{B}$. 
Thus, there is a relation between $\bullet$-discontinuity of orthogonal projectors and basis subsets. 
We first find the relation and then connect it to the condition of $\mathbf{P}_{\mathcal{B}_{b:b'}}\mathbf{j}^T\not\in C_{\mathbf{x}_0}^\bullet$ later.

\begin{definition}[Smooth Minimum Basis Subset]
	We say that $\mathcal{S} \subset \mathcal{B}$ is a $\bullet$-continuous basis subset ($\bullet$-CBS) at $\mathbf{x}_0$ when $\mathbf{0}\neq\mathbf{P}_\mathcal{S}\in C_{\mathbf{x}_0}^\bullet$. 
	We say that $\mathcal{S}$ is a $\bullet$-continuous minimum basis subset ($\bullet$-CMBS) at $\mathbf{x}_0$ when $\mathcal{S}$ is a $\bullet$-CBS at $\mathbf{x}_0$ and any proper subset of $\mathcal{S}$ is not a $\bullet$-CBS at $\mathbf{x}_0$. 
	We define $\mathcal{B}_{\mathbf{x}_0}^\bullet$ as the set of all $\bullet$-CMBSs at $\mathbf{x}_0$. 
	For $\mathcal{T}\subset\mathcal{S}\in\mathcal{B}_{\mathbf{x}_0}^\bullet$, we define $((\mathbf{P}_\mathcal{T}))_{\mathbf{x}_0}^\bullet = (\mathbf{P}_\mathcal{T})_{\mathbf{x}_0}^\bullet + (\mathbf{P}_{\mathcal{S}\setminus\mathcal{T}})_{\mathbf{x}_0}^\bullet$ and $([\mathbf{P}_\mathcal{T}])_{\mathbf{x}_0}^\bullet = [\mathbf{P}_\mathcal{T}]_{\mathbf{x}_0}^\bullet + [\mathbf{P}_{\mathcal{S}\setminus\mathcal{T}}]_{\mathbf{x}_0}^\bullet$. 
	Obviously, $((\mathbf{P}_\mathcal{T}))_{\mathbf{x}_0}^\bullet \in C_{\mathbf{x}_0}^\bullet$ and $([\mathbf{P}_\mathcal{T}])_{\mathbf{x}_0}^\bullet = \mathbf{P}_\mathcal{S} - ((\mathbf{P}_\mathcal{T}))_{\mathbf{x}_0}^\bullet \in C_{\mathbf{x}_0}^\bullet$.
\end{definition}

\begin{proposition}
	\label{prp:properties_of_MCBS}
	Let $\mathcal{T}\subset\mathcal{S}\in\mathcal{B}_{\mathbf{x}_0}^\bullet$ and $\mathbf{v}\in C_{\mathbf{x}_0}^\bullet(X,\mathbb{R}^n)$.
	\begin{enumerate}
		\item $\mathcal{B}_{\mathbf{x}_0}^\bullet$ is a partition of $\mathcal{B}$.
		\label{prp:MCBS:partition}
		\item $\mathbf{P}_\mathcal{T}\mathbf{v}\in C_{\mathbf{x}_0}^\bullet \iff [\mathbf{P}_\mathcal{T}]_{\mathbf{x}_0}^\bullet\mathbf{v} \in C_{\mathbf{x}_0}^\bullet \iff [\mathbf{P}_{\mathcal{S}\setminus\mathcal{T}}]_{\mathbf{x}_0}^\bullet\mathbf{v}\in C_{\mathbf{x}_0}^\bullet \iff \mathbf{P}_{\mathcal{S}\setminus\mathcal{T}}\mathbf{v}\in C_{\mathbf{x}_0}^\bullet$.
		\label{prp:MCBS:iff_conditions}
		\item $\mathbf{P}_\mathcal{T}\mathbf{v}\notin C_{\mathbf{x}_0}^\bullet,\,\mathcal{U} \subset \mathcal{B}\setminus\mathcal{S} \implies (\mathbf{P}_\mathcal{T}+\mathbf{P}_\mathcal{U})\mathbf{v}\notin C_{\mathbf{x}_0}^\bullet$.
		\label{prp:MCBS:nonsmoothness_condition}
		\item $\mathbf{P}_\mathcal{T}\mathbf{v}\in C_{\mathbf{x}_0}^\bullet$ only if $\mathcal{T}\in\{\emptyset,\mathcal{S}\}$ or both (1) $(([\mathbf{P}_\mathcal{T}])_{\mathbf{x}_0}^\bullet \mathbf{v})(\mathbf{x}_0) = \mathbf{0}$ and (2) if $\bullet\in\{1_p,1\}$, there exists $\mathbf{M}\in\mathbb{R}^{n\times n}$ such that
		\begin{equation*}
			\label{eq:DifferentiabilityCondition-modified}
			\lim_{\mathbf{x}\to \mathbf{x}_0}\frac{\|([\mathbf{P}_\mathcal{T}]_{\mathbf{x}_0}^\bullet(\mathbf{x})\mathbf{A} - \mathbf{M})(\mathbf{x}-\mathbf{x}_0)\|}{\|\mathbf{x}-\mathbf{x}_0\|} = 0
		\end{equation*}
		where $\mathbf{A} = \mathsf{D}(([\mathbf{P}_\mathcal{T}])_{\mathbf{x}_0}^\bullet\mathbf{v})(\mathbf{x}_0)$. 
		If such $\mathbf{M}$ exists, then $\mathbf{M} = \mathsf{D}([\mathbf{P}_\mathcal{T}]_{\mathbf{x}_0}^\bullet\mathbf{v})(\mathbf{x}_0)$. 
		The statement becomes ``if and only if" for $\bullet \in \{0,L_p,1_p\}$.
		\label{prp:MCBS:necessary_conditions}
		\item $\mathbf{J}_{1:a} \mathbf{P}_{\mathcal{S}\cap\mathcal{B}_{1:a}}\in C_{\mathbf{x}_0}^\bullet$ for all $a\in\overline{1,m}$.
		\label{prp:MCBS:smoothness_condition}
		\item Let $\mathcal{S} = \{\hat{\mathbf{j}}_{s_1},\dots,\hat{\mathbf{j}}_{s_{|\mathcal{S}|}}\}$ and $s_1\le \cdots \le s_{|\mathcal{S}|}$. 
		Then, $c_{s_1s_1}\hat{\mathbf{j}}_{s_1}\in C_{\mathbf{x}_0}^\bullet$ and if $c_{s_1s_1}(\mathbf{x}_0) \neq 0$, then $\mathcal{S} = \{\hat{\mathbf{j}}_{s_1}\}$.
		\label{prp:MCBS:first_element_of_S_is_smooth}
	\end{enumerate}
\end{proposition}
\begin{proof}
	\textit{\ref{prp:MCBS:partition}.} 
	Suppose that there exists $\mathcal{S}\neq\mathcal{S}'\in\mathcal{B}_{\mathbf{x}_0}^\bullet$ satisfying $\mathcal{U} = \mathcal{S}\cap\mathcal{S}'\neq\emptyset$. 
	Then, $\mathcal{U} \not\in \mathcal{B}_{\mathbf{x}_0}^\bullet$ and $C_{\mathbf{x}_0}^\bullet \ni \mathbf{P}_{\mathcal{S}}\mathbf{P}_{\mathcal{S}'} = \mathbf{P}_{\mathcal{S}\cap \mathcal{S}'} = \mathbf{P}_{\mathcal{U}}\notin C_{\mathbf{x}_0}^\bullet$, a contradiction. 
	Thus, $\mathcal{B}_{\mathbf{x}_0}^\bullet$ is mutually disjoint. 
	$\mathcal{C} = \bigcup_{\mathcal{S}'\in\mathcal{B}_{\mathbf{x}_0}^\bullet}\mathcal{S}' \subset \mathcal{B}$ because $\mathcal{S}' \subset \mathcal{B}$ if $\mathcal{S}' \in \mathcal{B}_{\mathbf{x}_0}^\bullet$. 
	Suppose $\mathcal{B} \not\subset\mathcal{C}$. 
	Then, $\emptyset \neq \mathcal{B}\setminus\mathcal{C} \notin\mathcal{B}_{\mathbf{x}_0}^\bullet$ and $C_{\mathbf{x}_0}^\bullet\not\ni\mathbf{P}_{\mathcal{B}\setminus \mathcal{C}} = \mathbf{P}_\mathcal{B} - \mathbf{P}_{\mathcal{B}\cap\mathcal{C}} = \mathbf{I}_n - \sum_{\mathcal{S}'\in\mathcal{B}_{\mathbf{x}_0}^\bullet}\mathbf{P}_{\mathcal{S}'} \in C_{\mathbf{x}_0}^\bullet$, a contradiction. 
	Thus, $\mathcal{B} = \bigcup_{\mathcal{S}'\in\mathcal{B}_{\mathbf{x}_0}^\bullet}\mathcal{S}'$.
	
	\textit{\ref{prp:MCBS:iff_conditions}.} 
	$\mathbf{P}_\mathcal{T}\mathbf{v}\in C_{\mathbf{x}_0}^\bullet \implies [\mathbf{P}_\mathcal{T}]_{\mathbf{x}_0}^\bullet\mathbf{v} = \mathbf{P}_\mathcal{T}\mathbf{v} - (\mathbf{P}_\mathcal{T})_{\mathbf{x}_0}^\bullet\mathbf{v} \in C_{\mathbf{x}_0}^\bullet \implies [\mathbf{P}_{\mathcal{S}\setminus\mathcal{T}}]_{\mathbf{x}_0}^\bullet\mathbf{v} = \mathbf{P}_\mathcal{S}\mathbf{v} - ((\mathbf{P}_\mathcal{T}))_{\mathbf{x}_0}^\bullet\mathbf{v} - [\mathbf{P}_\mathcal{T}]_{\mathbf{x}_0}^\bullet\mathbf{v} \in C_{\mathbf{x}_0}^\bullet \implies \mathbf{P}_{\mathcal{S}\setminus\mathcal{T}}\mathbf{v} = (\mathbf{P}_{\mathcal{S}\setminus\mathcal{T}})_{\mathbf{x}_0}^\bullet\mathbf{v} + [\mathbf{P}_{\mathcal{S}\setminus\mathcal{T}}]_{\mathbf{x}_0}^\bullet\mathbf{v} \in C_{\mathbf{x}_0}^\bullet \implies \mathbf{P}_\mathcal{T}\mathbf{v} = \mathbf{P}_\mathcal{S}\mathbf{v} - \mathbf{P}_{\mathcal{S}\setminus\mathcal{T}}\mathbf{v} \in C_{\mathbf{x}_0}^\bullet$. 
	
	\textit{\ref{prp:MCBS:nonsmoothness_condition}.} 
	If $(\mathbf{P}_\mathcal{T}+\mathbf{P}_\mathcal{U})\mathbf{v}\in C_{\mathbf{x}_0}^\bullet$, $\mathbf{P}_\mathcal{T}\mathbf{v} = \mathbf{P}_\mathcal{S}(\mathbf{P}_\mathcal{T}+\mathbf{P}_\mathcal{U})\mathbf{v} \in C_{\mathbf{x}_0}^\bullet$.
	
	\textit{\ref{prp:MCBS:necessary_conditions}.} 
	$\implies$: Suppose that $\mathcal{T}\notin\{\emptyset,\mathcal{S}\}$ and $(([\mathbf{P}_\mathcal{T}])_{\mathbf{x}_0}^\bullet\mathbf{v})(\mathbf{x}_0) \neq \mathbf{0}$. 
	Without loss of generality, $([\mathbf{P}_\mathcal{T}]_{\mathbf{x}_0}^\bullet\mathbf{v})(\mathbf{x}_0) \neq \mathbf{0}$. 
	Since $[\mathbf{P}_\mathcal{T}]_{\mathbf{x}_0}^\bullet\mathbf{v} \in C_{\mathbf{x}_0}^\bullet$, $\lim_{\mathbf{x}\to \mathbf{x}_0}\mathrm{rank}([\mathbf{P}_\mathcal{T}]_{\mathbf{x}_0}^\bullet\mathbf{v})(\mathbf{x}) = \mathrm{rank}([\mathbf{P}_\mathcal{T}]_{\mathbf{x}_0}^\bullet\mathbf{v})(\mathbf{x}_0) = 1$, so we can construct $(\mathbf{P}) = ([\mathbf{P}_\mathcal{T}]_{\mathbf{x}_0}^\bullet\mathbf{v})([\mathbf{P}_\mathcal{T}]_{\mathbf{x}_0}^\bullet\mathbf{v})^+ \in \mathrm{OP}\cap C_{\mathbf{x}_0}^\bullet$ with $(\mathbf{P})(\mathbf{x}_0) \neq \mathbf{0}$ by Lemma \ref{lem:ContinuityOfPseudoinverse-modified}. 
	Since $(\mathcal{V}) \subset [\mathcal{V}_\mathcal{T}]_{\mathbf{x}_0}^\bullet \notin C_{\mathbf{x}_0}^\bullet$, there must exist $[\mathcal{V}] = (\mathcal{V})^\perp\cap[\mathcal{V}_\mathcal{T}]_{\mathbf{x}_0}^\bullet \in\mathrm{VS}\setminus C_{\mathbf{x}_0}^\bullet$ such that $[\mathcal{V}_\mathcal{T}]_{\mathbf{x}_0}^\bullet = (\mathcal{V}) + [\mathcal{V}]$, $(\mathcal{V})\perp[\mathcal{V}]$, and $(\mathcal{V})(\mathbf{x}_0) \neq \{\mathbf{0}\}$, a contradiction that $[\mathcal{V}_\mathcal{T}]_{\mathbf{x}_0}^\bullet \notin C_{\mathbf{x}_0}^\bullet$ purely. 
	Let $\bullet \in \{1_p, 1\}$, $\mathcal{T}\notin\{\emptyset,\mathcal{S}\}$, and $(([\mathbf{P}_\mathcal{T}])_{\mathbf{x}_0}^\bullet\mathbf{v})(\mathbf{x}_0) = \mathbf{0}$. 
	Since $[\mathbf{P}_\mathcal{T}]_{\mathbf{x}_0}^\bullet\mathbf{v}\in C_{\mathbf{x}_0}^\bullet$, by letting $\mathbf{M} = \mathsf{D}([\mathbf{P}_\mathcal{T}]_{\mathbf{x}_0}^\bullet\mathbf{v})(\mathbf{x}_0)$ and using $(([\mathbf{P}_\mathcal{T}])_{\mathbf{x}_0}^\bullet\mathbf{v})(\mathbf{x}_0) = \mathbf{0}$, we find $\|([\mathbf{P}_\mathcal{T}]_{\mathbf{x}_0}^\bullet(\mathbf{x})\mathbf{A} - \mathbf{M})(\mathbf{x}-\mathbf{x}_0)\| \le \|(([\mathbf{P}_\mathcal{T}])_{\mathbf{x}_0}^\bullet\mathbf{v})(\mathbf{x}) - (([\mathbf{P}_\mathcal{T}])_{\mathbf{x}_0}^\bullet\mathbf{v})(\mathbf{x}_0) - \mathbf{A}(\mathbf{x}-\mathbf{x}_0)\| + \|([\mathbf{P}_\mathcal{T}]_{\mathbf{x}_0}^\bullet\mathbf{v})(\mathbf{x}) - ([\mathbf{P}_\mathcal{T}]_{\mathbf{x}_0}^\bullet\mathbf{v})(\mathbf{x}_0) - \mathbf{M}(\mathbf{x}-\mathbf{x}_0)\|$.

	$\impliedby$: Let $\bullet \in \{0,L_p,1_p\}$. 
	$\mathbf{P}_\mathcal{T}\mathbf{v}\in C_{\mathbf{x}_0}^\bullet$ if $\mathcal{T}\in\{\emptyset,\mathcal{S}\}$. 
	Let $\mathcal{T}\notin\{\emptyset,\mathcal{S}\}$. 
	Since $[\mathbf{P}_\mathcal{T}]_{\mathbf{x}_0}^\bullet[\mathbf{P}_{\mathcal{S}\setminus\mathcal{T}}]_{\mathbf{x}_0}^\bullet = \mathbf{0}$, $(([\mathbf{P}_\mathcal{T}])_{\mathbf{x}_0}^\bullet\mathbf{v})(\mathbf{x}_0) = \mathbf{0} \iff ([\mathbf{P}_\mathcal{T}]_{\mathbf{x}_0}^\bullet\mathbf{v})(\mathbf{x}_0) = ([\mathbf{P}_{\mathcal{S}\setminus\mathcal{T}}]_{\mathbf{x}_0}^\bullet\mathbf{v})(\mathbf{x}_0) = \mathbf{0}$. 
	So, $\|([\mathbf{P}_\mathcal{T}]_{\mathbf{x}_0}^\bullet\mathbf{v})(\mathbf{x}) - ([\mathbf{P}_\mathcal{T}]_{\mathbf{x}_0}^\bullet\mathbf{v})(\mathbf{x}_0)\| \le \sqrt{\|([\mathbf{P}_\mathcal{T}]_{\mathbf{x}_0}^\bullet\mathbf{v})(\mathbf{x})\|^2 + \|([\mathbf{P}_{\mathcal{S}\setminus\mathcal{T}}]_{\mathbf{x}_0}^\bullet\mathbf{v})(\mathbf{x})\|^2} = \|(([\mathbf{P}_\mathcal{T}])_{\mathbf{x}_0}^\bullet\mathbf{v})(\mathbf{x}) - (([\mathbf{P}_\mathcal{T}])_{\mathbf{x}_0}^\bullet\mathbf{v})(\mathbf{x}_0)\|$
	and $\|([\mathbf{P}_\mathcal{T}]_{\mathbf{x}_0}^{1_p}\mathbf{v})(\mathbf{x}) - ([\mathbf{P}_\mathcal{T}]_{\mathbf{x}_0}^{1_p}\mathbf{v})(\mathbf{x}_0) - \mathbf{M}(\mathbf{x} - \mathbf{x}_0)\| \le \|(([\mathbf{P}_\mathcal{T}])_{\mathbf{x}_0}^{1_p}\mathbf{v})(\mathbf{x}) - (([\mathbf{P}_\mathcal{T}])_{\mathbf{x}_0}^{1_p}\mathbf{v})(\mathbf{x}_0) - \mathbf{A}(\mathbf{x}-\mathbf{x}_0)\| + \|([\mathbf{P}_\mathcal{T}]_{\mathbf{x}_0}^{1_p}(\mathbf{x})\mathbf{A} - \mathbf{M})(\mathbf{x} - \mathbf{x}_0)\|$.
	It completes the proof.
	
	\textit{\ref{prp:MCBS:smoothness_condition}.} 
	By Lemma \ref{lem:derivative_of_matrix_and_its_entries}, $\mathbf{J}_{1:a}\in C_{\mathbf{x}_0}^\bullet$. 
	Then, $\mathbf{J}_{1:a} \mathbf{P}_{\mathcal{S}\cap\mathcal{B}_{1:a}} = \mathbf{J}_{1:a} (\mathbf{P}_{\mathcal{S}\cap\mathcal{B}_{1:a}} + \mathbf{P}_\emptyset) = \mathbf{J}_{1:a}(\mathbf{P}_{\mathcal{S}\cap\mathcal{B}_{1:a}} + \mathbf{P}_{\mathcal{B}_{1:a}}\mathbf{P}_{\mathcal{S}\setminus\mathcal{B}_{1:a}}) = \mathbf{J}_{1:a}(\mathbf{P}_{\mathcal{S}\cap\mathcal{B}_{1:a}} + \mathbf{P}_{\mathcal{S}\setminus\mathcal{B}_{1:a}}) = \mathbf{J}_{1:a}\mathbf{P}_\mathcal{S} \in C_{\mathbf{x}_0}^\bullet$.
	
	\textit{\ref{prp:MCBS:first_element_of_S_is_smooth}.} 
	By the property \ref{prp:MCBS:smoothness_condition}, $\mathbf{J}_{1:s_1}\mathbf{P}_{\mathcal{S}\cap\mathcal{B}_{1:s_1}} \in C_{\mathbf{x}_0}^\bullet$. 
	Since $\mathcal{S}\cap\mathcal{B}_{1:s_1} = \{\hat{\mathbf{j}}_{s_1}\}$, we have $c_{s_1s_1}\hat{\mathbf{j}}_{s_1} = \mathbf{j}_{s_1}\mathbf{P}_{\mathcal{S}\cap\mathcal{B}_{1:s_1}} \in C_{\mathbf{x}_0}^\bullet$ by Lemma \ref{lem:derivative_of_matrix_and_its_entries}. 
	If $c_{s_1s_1}(\mathbf{x}_0) \neq 0$, then $\mathbf{P}_{\{\hat{\mathbf{j}}_{s_1}\}} = (c_{s_1s_1}\hat{\mathbf{j}}_{s_1})^+(c_{s_1s_1}\hat{\mathbf{j}}_{s_1}) \in C_{\mathbf{x}_0}^\bullet$ by Lemma \ref{lem:ContinuityOfPseudoinverse-modified}, so $\mathcal{S} = \{\hat{\mathbf{j}}_{s_1}\}$.
\end{proof}

We are ready to find the condition of $\mathbf{C}_{a:a',b:b'}\hat{\mathbf{J}}_{b:b'} \not\in C_{\mathbf{x}_0}^\bullet$.

\begin{theorem}
	\label{thm:discontinuity-of-QRD-modified}
	Let $a,a'\in\overline{1,m}$, $b,b'\in\overline{1,n}$, $a\le a'$, $b\le b'$, and $\mathbf{J}\in C_{\mathbf{x}_0}^\bullet$. 
	Then, $\mathbf{C}_{a:a',b:b'}\hat{\mathbf{J}}_{b:b'} \notin C_{\mathbf{x}_0}^\bullet$ if and only if there exist $\mathbf{j}\in\mathcal{A}_{a:a'}$ and $\mathcal{S}\in\mathcal{B}_{\mathbf{x}_0}^\bullet$ such that $[\mathbf{P}_{\mathcal{S}\cap\mathcal{B}_{b:b'}}]_{\mathbf{x}_0}^\bullet\mathbf{j}^T \notin C_{\mathbf{x}_0}^\bullet$. 
	Also, $\mathbf{C}_{a:a',b:b'}\hat{\mathbf{J}}_{b:b'} \notin C_{\mathbf{x}_0}^\bullet$ if there exist $\mathbf{j}\in\mathcal{A}_{a:a'}$ and $\mathcal{S}\in\mathcal{B}_{\mathbf{x}_0}^\bullet$ such that $\mathcal{S}\cap\mathcal{B}_{b:b'}\not\in\{\emptyset,\mathcal{S}\}$ and either (1) $(([\mathbf{P}_{\mathcal{S}\cap\mathcal{B}_{b:b'}}])_{\mathbf{x}_0}^\bullet\mathbf{j}^T)(\mathbf{x}_0) \neq\mathbf{0}$ or (2) if $\bullet\in\{1_p,1\}$, for every $\mathbf{M}\in\mathbb{R}^{n\times n}$
	\begin{equation*}
		\limsup_{\mathbf{x}\to \mathbf{x}_0}\frac{\|([\mathbf{P}_{\mathcal{S}\cap\mathcal{B}_{b:b'}}]_{\mathbf{x}_0}^\bullet(\mathbf{x})\mathbf{A} - \mathbf{M})(\mathbf{x}-\mathbf{x}_0)\|}{\|\mathbf{x}-\mathbf{x}_0\|} > 0
		\label{eq:discontinuity-nonLipschitz-nondifferentiability-of-QRD-modified}
	\end{equation*}
	where $\mathbf{A} = \mathsf{D}(([\mathbf{P}_{\mathcal{S}\cap\mathcal{B}_{b:b'}}])_{\mathbf{x}_0}^\bullet\mathbf{j}^T)(\mathbf{x}_0)$. 
	The condition becomes if and only if for $\bullet \in \{0,L_p,1_p\}$.
\end{theorem}
\begin{proof}
	By Lemma \ref{lem:derivative_of_matrix_and_its_entries} and Proposition \ref{prp:properties_of_MCBS},
	\begin{align*}
		&\mathbf{C}_{a:a',b:b'}\hat{\mathbf{J}}_{b:b'} = \mathbf{J}_{a:a'}\mathbf{P}_{\mathcal{B}_{b:b'}}\notin C_{\mathbf{x}_0}^\bullet \\
		&\iff \exists\mathbf{j}\in\mathcal{A}_{a:a'}:\mathbf{P}_{\mathcal{B}_{b:b'}}\mathbf{j}^T = \sum_{\mathcal{S}\in\mathcal{B}_{\mathbf{x}_0}^\bullet}\mathbf{P}_{\mathcal{S}\cap\mathcal{B}_{b:b'}}\mathbf{j}^T \notin C_{\mathbf{x}_0}^\bullet \\
		&\iff \exists \mathbf{j}\in\mathcal{A}_{a:a'},\,\mathcal{S}\in\mathcal{B}_{\mathbf{x}_0}^\bullet:\mathbf{P}_{\mathcal{S}\cap\mathcal{B}_{b:b'}}\mathbf{j}^T\notin C_{\mathbf{x}_0}^\bullet \\
		&\iff \exists \mathbf{j}\in\mathcal{A}_{a:a'},\,\mathcal{S}\in\mathcal{B}_{\mathbf{x}_0}^\bullet:[\mathbf{P}_{\mathcal{S}\cap\mathcal{B}_{b:b'}}]_{\mathbf{x}_0}^\bullet\mathbf{j}^T\notin C_{\mathbf{x}_0}^\bullet.
	\end{align*}
	The rest is the direct consequence of Proposition \ref{prp:properties_of_MCBS}.\ref{prp:MCBS:necessary_conditions}.
\end{proof}

Let $\mathbf{u}$ be a PIK solution of $[\mathsf{S}]^\bullet\subset\mathbb{S}^\bullet$ in the form of \eqref{eqn:class_of_pik_solutions} and fix $\mathsf{S}\in[\mathsf{S}]^\bullet$.
Then, $\mathbf{f}_t,\mathbf{F}_q,\mathbf{r}'\in C^\bullet$ by Lemma \ref{lem:derivative_of_matrix_and_its_entries}; $\mathbf{R}^{-1}, \mathbf{J} \in C^\bullet$ by Lemma \ref{lem:ContinuityOfPseudoinverse-modified}; and $\mathbf{u}(\cdot,\mathsf{S})\in C^\bullet$ if and only if $\mathbf{R}\mathbf{u}(\cdot,\mathsf{S}) = \hat{\mathbf{J}}^T\mathbf{C}_D^T\mathbf{L}\mathbf{r}'\in C^\bullet$ by Lemma \ref{lem:Af_continuous_iff_f_continuous}. 
Theorem \ref{thm:discontinuity-of-QRD-modified} discovers that $\mathbf{u}$ contains a source of $\bullet$-discontinuity that is $\hat{\mathbf{J}}^T\mathbf{C}_D^T$.
Define set-valued maps $\mathcal{G},\mathcal{G}^\bullet:\mathbb{S}^\bullet\to\bigcup_{i\in\mathbb{N}}2^{\mathbb{R}^i}$ as:
\begin{align*}
	\mathcal{G}(\mathsf{S}) &= \mathcal{G}_\mathsf{S} = \{\mathbf{x}\in X\mid \det(\mathbf{C}(\mathbf{x})) \neq 0\} \\
	\mathcal{G}^\bullet(\mathsf{S}) &= \mathcal{G}_\mathsf{S}^\bullet = \{\mathbf{x}\in X\mid \mathbf{C}_{aa}\hat{\mathbf{J}}_a \in C_{\mathbf{x}}^\bullet,\,a\in\overline{1,l}\}.
\end{align*}
Obviously, $\mathcal{G}([\mathsf{S}]^\bullet) = \mathcal{G}(\mathsf{S})$ and $\mathcal{G}^\bullet([\mathsf{S}]^\bullet) = \mathcal{G}^\bullet(\mathsf{S})$ for all $\mathsf{S}\in[\mathsf{S}]^\bullet$. 
$\mathcal{G}_\mathsf{S} = \mathrm{int}(\mathcal{G}_\mathsf{S}) \subset \mathcal{G}_\mathsf{S}^\bullet$ and $X\setminus\mathcal{G}_\mathsf{S}^\bullet \subset X\setminus\mathcal{G}_\mathsf{S} = \mathrm{cl}(X\setminus\mathcal{G}_\mathsf{S})$ because if $\mathbf{x}\in\mathcal{G}_\mathsf{S}$, then $c_{aa}(\mathbf{x}) > 0$, $c_{aa}\hat{\mathbf{j}}_a\in C_\mathbf{x}^\bullet$, and $\{\hat{\mathbf{j}}_a\}\in\mathcal{B}_\mathbf{x}^\bullet$ for all $a\in\overline{1,m}$ by Lemma \ref{lem:orthogonalization_of_J} and Proposition \ref{prp:properties_of_MCBS}.\ref{prp:MCBS:first_element_of_S_is_smooth}, so $c_{aa} = \langle c_{aa}\hat{\mathbf{j}}_a,c_{aa}\hat{\mathbf{j}}_a\rangle^{1/2}\in C_\mathbf{x}^\bullet$ and $c_{ab}\hat{\mathbf{j}}_b = \mathbf{j}_a\mathbf{P}({\{\hat{\mathbf{j}}_b\}})\in C_\mathbf{x}^\bullet$ for all $a,b\in\overline{1,m}$. 
It is obvious that $\bullet$-continuity of $\mathbf{u}(\cdot,\mathsf{S})$ cannot be guaranteed at $\mathbf{x}\in X\setminus\mathcal{G}_\mathsf{S}^\bullet$ in general.
If $\Omega = \{\mathbf{x}\in X\mid\mathbf{L}\in C_\mathbf{x}^\bullet\}$, then we can only guarantee $\mathbf{u}(\cdot,\mathsf{S})\in C_{\Omega\cap\mathcal{G}^\bullet(\mathsf{S})}^\bullet$.

Let $\mathbf{u}_\alpha$ for $\alpha\in\overline{1,4}$ be the $\boldsymbol{\pi}_\alpha$-PIK solution of $[\mathsf{S}]^\bullet$ with the damping functions given by \eqref{eqn:damping_function}.
Obviously, $\mathbf{u}_4 \in C_{\mathcal{G}^\bullet(\mathsf{S})}^\bullet$.
Let $\alpha\in\overline{1,3}$.
Let $\mu_1 = \cdots = \mu_l = 0$ and $\mathbf{x}_0\in\mathcal{G}_\mathsf{S}$.
Since $\mathbf{C}_{aa}\hat{\mathbf{J}}_a\in C_{\mathbf{x}_0}^\bullet$ and $\lim_{\mathbf{x}\to \mathbf{x}_0}\mathrm{rank}((\mathbf{C}_{aa}\hat{\mathbf{J}}_a)(\mathbf{x})) = \mathrm{rank}((\mathbf{C}_{aa}\hat{\mathbf{J}}_a)(\mathbf{x}_0)) = m_a$ by $\mathcal{G}_\mathsf{S} = \mathrm{int}(\mathcal{G}_\mathsf{S}) \subset \mathcal{G}_\mathsf{S}^\bullet$, we have $\hat{\mathbf{J}}_a^T\mathbf{C}_{aa}^* = (\mathbf{C}_{aa}\hat{\mathbf{J}}_a)^+\in C_{\mathbf{x}_0}^\bullet$ and $\mathbf{D}_a = (\mathbf{C}_{aa}\hat{\mathbf{J}}_a\hat{\mathbf{J}}_a^T\mathbf{C}_{aa}^T)^+\in C_{\mathbf{x}_0}^\bullet$ by Lemma \ref{lem:ContinuityOfPseudoinverse-modified}. 
Similarly, $\mathbf{H}_a\in C_{\mathbf{x}_0}^\bullet$.
It follows that $\mathbf{D},\mathbf{H},\hat{\mathbf{J}}^T\mathbf{C}_D^\circledast\in C_{\mathbf{x}_0}^\bullet$ and $\mathbf{C}_L\mathbf{C}_D^\circledast = (\mathbf{C}_L\hat{\mathbf{J}})(\hat{\mathbf{J}}^T\mathbf{C}_D^\circledast) = (\mathbf{J} - \mathbf{C}_D\hat{\mathbf{J}})(\hat{\mathbf{J}}^T\mathbf{C}_D^\circledast)\in C_{\mathbf{x}_0}^\bullet$.
Therefore, $\mathbf{u}_\alpha(\cdot,\mathsf{S}) \in C_{\mathbf{x}_0}^\bullet$.
If $\mathbf{x}_0\in\mathrm{bd}(\mathcal{G}_\mathsf{S})\cap\mathcal{G}_\mathsf{S}^\bullet$, then there exist $a\in\overline{1,l}$ and a sequence $\mathbf{x}_i\to\mathbf{x}_0$ such that $\lim_{i\to\infty}\|(\mathbf{C}_{aa}\mathbf{C}_{aa})^+(\mathbf{x}_i)\| = \infty$ or $\lim_{i\to\infty}\|(\mathbf{J}_a\mathbf{J}_a^T)^+(\mathbf{x}_i)\| = \infty$, so a regularization should be introduced to the $\boldsymbol{\pi}_\alpha$-PIK solution of $[\mathsf{S}]^\bullet$.
Let $\mu_1,\dots,\mu_l\in(0,\infty)$, $\nu\in\mathbb{N}\cup\{0\}$, and $\mathbf{x}_0\in\mathcal{G}_\mathsf{S}^\bullet$.
$\hat{\mathbf{J}}_a^T\mathbf{C}_{aa}^* = (\mathbf{C}_{aa}\hat{\mathbf{J}}_a)^*\in C_{\mathbf{x}_0}^\bullet$ by Lemma \ref{lem:smoothness_of_extended_damped_pseudoinverse}.
We can show $\mathbf{D}_a,\mathbf{H}_a\in C_{\mathbf{x}_0}^\bullet$ similarly as in the proof of Lemma \ref{lem:smoothness_of_extended_damped_pseudoinverse}.
It follows that $\mathbf{D},\mathbf{H},\mathbf{C}_L\mathbf{C}_D^\circledast\in C_{\mathbf{x}_0}^\bullet$.
Therefore, $\mathbf{u}_\alpha(\cdot,\mathsf{S})\in C_{\mathbf{x}_0}^\bullet$.
$\bullet$-continuity of $\mathbf{u}_\alpha(\cdot,\mathsf{S})$ on $X\setminus\mathcal{G}_\mathsf{S}^\bullet$ is not guaranteed.

Now, we have a canonical question: Can we always find a $\bullet$-continuous PIK solution of $[\mathsf{S}]^\bullet$?
It depends on the level of priority required.
For example, if $\mathbf{J}\neq\mathbf{0}$ and $\pi_a(\mathbf{x},\mathbf{y},\mathsf{S}) = \|\mathbf{P}_a(\mathbf{x})\mathbf{y}\|$ for $a\in\overline{1,l}$, then the $\boldsymbol{\pi}$-WPIK solution of $[\mathsf{S}]^\bullet$ becomes $\mathbf{u}(\cdot,\mathsf{S}) = \mathbf{0} \in C^\bullet$ for all $\mathsf{S}\in[\mathsf{S}]^\bullet$.
So, a more interesting question is: Can we always find a $\bullet$-continuous SPIK solution of $[\mathsf{S}]^\bullet$?
We prove by contradiction that there does not always exist a $\bullet$-continuous SPIK solution of $[\mathsf{S}]^\bullet$. 
Define $\mathcal{B}_a = \{\text{rows of $\hat{\mathbf{J}}_a$}\}$ for $a\in\overline{1,l}$ and assume that there exist $\mathbf{x}_0\in X$ and $\mathcal{S}\in\mathcal{B}_{\mathbf{x}_0}^\bullet$ satisfying $\mathcal{R}([\mathbf{P}_{\mathcal{S}\setminus\bigcup_{a=1}^{l-1}\mathcal{B}_b}]_{\mathbf{x}_0}^\bullet(\mathbf{x}_0)) \not\perp \mathcal{R}(\mathbf{P}_l(\mathbf{x}_0))$. 
Let $\mathcal{T} = \mathcal{S}\cap\bigcup_{a=1}^{l-1}\mathcal{B}_a$. 
Then, $\mathcal{S}\setminus\bigcup_{a=1}^{l-1}\mathcal{B}_a = \mathcal{S}\setminus\mathcal{T}$. 
Suppose that there exists a $\bullet$-continuous SPIK solution $\mathbf{u}$ of $[\mathsf{S}]^\bullet$. 
By Lemma \ref{lem:Af_continuous_iff_f_continuous}, $\mathbf{v}(\cdot,\mathsf{S}) = \mathbf{R}\mathbf{u}(\cdot,\mathsf{S})\in C^\bullet$ for all $\mathsf{S}\in[\mathsf{S}]^\bullet$. 
There exists a strongly proper objective function $\boldsymbol{\pi}$ for $[\mathsf{S}]^\bullet$ satisfying
\begin{align*}
	\mathbf{v}(\mathbf{x},\mathsf{S}) &= \mathbf{v}_1(\mathbf{x},\mathsf{S}) + \cdots + \mathbf{v}_l(\mathbf{x},\mathsf{S}) \\
	&= \arglexmin_{\mathbf{y}\in\mathbb{R}^n} (\pi_1(\mathbf{x},\mathbf{y},\mathsf{S}), \dots, \pi_l(\mathbf{x},\mathbf{y},\mathsf{S}), \|\mathbf{y}\|^2/2) 
\end{align*}
for every $(\mathbf{x},\mathsf{S})$ where $\mathbf{v}_a(\mathbf{x},\mathsf{S}) \in \mathcal{R}(\mathbf{P}_a(\mathbf{x}))$ for $a\in\overline{1,l}$. 
Since only the reference is different for each $\mathsf{S}\in[\mathsf{S}]^\bullet$, we can write $\mathbf{v}(\cdot,\mathsf{S}) = \mathbf{v}(\cdot,\mathbf{r})$ for $\mathbf{r}\in C^\bullet$. 
Then, $([\mathbf{P}_\mathcal{T}])_{\mathbf{x}_0}^\bullet\mathbf{v}(\cdot,\mathbf{r}) = [\mathbf{P}_\mathcal{T}]_{\mathbf{x}_0}^\bullet\mathbf{v}_{1:l-1}(\cdot,\mathbf{r}) + [\mathbf{P}_{\mathcal{S}\setminus\mathcal{T}}]_{\mathbf{x}_0}^\bullet\mathbf{v}_l(\cdot,\mathbf{r})\in C_{\mathbf{x}_0}^\bullet$ for all $\mathbf{r}\in C^\bullet$ where $\mathbf{v}_{1:l-1} = \mathbf{v}_1 + \cdots + \mathbf{v}_{l-1}$. 
Fix $\mathbf{r} = (\mathbf{r}_1,\dots,\mathbf{r}_l)\in C^\bullet$ and let $\tilde{\mathbf{r}} = (\mathbf{r}_1,\dots,\mathbf{r}_{l-1},\tilde{\mathbf{r}}_l)\in C^\bullet$. 
By (O1), $\mathbf{v}_{1:l-1}(\cdot,\tilde{\mathbf{r}}) = \mathbf{v}_{1:l-1}(\cdot,\mathbf{r})$ for all $\tilde{\mathbf{r}}_l\in C^\bullet$, so we have $([\mathbf{P}_\mathcal{T}])_{\mathbf{x}_0}^\bullet(\mathbf{v}(\cdot,\mathbf{r}) - \mathbf{v}(\cdot,\tilde{\mathbf{r}})) = [\mathbf{P}_{\mathcal{S}\setminus\mathcal{T}}]_{\mathbf{x}_0}^\bullet(\mathbf{v}(\cdot,\mathbf{r}) - \mathbf{v}(\cdot,\tilde{\mathbf{r}})) \in C_{\mathbf{x}_0}^\bullet$ for all $\tilde{\mathbf{r}}_l\in C^\bullet$. 
By Proposition \ref{prp:properties_of_MCBS}.\ref{prp:MCBS:necessary_conditions}, $[\mathbf{P}_{\mathcal{S}\setminus\mathcal{T}}]_{\mathbf{x}_0}^\bullet(\mathbf{x}_0)(\mathbf{v}(\mathbf{x}_0,\mathbf{r}) - \mathbf{v}(\mathbf{x}_0,\tilde{\mathbf{r}})) = [\mathbf{P}_{\mathcal{S}\setminus\mathcal{T}}]_{\mathbf{x}_0}^\bullet(\mathbf{x}_0)(\mathbf{v}_l(\mathbf{x}_0,\mathbf{r}) - \mathbf{v}_l(\mathbf{x}_0,\tilde{\mathbf{r}})) = \mathbf{0}$ for all $\tilde{\mathbf{r}}_l(\mathbf{x}_0)\in\mathbb{R}^{m_l}$. 
Since the condition $\mathcal{R}([\mathbf{P}_{\mathcal{S}\setminus\mathcal{T}}]_{\mathbf{x}_0}^\bullet(\mathbf{x}_0)) \not\perp \mathcal{R}(\mathbf{P}_l(\mathbf{x}_0))$ implies $\mathcal{R}(\mathbf{P}_l(\mathbf{x}_0))\not\subset\mathcal{N}([\mathbf{P}_{\mathcal{S}\setminus\mathcal{T}}]_{\mathbf{x}_0}^\bullet(\mathbf{x}_0))$, we find $\mathbf{v}_l(\mathbf{x}_0,\mathbf{r}) - \mathbf{v}_l(\mathbf{x}_0,\tilde{\mathbf{r}}) \in \mathcal{N}([\mathbf{P}_{\mathcal{S}\setminus\mathcal{T}}]_{\mathbf{x}_0}^\bullet(\mathbf{x}_0)) \cap \mathcal{R}(\mathbf{P}_l(\mathbf{x}_0)) \neq \mathcal{R}(\mathbf{P}_l(\mathbf{x}_0))$ for all $\tilde{\mathbf{r}}_l(\mathbf{x}_0)\in\mathbb{R}^{m_l}$, a contradiction of (O3) that the mapping $\tilde{\mathbf{r}}_l(\mathbf{x}_0)\mapsto\mathbf{v}_l(\mathbf{x}_0,\tilde{\mathbf{r}})$ of $\mathcal{R}((\mathbf{C}_{ll}\hat{\mathbf{J}}_l)(\mathbf{x}_0))$ into $\mathcal{R}(\mathbf{P}_l(\mathbf{x}_0))$ is one-to-one and onto. 
We summarize this discussion in Theorem \ref{thm:smoothness_of_the_PIK_solution}.

\begin{theorem}
	\label{thm:smoothness_of_the_PIK_solution}
	Let $[\mathsf{S}]^\bullet\subset\mathbb{S}^\bullet$.
	A PIK solution of $[\mathsf{S}]^\bullet$ in the form of \eqref{eqn:class_of_pik_solutions} is $\bullet$-continuous on $\mathcal{G}_\mathsf{S}^\bullet\cap\{\mathbf{x}\in X\mid \mathbf{L}\in C_\mathbf{x}^\bullet\}$.
	The $\boldsymbol{\pi}_\alpha$-PIK solution of $[\mathsf{S}]^\bullet$ for $\alpha\in\overline{1,3}$ with the damping functions given by \eqref{eqn:damping_function} is $\bullet$-continuous on $\mathcal{G}_\mathsf{S}$ if $\mu_1 = \cdots = \mu_l = 0$ and $\bullet$-continuous on $\mathcal{G}_\mathsf{S}^\bullet$ if $\mu_1,\dots,\mu_l\in(0,\infty)$ and $\nu\in\mathbb{N}\cup\{0\}$. 
	The $\boldsymbol{\pi}_4$-PIK solution of $[\mathsf{S}]^\bullet$ is $\bullet$-continuous on $\mathcal{G}_\mathsf{S}^\bullet$.
	Define $\mathcal{B}_a = \{\text{rows of $\hat{\mathbf{J}}_a$}\}$ for $a\in\overline{1,l}$. 
	If there exist $\mathbf{x}_0\in X$ and $\mathcal{S}\in\mathcal{B}_{\mathbf{x}_0}^\bullet$ satisfying 
	\begin{equation*}
		\mathcal{R}([\mathbf{P}_{\mathcal{S}\setminus\bigcup_{a=1}^{l-1}\mathcal{B}_b}]_{\mathbf{x}_0}^\bullet(\mathbf{x}_0)) \not\perp \mathcal{R}(\mathbf{P}_{\mathcal{F}(\mathcal{B}_l,\mathbf{x}_0)}(\mathbf{x}_0)),
	\end{equation*}
	then there does not exist a $\bullet$-continuous SPIK solution of $[\mathsf{S}]^\bullet$.
\end{theorem}

\section{Trajectory Existence}
\label{sec:trajectory_existence}

Once we find a PIK solution $\mathbf{u}$ of $[\mathsf{S}]\subset\mathbb{S}$ in the form of \eqref{eqn:class_of_pik_solutions}, we generate a joint trajectory $\mathbf{q}:[t_0,\infty)\to\mathbb{R}^n$ of a kinematic system $\mathsf{S}\in[\mathsf{S}]$ by solving the differential equation
\begin{equation}
	\label{eqn:differential_equation}
	\dot{\mathbf{q}} = \mathbf{u}(t,\mathbf{q},\mathsf{S})
\end{equation}
with an initial value $\mathbf{q}(t_0) = \mathbf{q}_0$.
Then, the joint trajectory is utilized to operate a mechanism in an environment according to a scenario.
Therefore, existence of a joint trajectory satisfying \eqref{eqn:differential_equation} is an important property we must check.
The classical solution of this initial value problem is a function $\mathbf{q}:[t_0,\infty)\to\mathbb{R}^n$ that is differentiable on $(t_0,\infty)$ and satisfies $\mathbf{q}(t_0) = \mathbf{q}_0$ and $\dot{\mathbf{q}}(t) = \mathbf{u}(t,\mathbf{q}(t),\mathsf{S})$ for all $t\in(t_0,\infty)$. 
The classical solution exists if $\mathbf{u}(\cdot,\mathsf{S})$ is continuous and linearly bounded on $X$ and the classical solution is unique if additionally $\mathbf{u}(\cdot,\mathsf{S})$ is locally Lipschitz on $X$ where $\mathbf{u}(\cdot,\mathsf{S})$ is said to be linearly bounded on $X$ if there exist $\gamma,c\in[0,\infty)$ such that $\|\mathbf{u}(t,\mathbf{q},\mathsf{S})\|\le\gamma\|\mathbf{q}\| + c$ for all $(t,\mathbf{q})\in X$ \cite[pp.178]{Clarke2008}. 
Unfortunately, even if we assume $\mathsf{S}\in\mathbb{S}^0$, continuity of $\mathbf{u}(\cdot,\mathsf{S})$ is not guaranteed in general as we can see in Appendix.
One way to resolve this existence problem is to extend the definition of the classical solution. 

The Carath\'{e}odory solution is a function $\mathbf{q}:[t_0,\infty)\to\mathbb{R}^n$ that is absolutely continuous on $[t_0,\infty)$ and satisfies $\mathbf{q}(t_0) = \mathbf{q}_0$ and $\dot{\mathbf{q}}(t) = \mathbf{u}(t,\mathbf{q}(t),\mathsf{S})$ for almost all $t\in(t_0,\infty)$. 
The Carath\'{e}odory solution exists if $\mathbf{u}(t,\cdot,\mathsf{S})$ is continuous on $\mathbb{R}^n$ for almost all $t\in[t_0,\infty)$; $\mathbf{u}(\cdot,\mathbf{q},\mathsf{S})$ is measurable in $\mathbb{R}$ for each $\mathbf{q}\in\mathbb{R}^n$; there exists a function $m(t)$ integrable for each finite interval of $[t_0,\infty)$ such that $\|\mathbf{u}(t,\mathbf{q},\mathsf{S})\|\le m(t)$; and $\mathbf{u}(\cdot,\mathsf{S})$ is linearly bounded on $X$ and the Carath\'{e}odory solution is unique if additionally for each compact set $A\subset X$, there exists an integrable function $l(t)$ such that $\|\mathbf{u}(t,\mathbf{q}_1,\mathsf{S}) - \mathbf{u}(t,\mathbf{q}_2,\mathsf{S})\| \le l(t)\|\mathbf{q}_1-\mathbf{q}_2\|$ for every $(t,\mathbf{q}_1), (t,\mathbf{q}_2)\in A$ \cite[\S1]{Filippov1988}. 
Obviously, if $\mathbf{u}(\cdot,\mathsf{S})$ is continuous and linearly bounded on $X$, then the Carath\'{e}odory solution coincides with the classical one. 
The Carath\'{e}odory solution allows discontinuity of $\mathbf{u}(\cdot,\mathsf{S})$ in $t$ but still requires continuity in $\mathbf{q}$ that is not guaranteed. 

Further extension that allows discontinuity of $\mathbf{u}(\cdot,\mathsf{S})$ in both $t$ and $\mathbf{q}$ can be given by Krasovskii \cite{Ceragioli1999}\cite{Cortes2008}. 
The Krasovskii regularization of $\mathbf{u}(\cdot,\mathsf{S})$ is given by
\begin{equation*}
	U(t,\mathbf{q},\mathsf{S}) = \bigcap_{\delta>0}\overline{\mathrm{co}}\,\mathbf{u}(t,\mathbf{q}+\delta B_n,\mathsf{S})
\end{equation*}
where $\overline{\mathrm{co}}$ stands for the convex closure.
The Krasovskii solution is a function $\mathbf{q}:[t_0,\infty)\to\mathbb{R}^n$ that is absolutely continuous on $[t_0,\infty)$ and satisfies $\mathbf{q}(t_0) = \mathbf{q}_0$ and the differential inclusion $\dot{\mathbf{q}}(t) \in U(t,\mathbf{q}(t),\mathsf{S})$ for almost all $t\in(t_0,\infty)$. 
The Krasovskii solution exists if $U(\mathbf{x},\mathsf{S})$ is a nonempty compact convex set for every $\mathbf{x}\in X$; $U(\cdot,\mathsf{S})$ is upper semicontinuous on $X$; and there exist $\gamma,c\in[0,\infty)$ such that $\|\mathbf{u}'\| \le \gamma\|\mathbf{q}\| + c$ for all $\mathbf{x}\in X$ and $\mathbf{u}'\in U(\mathbf{x},\mathsf{S})$ where $U(\cdot,\mathsf{S})$ is said to be upper semicontinuous at $\mathbf{x}_0\in X$ if for every open set $U_0\subset\mathbb{R}^n$ containing $U(\mathbf{x}_0,\mathsf{S})$ there exists a neighborhood $\Omega\subset X$ of $\mathbf{x}_0$ such that $U(\Omega,\mathsf{S}) = \bigcup_{\mathbf{x}\in\Omega}U(\mathbf{x},\mathsf{S})\subset U_0$ \cite[Corollary 1.12, Exercise 1.14]{Clarke2008}. 
If a Carath\'{e}odory solution exists, then it is also a Krasovskii solution because $\mathbf{u}(\mathbf{x},\mathsf{S}) \in U(\mathbf{x},\mathsf{S})$ for all $\mathbf{x}\in X$. 
If $\mathbf{u}$ is continuous on $X$, then $U(\mathbf{x},\mathsf{S}) = \{\mathbf{u}(\mathbf{x},\mathsf{S})\}$ for all $\mathbf{x}\in X$, so all solutions are identical.

\begin{lemma}
	\label{lem:existence_of_Krasovskii_solution}
	Let $\mathbf{u}$ be a PIK solution of $[\mathsf{S}]$ in the form of \eqref{eqn:class_of_pik_solutions} and $\mathsf{S}\in[\mathsf{S}]$.
	If $\mathbf{r}'$ is linearly bounded and $\mathbf{F}_q$, $\mathbf{R}^{-1}$, and $\mathbf{L}$ are bounded, then for every $(t_0,\mathbf{q}_0)\in X$ there exists a Krasovskii solution $\mathbf{q}:[t_0,\infty)\to\mathbb{R}^n$ of \eqref{eqn:differential_equation} satisfying $\mathbf{q}(t_0) = \mathbf{q}_0$.
\end{lemma}
\begin{proof}
	By the assumption, there exist $\gamma,c\in[0,\infty)$ satisfying $\|\mathbf{u}(\mathbf{x},\mathsf{S})\| \le \gamma\|\mathbf{q}\| + c$ for all $\mathbf{x}\in X$.
	It follows that $U(\mathbf{x},\mathsf{S}) \subset (\gamma\|\mathbf{q}\| + c)B_n$.
	Define a set-valued map $A:X\times (0,\infty) \to 2^{\mathbb{R}^n}$ as $A(\mathbf{x},\delta) = \overline{\mathrm{co}}\, \mathbf{u}(\mathbf{x} + \delta B_X) = \overline{\mathrm{co}}\, \bigcup_{\mathbf{x}' \in \mathbf{x} + \delta B_X}\{\mathbf{u}(\mathbf{x}')\}$. 
	Then, $U(\mathbf{x}) = \bigcap_{\delta>0}A(\mathbf{x},\delta) = \{\mathbf{u}'\in \mathbb{R}^n\mid \mathbf{u}'\in A(\mathbf{x},\delta),\forall\delta>0\}$. 
	Fix $\mathbf{x}\in X$. Since $\mathbf{u}(\mathbf{x}) \in U(\mathbf{x})$, $U(\mathbf{x})$ is nonempty. 
	Let $\mathbf{u}_i'\in U(\mathbf{x})$, $i=1,2,\dots$, and $\mathbf{u}_i'\to \mathbf{u}_0$ as $i\to\infty$. 
	Suppose that $\mathbf{u}_0\not\in U(\mathbf{x})$. 
	Then, there exists $\delta_0>0$ such that $\mathbf{u}_0\notin A(\mathbf{x},\delta_0)$. 
	Since $0<\delta<\delta_0$ implies $A(\mathbf{x},\delta)\subset A(\mathbf{x},\delta_0)$, we have $\mathbf{u}_0\notin A(\mathbf{x},\delta)$ for all $0<\delta\le\delta_0$. 
	Let $\delta_0>\delta_1>\delta_2>\cdots$ and $\delta_i\to0$ as $i\to\infty$. 
	Since $A(\mathbf{x},\delta)$ is closed and $\mathbf{u}_0\not\in A(\mathbf{x},\delta_0) \supset A(\mathbf{x},\delta_1) \supset A(\mathbf{x},\delta_2) \supset \cdots$, we have $d(\mathbf{u}_0,A(\mathbf{x},\delta_i)) \ge d(\mathbf{u}_0,A(\mathbf{x},\delta_0)) > 0$ for all $i\in\mathbb{N}$, while $\mathbf{u}_i'\in U(\mathbf{x})\iff \mathbf{u}_i'\in A(\mathbf{x},\delta)$ for all $\delta>0$, so $d(\mathbf{u}_0,A(\mathbf{x},\delta_i)) \le \|\mathbf{u}_0 - \mathbf{u}_i'\| \to 0$, a contradiction. 
	Thus, $\mathbf{u}_0 \in U(\mathbf{x})$ and $U(\mathbf{x})$ is closed.
	Since $A(\mathbf{x},\delta)$ is convex, if $\mathbf{u}_1',\mathbf{u}_2'\in U(\mathbf{x})$, then $\beta\mathbf{u}_1' + (1-\beta)\mathbf{u}_2'\in A(\mathbf{x},\delta)$ for all $\beta\in[0,1]$ and $\delta\in(0,\infty)$, so $\beta\mathbf{u}_1' + (1-\beta)\mathbf{u}_2'\in U(\mathbf{x})$ for all $\beta\in[0,1]$ and $U(\mathbf{x})$ is convex.
	From $\mathbf{u}(\mathbf{x}') \in (\gamma\|\mathbf{q}'\| + c)B_n$ for all $\mathbf{x}' = (t',\mathbf{q}') \in X$, we find 
	\begin{align*}
	U(\mathbf{x}) &= \bigcap_{\delta>0}\overline{\mathrm{co}}\,\bigcup_{\mathbf{x}' \in \mathbf{x} + \delta B_X}\{\mathbf{u}(\mathbf{x}')\} \\
	&\subset \bigcap_{\delta>0}\overline{\mathrm{co}}\,\bigcup_{\mathbf{q}' \in \mathbf{q} + \delta B_n}(\gamma\|\mathbf{q}'\|+c)B_n \\
	&\subset \bigcap_{\delta>0}(\gamma\|\mathbf{q}\| + \gamma\delta + c)B_n \\
	&= (\gamma\|\mathbf{q}\|+c)B_n.
	\end{align*}
	Therefore, $U(\mathbf{x})$ is a nonempty compact convex set for every $\mathbf{x}\in X$ and there exist $\gamma,c\in[0,\infty)$ such that $\|\mathbf{u}'\| \le \gamma\|\mathbf{q}\| + c$ for all $\mathbf{x}\in X$ and $\mathbf{u}'\in U(\mathbf{x})$.
	
	The graph of $U$ is defined as $\mathrm{gr}U  = \{(\mathbf{x}',\mathbf{u}')\in X\times\mathbb{R}^n\mid \mathbf{u}'\in U(\mathbf{x}')\}$. 
	Let $(\mathbf{x}_i',\mathbf{u}_i')\in\mathrm{gr}U$ for $i=1,2,3,\dots$ and $(\mathbf{x}_i',\mathbf{u}_i')\to(\mathbf{x}_0,\mathbf{u}_0)$ as $i\to\infty$. 
	Suppose that $\mathbf{u}_0\not\in U(\mathbf{x}_0)$. 
	Then, there exists $\delta_0>0$ such that $\mathbf{u}_0\not\in A(\mathbf{x}_0,\delta_0)$ and $d(\mathbf{u}_0,A(\mathbf{x}_0,\delta_0))>0$. 
	Let $\delta_0>\delta_1>\delta_2>\cdots$ and $\delta_i\to0$ as $i\to\infty$. 
	Without loss of generality, $\mathbf{x}_i'+\delta_iB_X \subset \mathbf{x}_0+\delta_0B_X$ for all $i\in\mathbb{N}$. 
	Then, $\mathbf{u}_i' \in A(\mathbf{x}_i',\delta_i) = \overline{\mathrm{co}}\,\mathbf{u}(\mathbf{x}_i'+\delta_iB_n) \subset \overline{\mathrm{co}}\,\mathbf{u}(\mathbf{x}_0+\delta_0B_n) = A(\mathbf{x}_0,\delta_0)$ for all $i\in\mathbb{N}$, so $d(\mathbf{u}_0,A(\mathbf{x}_0,\delta_0)) \le \|\mathbf{u}_0-\mathbf{u}_i'\|\to0$ as $i\to\infty$, a contradiction. 
	Thus, $(\mathbf{x}_0,\mathbf{u}_0)\in\mathrm{gr}U$ and $\mathrm{gr}U$ is closed. 
	Let $(\mathbf{x},\delta)\in X\times(0,\infty)$. 
	Define $U_0 = \mathrm{cl}\,\bigcup_{\mathbf{x}' \in \mathbf{x} + \delta B_X}U(\mathbf{x}')$. 
	Then, $U_0 \subset \mathrm{cl}\,\bigcup_{\mathbf{q}' \in \mathbf{q} + \delta B_n}(\gamma\|\mathbf{q}'\|+c)B_n \subset (\gamma\|\mathbf{q}\|+\gamma\delta+c)B_n$, so $U_0$ is compact. 
	Therefore, $U$ is upper semicontinuous \cite[Proposition 2.2]{Smirnov2002} and there exists $\mathbf{q}\in\mathrm{AC}([t_0,\infty),\mathbb{R}^n)$ satisfying $\mathbf{q}(t_0) = \mathbf{q}_0$ and $\dot{\mathbf{q}}(t)\in U(t,\mathbf{q}(t))$ for almost all $t\in(t_0,\infty)$ \cite[Corollary 1.12, Exercise 1.14]{Clarke2008}.
\end{proof}

\begin{corollary}
	\label{cor:existence_of_krasovskii_solution}
	Let $\alpha\in\overline{1,4}$, $\mathbf{u}$ be the $\boldsymbol{\pi}_\alpha$-PIK solution of $[\mathsf{S}]$ with the damping functions given by \eqref{eqn:damping_function}, and $\mathsf{S}\in[\mathsf{S}]$.
	Assume $\mu_1,\dots,\mu_l\in(0,\infty)$ if $\alpha \in\overline{1,3}$; $\mathbf{r}'$ is linearly bounded; and $\mathbf{F}_q$ and $\mathbf{R}^{-1}$ are bounded.
	Then, for every $(t_0,\mathbf{q}_0)\in X$, there exists a Krasovskii solution $\mathbf{q}:[t_0,\infty)\to\mathbb{R}^n$ of \eqref{eqn:differential_equation} satisfying $\mathbf{q}(t_0) = \mathbf{q}_0$.
\end{corollary}
\begin{proof}
	It will be sufficient to show that $\mathbf{L}$ is bounded.
	It is obvious for $\alpha = 4$.
	Let $\alpha\in\overline{1,3}$.
	If $|(\mathbf{C}_{aa}\mathbf{C}_{aa}^T)(\mathbf{x})|^\nu = 0$, then $\lambda_a^2(\mathbf{x}) = \infty$ and $\|\mathbf{D}_a(\mathbf{x})\| = \|\mathbf{H}_a(\mathbf{x})\| = 0$.
	If $|(\mathbf{C}_{aa}\mathbf{C}_{aa}^T)(\mathbf{x})|^\nu > 0$, then $\lambda_a^2(\mathbf{x})\in(0,\infty)$ and 
	\begin{equation*}
		\max\{\|\mathbf{D}_a(\mathbf{x})\|,\,\|\mathbf{H}_a(\mathbf{x})\|\} 
			\le \frac{1}{\lambda_a^2(\mathbf{x})} 
			\le \frac{\|\mathbf{J}(\mathbf{x})\|_F^{2\nu m_a}}{\mu_a^2}.
	\end{equation*}
	It follows that $\mathbf{D}$ and $\mathbf{H}$ are bounded.
	By Lemma \ref{lem:bound_of_extended_damped_pseudoinverse}, $\|\mathbf{C}_{aa}^*(\mathbf{x})\| \le \frac{1}{2\mu_a}\|\mathbf{C}_{aa}(\mathbf{x})\|^{\nu m_a}$ for all $(a,\mathbf{x})$. 
	Then, 
	\begin{align*}
		\|(\mathbf{I}_m + \mathbf{C}_L\mathbf{C}_D^\circledast)^{-1}(\mathbf{x})\|
			&\le \sum_{i=0}^{l-1}(\|\mathbf{C}_L(\mathbf{x})\|\|\mathbf{C}_D^\circledast(\mathbf{x})\|)^i \\
			&\le \sum_{i=0}^{l-1}\left(\sum_{a=1}^l\frac{\|\mathbf{J}(\mathbf{x})\|_F^{\nu m_a + 1}}{2\mu_a}\right)^i.
	\end{align*}
	Therefore, $\mathbf{L}$ is bounded for all $\alpha\in\overline{1,4}$ and a Krasovskii solution of \eqref{eqn:differential_equation} satisfying $\mathbf{q}(t_0) = \mathbf{q}_0$ exists by Lemma \ref{lem:existence_of_Krasovskii_solution}.
\end{proof}

The Krasovskii solution exists under mild conditions compared to the classical and Carath\'{e}odory solutions but raise difficulties to handle set-valued maps and nonsmooth analysis in studying properties of joint trajectories such as task convergence and stability. 
A more serious problem is that the Krasovskii solution does not guarantee priority relations between tasks.
For example, if a Krasovskii solution follows a manifold in which discontinuity of a PIK solution occurs, then the joint velocity tangent to the manifold may violate the priority relations.
On the other hand, the Carath\'{e}odory solution guarantees the priority relations almost everywhere.
Therefore, it would be more beneficial to confine ourselves to a set of PIK solutions and a set of initial values that guarantee existence of Carath\'{e}odory or classical solutions in the context of the PIK.
So, we find an alternative existence and uniqueness condition of the classical solution of \eqref{eqn:differential_equation}.
We will need the next technical lemma.

\begin{lemma}
	\label{lem:Lipschitz_continuity_of_linear_approximation}
	Let $[\mathsf{S}]^\bullet\subset\mathbb{S}^\bullet$ and $\mathbf{x}_0\in X\setminus\mathrm{int}(\mathcal{G}_\mathsf{S}^\bullet)\neq\emptyset$. If $\mathbf{J}\in C_{\mathbf{x}_0}^1(X,\mathbb{R}^{m\times n})$ and $\mathsf{D}\mathbf{J}\in C_{\mathbf{x}_0}^{L_p}(X,\mathbb{R}^{m\times n\times (n+1)})$, then there exist $a\in\overline{1,m}$ and $r,L\in(0,\infty)$ satisfying 
	\begin{align}
	c_{aa}(\mathbf{x}_0) = 0 \label{eqn:priority_singularity_condition1}\\
	c_{aa}\hat{\mathbf{j}}_a^T \in C_{\mathbf{x}_0+rB_X}^{1_p}(X,\mathbb{R}^n) \label{eqn:priority_singularity_condition2}\\
	\|\mathbf{h}(\mathbf{x})\| \le L\|\mathbf{x} - \mathbf{x}_0\|^2, \, \forall\mathbf{x}\in \mathbf{x}_0+rB_X \label{eqn:priority_singularity_condition3}
	\end{align}
	where $\mathbf{h}(\mathbf{x}) = (c_{aa}\hat{\mathbf{j}}_a^T)(\mathbf{x}) - \mathsf{D}(c_{aa}\hat{\mathbf{j}}_a^T)(\mathbf{x}_0)(\mathbf{x} - \mathbf{x}_0)$.
\end{lemma}
\begin{proof}
	$\det(\mathbf{C}(\mathbf{x}_0)) = \prod_{i=1}^mc_{ii}(\mathbf{x}_0) = 0$ because $\mathbf{x}_0 \in X\setminus\mathrm{int}(\mathcal{G}_\mathsf{S}^\bullet)\subset X\setminus\mathcal{G}_\mathsf{S}$. 
	So, there exists $a \in\{1,\dots,m\}$ satisfying $c_{aa}(\mathbf{x}_0) = 0$ and $c_{ii}(\mathbf{x}_0)>0$ if $1\le i<a$. 
	Let $\mathcal{S} = \{\hat{\mathbf{j}}_i\mid a\le i\le n\}$. 
	Assume $a > 1$. 
	$c_{ii} = \langle c_{ii}\hat{\mathbf{j}}_i,c_{ii}\hat{\mathbf{j}}_i\rangle^{1/2} \in C_{\mathbf{x}_0}^0$ for $i\in\overline{1,a-1}$ because $c_{ii}\hat{\mathbf{j}}_i \in C_{\mathbf{x}_0}^0$ by Proposition \ref{prp:properties_of_MCBS}.\ref{prp:MCBS:first_element_of_S_is_smooth}. 
	$\mathbf{J}_{1:a-1},\mathbf{j}_a\in C_{\mathbf{x}_0}^1$ and $\mathsf{D}\mathbf{J}_{1:a-1},\mathsf{D}\mathbf{j}_a\in C_{\mathbf{x}_0}^{L_p}$ by Lemma \ref{lem:derivative_of_matrix_and_its_entries}. 
	So, there exists $r_0>0$ such that $\mathbf{J}_{1:a-1},\mathbf{j}_a\in C_{\mathbf{x}_0+r_0B_X}^{1_p}$ and $\lim_{\mathbf{x}'\to\mathbf{x}}\mathrm{rank}(\mathbf{J}_{1:a-1}(\mathbf{x}')) = \mathrm{rank}(\mathbf{J}_{1:a-1}(\mathbf{x})) = a-1$ if $\mathbf{x} \in \mathbf{x}_0 + r_0B_X$. 
	By Lemma \ref{lem:ContinuityOfPseudoinverse-modified}, $\mathsf{D}\mathbf{J}_{1:a-1}^+(\mathbf{x}) = -(\mathbf{J}_{1:a-1}^+\mathsf{D}\mathbf{J}_{1:a-1} \mathbf{J}_{1:a-1}^+)(\mathbf{x}) + \big(\mathbf{J}_{1:a-1}^+\mathbf{J}_{1:a-1}^{+T}\mathsf{D}\mathbf{J}_{1:a-1}^T(\mathbf{I}_{a-1} - \mathbf{J}_{1:a-1} \mathbf{J}_{1:a-1}^+)\big)(\mathbf{x}) + \big((\mathbf{I}_n - \mathbf{J}_{1:a-1}^+\mathbf{J}_{1:a-1})\mathsf{D}\mathbf{J}_{1:a-1}^T\mathbf{J}_{1:a-1}^{+T}\mathbf{J}_{1:a-1}^+\big)(\mathbf{x})$
	for every $\mathbf{x}\in \mathbf{x}_0 + r_0B_X$.
	By Lemmas \ref{lem:derivative_of_matrix_and_its_entries} and \ref{lem:ContinuityOfPseudoinverse-modified}, $\mathbf{J}_{1:a-1}^+ \in C_{\mathbf{x}_0+r_0B_X}^{1_p}$ and $\mathsf{D}\mathbf{J}_{1:a-1}^+\in C_{\mathbf{x}_0}^{L_p}$.
	Therefore, $\mathbf{P}_\mathcal{S} = \sum_{i=a}^n\hat{\mathbf{j}}_i^T\hat{\mathbf{j}}_i = \mathbf{I}_n - \sum_{i=1}^{a-1}\hat{\mathbf{j}}_i^T\hat{\mathbf{j}}_i = \mathbf{I}_n - \mathbf{J}_{1:a-1}^+\mathbf{J}_{1:a-1} \in C_{\mathbf{x}_0+r_0B_X}^{1_p}$ and $\mathsf{D}\mathbf{P}_\mathcal{S} \in C_{\mathbf{x}_0}^{L_p}$. 
	If $a = 1$, then $\mathbf{P}_\mathcal{S} = \mathbf{I}_n$ and $\mathsf{D}\mathbf{P} = \mathbf{0}$, so we have the same result with arbitrary $r_0\in(0,\infty)$.
	
	Define $\mathbf{g} = (g_1,\dots,g_n) = c_{aa}\hat{\mathbf{j}}_a^T = (\hat{\mathbf{j}}_a^T\hat{\mathbf{j}}_a)\mathbf{j}_a^T = \mathbf{P}_\mathcal{S}\mathbf{j}_a^T$. 
	Since $\mathbf{g}\in C_{\mathbf{x}_0+r_0B_X}^{1_p}$ and $\mathsf{D}\mathbf{g}\in C_{\mathbf{x}_0}^{L_p}$, there exist $r \in (0,r_0)$ and $L\in(0,\infty)$ such that $\|\mathsf{D}\mathbf{g}(\mathbf{x}) - \mathsf{D}\mathbf{g}(\mathbf{x}_0)\| \le n^{-3/2}L\|\mathbf{x} - \mathbf{x}_0\|$ if $\mathbf{x}\in\mathbf{x}_0+rB_X$. 
	Let $\mathbf{x}_1\in(\mathbf{x}_0+rB_X)\setminus\{\mathbf{x}_0\}$ and define $\mathbf{x}:[0,1]\to\mathbf{x}_0+rB_X$ as $\mathbf{x}(s) = \mathbf{x}_0 + s(\mathbf{x}_1 - \mathbf{x}_0)$. 
	Since $\mathbf{x}$ is differentiable for all $s\in(0,1)$, we have
	\begin{equation*}
	\frac{d}{ds}g_i(\mathbf{x}(s)) = \mathsf{D}g_i(\mathbf{x}(s))\frac{d}{ds}\mathbf{x}(s) = \mathsf{D}g_i(\mathbf{x}(s))(\mathbf{x}_1-\mathbf{x}_0)
	\end{equation*}
	for all $s\in(0,1)$ by the chain rule.
	Since $s\mapsto g_i(\mathbf{x}(s))$ is continuous on $[0,1]$ and differentiable in $(0,1)$, there exists $s_i\in(0,1)$ such that
	\begin{equation*}
	g_i(\mathbf{x}_1) = g_i(\mathbf{x}(1)) - g_i(\mathbf{x}(0)) = (1-0)\frac{d}{ds}g_i(\mathbf{x}(s_i))
	\end{equation*}
	for all $i\in\overline{1,n}$ by the mean value theorem.
	Then,
	\begin{align*}
	&\|(c_{aa}\hat{\mathbf{j}}_a^T)(\mathbf{x}_1) - \mathsf{D}(c_{aa}\hat{\mathbf{j}}_a^T)(\mathbf{x}_0)(\mathbf{x}_1-\mathbf{x}_0)\| \\
	&\le \sum_{i=1}^n|g_i(\mathbf{x}_1) - \mathsf{D}g_i(\mathbf{x}_0)(\mathbf{x}_1-\mathbf{x}_0)| \\
	&\le \sum_{i=1}^n\|\mathsf{D}g_i(\mathbf{x}(s_i)) - \mathsf{D}g_i(\mathbf{x}_0)\|\|\mathbf{x}_1-\mathbf{x}_0\| \\
	&\le \sum_{i=1}^n\sqrt{n}\|\mathsf{D}\mathbf{g}(\mathbf{x}(s_i))-\mathsf{D}\mathbf{g}(\mathbf{x}_0)\|\|\mathbf{x}_1-\mathbf{x}_0\| \\
	&\le L\|\mathbf{x}_1-\mathbf{x}_0\|^2.
	\end{align*}
	If $\mathbf{x}_1 = \mathbf{x}_0$, then $\|(c_{aa}\hat{\mathbf{j}}_a^T)(\mathbf{x}_1) - \mathsf{D}(c_{aa}\hat{\mathbf{j}}_a^T)(\mathbf{x}_0)(\mathbf{x}_1-\mathbf{x}_0)\| = 0 = L\|\mathbf{x}_1-\mathbf{x}_0\|^2$.
\end{proof}

Now, we are ready to deliver an alternative existence and uniqueness theorem for the solution of \eqref{eqn:differential_equation}.
\begin{theorem}
	\label{thm:existence_of_joint_trajectory}
	Let $\mathbf{u}$ be a PIK solution of $[\mathsf{S}]^\bullet$ in the form of \eqref{eqn:class_of_pik_solutions} and $\mathsf{S}\in[\mathsf{S}]^\bullet$.
	Assume that
	\begin{enumerate}
		\item $\mathbf{r}'$ is linearly bounded;
		\item $\mathbf{F}_q$, $\mathbf{R}^{-1}$, and $\mathbf{L}$ are bounded;
		\item $\mathbf{F}_q(t,\cdot) = \mathbf{F}_q(0,\cdot)$ and $\mathbf{R}(t,\cdot) = \mathbf{R}(0,\cdot)$ for all $t\in\mathbb{R}$;
		\item $\mathbf{R}\in C^{1_p}$ and $\mathbf{L}\in C_{\mathcal{G}^\bullet(\mathsf{S})}^\bullet$;
		\item if $X\setminus\mathcal{G}_\mathsf{S}^\bullet\neq\emptyset$, then for every $\mathbf{x}=(t,\mathbf{q})\in X\setminus\mathcal{G}_\mathsf{S}^\bullet$ one of the followings holds:
		\begin{enumerate}
			\item there exist $r,L\in(0,\infty)$ such that $\|(\mathbf{L}\mathbf{r}')(\mathbf{x}')\| \le L\|\mathbf{q}'-\mathbf{q}\|$ for all $\mathbf{x}' = (t',\mathbf{q}')\in\mathbf{x} + rB_X$;
			\item there exist $a\in\overline{1,m}$ and $r,L\in(0,\infty)$ satisfying \eqref{eqn:priority_singularity_condition1} to \eqref{eqn:priority_singularity_condition3} such that $\mathbf{R}^T(\mathbf{x})\mathsf{D}_q(c_{aa}\hat{\mathbf{j}}_a^T)(\mathbf{x}) + (\mathsf{D}_q(c_{aa}\hat{\mathbf{j}}_a^T)(\mathbf{x}))^T\mathbf{R}(\mathbf{x})$ is either positive or negative definite and $\|\sum_{i=1}^b(\mathbf{L}_{bi}\mathbf{r}_i')(\mathbf{x}')\| \le L\|\mathbf{q}'-\mathbf{q}\|$ for all $\mathbf{x}'\in\mathbf{x}+rB_X$ and $\hat{\mathbf{j}}_a\hat{\mathbf{J}}_b^T\neq\mathbf{0}$.
		\end{enumerate}
		\label{thm:existence_of_joint_trajectory_condition_for_priority_singularity}
	\end{enumerate}
	Then, for each $(t_0,\mathbf{q}_0)\in\mathrm{int}(\mathcal{G}_\mathsf{S}^\bullet)$ there exists a classical solution $\mathbf{q}:[t_0,\infty)\to\mathrm{int}(\mathcal{G}_\mathsf{S}^\bullet)$ of \eqref{eqn:differential_equation} satisfying $\mathbf{q}(t_0) = \mathbf{q}_0$. 
	If $\bullet\in\{L,1\}$, then the solution is unique.
\end{theorem}
\begin{proof}
	Fix $\mathbf{x}_0 = (t_0,\mathbf{q}_0) \in \mathrm{int}(\mathcal{G}_\mathsf{S}^\bullet)$.
	By Lemma \ref{lem:existence_of_Krasovskii_solution}, there exists a Krasovskii solution $\mathbf{q}:[t_0,\infty)\to\mathbb{R}^n$ of \eqref{eqn:differential_equation} with the initial value $\mathbf{q}(t_0) = \mathbf{q}_0$. 
	If $X\setminus\mathrm{int}(\mathcal{G}_\mathsf{S}^\bullet) = \emptyset$, then $\mathcal{G}_\mathsf{S}^\bullet = X$ and $\mathbf{u}(\cdot,\mathsf{S}) \in C^\bullet\subset C^0$.
	Therefore, the Krasovskii solution coincides with the classical one. 
	Assume that $X\setminus\mathrm{int}(\mathcal{G}_\mathsf{S}^\bullet) \neq\emptyset$. 
	We show that $\mathbf{x}(t) = (t,\mathbf{q}(t)) \in \mathrm{int}(\mathcal{G}_\mathsf{S}^\bullet)$ for all $t\in[t_0,\infty)$ by contradiction. 
	Suppose that there exists $t_2\in(t_0,\infty)$ such that $\mathbf{x}(t) \in \mathrm{int}(\mathcal{G}_\mathsf{S}^\bullet)$ if $t\in[t_0,t_2)$ and $\mathbf{x}(t_2) = \mathbf{x}_2 = (t_2,\mathbf{q}_2) \in X\setminus\mathrm{int}(\mathcal{G}_\mathsf{S}^\bullet)$. 
	Since $\mathbf{u}(\cdot,\mathsf{S})$ is continuous on $\mathrm{int}(\mathcal{G}_\mathsf{S}^\bullet)\subset\mathcal{G}_\mathsf{S}^0$, $\dot{\mathbf{q}}(t)\in U(t,\mathbf{q}(t),\mathsf{S}) = \{\mathbf{u}(t,\mathbf{q}(t),\mathsf{S})\}$ for all $t\in(t_0,t_2)$. 
	
	We first assume that there exist $r,L\in(0,\infty)$ such that $\|(\mathbf{L}\mathbf{r}')(\mathbf{x}')\| \le L\|\mathbf{q}'-\mathbf{q}_2\|$ for all $\mathbf{x}'=(t',\mathbf{q}')\in\mathbf{x}_2+rB_X$.
	There exists $t_1\in[t_0,t_2)$ satisfying $\mathbf{x}(t)\in\mathbf{x}_2+rB_X$ for all $t\in[t_1,t_2]$.
	Since $\mathbf{F}_q$ and $\mathbf{R}^{-1}$ are assumed to be bounded, there exists $M_u\in(0,\infty)$ such that $\|\mathbf{u}(\mathbf{x}',\mathsf{S})\| \le M_u\|\mathbf{q}'-\mathbf{q}_2\|$ for all $\mathbf{x}'\in\mathbf{x}_2+rB_X$.
	Define $V:\mathbb{R}^n\to[0,\infty)$ as $V(\mathbf{q}') = \|\mathbf{q}'-\mathbf{q}_2\|^2/2$ and $\phi:[t_1,t_2]\to\mathbb{R}$ as $\phi(t) = -V(\mathbf{q}(t))$.
	Since $\mathbf{q}'\mapsto -V(\mathbf{q}')$ is Lipschitz on $\mathbf{q}_2+rB_n$ and $t\mapsto\mathbf{q}(t)$ is absolutely continuous on $[t_1,t_2]$, $\phi$ is absolutely continuous on $[t_1,t_2]$.
	Since $\dot{\mathbf{q}}(t) = \mathbf{u}(\mathbf{x}(t),\mathsf{S})$ for all $t\in(t_1,t_2)$, we have
	\begin{equation*}
	\dot{\phi}(t) = -\langle\mathbf{q}(t)-\mathbf{q}_2,\mathbf{u}(\mathbf{x}(t),\mathsf{S})\rangle \le -\alpha\phi(t)
	\end{equation*}
	for almost all $t\in[t_1,t_2]$ where $\alpha=2M_u\in(0,\infty)$.
	By the Gronwall's inequality, we find a contradiction that
	\begin{equation}
	\label{eqn:contradiction_of_V}
	0 = -V(\mathbf{q}(t_2)) \le -V(\mathbf{q}(t_1))e^{-\alpha(t_2-t_1)} < 0.
	\end{equation}
	Therefore, $\mathbf{x}(t) \in \mathrm{int}(\mathcal{G}^\bullet)$ for all $t\in[t_0,\infty)$. 
	Since $\mathbf{u}(\cdot,\mathsf{S})$ is continuous on $\mathrm{int}(\mathcal{G}^\bullet)$, $\dot{\mathbf{q}}(t) = \mathbf{u}(t,\mathbf{q}(t),\mathsf{S})$ for all $t\in(t_0,\infty)$.
	
	Next, we assume that there exist $a\in\overline{1,m}$ and $r,L\in(0,\infty)$ satisfying 
	\begin{align*}
	c_{aa}(\mathbf{x}_2) = 0 \\
	c_{aa}\hat{\mathbf{j}}_a^T \in C_{\mathbf{x}_2+rB_X}^{1_p} \\
	\|\mathbf{h}(\mathbf{x}')\| \le L\|\mathbf{x}' - \mathbf{x}_2\|^2 \\
	\mathbf{R}^T(\mathbf{x}_2)\mathsf{D}_q(c_{aa}\hat{\mathbf{j}}_a^T)(\mathbf{x}_2) + (\mathsf{D}_q(c_{aa}\hat{\mathbf{j}}_a^T)(\mathbf{x}_2))^T\mathbf{R}(\mathbf{x}_2)>0 \\
	\left\|\sum_{i=1}^b(\mathbf{L}_{bi}\mathbf{r}_i')(\mathbf{x}')\right\| \le L\|\mathbf{q}'-\mathbf{q}_2\| \\
	\|\mathbf{R}(\mathbf{x}') - \mathbf{R}(\mathbf{x}_2)\| \le L\|\mathbf{x}'-\mathbf{x}_2\|
	\end{align*}
	for all $\mathbf{x}' = (t',\mathbf{q}')\in\mathbf{x}_2 + rB_X$ where $\mathbf{h}(\mathbf{x}') = (c_{aa}\hat{\mathbf{j}}_a^T)(\mathbf{x}') - \mathsf{D}(c_{aa}\hat{\mathbf{j}}_a^T)(\mathbf{x}_2)(\mathbf{x}'-\mathbf{x}_2)$ and $b\in\overline{1,l}$ with $\hat{\mathbf{j}}_a\hat{\mathbf{J}}_b^T \neq \mathbf{0}$.
	Without loss of generality, $r \le 1$.
	Since $\mathbf{R}^Tc_{aa}\hat{\mathbf{j}}_a^T \in C_{\mathbf{x}_2}^{1_p}$, there exists $\mathbf{h}':X\to\mathbb{R}^n$ satisfying $\mathbf{h}'(\mathbf{x}_2) = \mathbf{0}$ and $\lim_{\mathbf{x}'\to\mathbf{x}_2}\|\mathbf{h}'(\mathbf{x}')\|/\|\mathbf{x}'-\mathbf{x}_2\| = 0$ such that
	\begin{equation*}
	(\mathbf{R}^Tc_{aa}\hat{\mathbf{j}}_a^T)(\mathbf{x}') = (\mathbf{R}^T\mathsf{D}(c_{aa}\hat{\mathbf{j}}_a^T))(\mathbf{x}_2)(\mathbf{x}'-\mathbf{x}_2) + \mathbf{h}'(\mathbf{x}')
	\end{equation*}
	for all $\mathbf{x}'\in X$.
	Let $\mathbf{B} = (\mathbf{R}^T\mathsf{D}_q(c_{aa}\hat{\mathbf{j}}_a^T))(\mathbf{x}_2)$ and $\mathbf{A} = \frac{1}{2}(\mathbf{B} + \mathbf{B}^T)$.
	The assumption $\mathbf{J}(t,\cdot) = \mathbf{J}(0,\cdot)$ for all $t\in\mathbb{R}$ implies that if $(t,\mathbf{q}') \in \mathcal{G}_\mathsf{S}^\bullet$, then $(t',\mathbf{q}') \in \mathcal{G}_\mathsf{S}^\bullet$ for all $t'\in\mathbb{R}$. 
	So, we can write $\mathcal{G}_\mathsf{S}^\bullet = \mathbb{R}\times\mathcal{H}_\mathsf{S}^\bullet$ and $\mathrm{int}(\mathcal{G}_\mathsf{S}^\bullet) = \mathbb{R}\times\mathrm{int}(\mathcal{H}_\mathsf{S}^\bullet)$ for some $\mathcal{H}_\mathsf{S}^\bullet\subset\mathbb{R}^n$. 
	Then, $\mathbf{q}(t)\in\mathrm{int}(\mathcal{H}_\mathsf{S}^\bullet)$ for all $t\in[t_0,t_2)$ and $\mathbf{q}(t_2) = \mathbf{q}_2\in\mathbb{R}^n\setminus\mathrm{int}(\mathcal{H}_\mathsf{S}^\bullet)$. 
	Define a function $V:\mathbb{R}^n\to[0,\infty)$ as
	\begin{equation*}
	V(\mathbf{q}') = \frac{1}{2}\langle\mathbf{q}'-\mathbf{q}_2,\mathbf{A}(\mathbf{q}'-\mathbf{q}_2)\rangle = \frac{1}{2}\langle\mathbf{q}'-\mathbf{q}_2,\mathbf{B}(\mathbf{q}'-\mathbf{q}_2)\rangle.
	\end{equation*}
	Let $\rho>0$ be the minimum eigenvalue of $\mathbf{A}$ such that
	\begin{equation*}
	\frac{\rho}{2}\|\mathbf{q}' - \mathbf{q}_2\|^2 \le V(\mathbf{q}')
	\end{equation*}
	for all $\mathbf{q}'\in \mathbb{R}^n$. 
	There exists $t_1\in[t_0,t_2)$ satisfying $\mathbf{x}(t) \in \mathbf{x}_2 + rB_X$ for all $t\in[t_1,t_2]$. 
	Obviously, $\mathbf{q}(t)\in\mathbf{q}_2+rB_n$ for all $t\in[t_1,t_2]$.
	Define $\phi:[t_1,t_2]\to\mathbb{R}$ as $\phi(t) = -V(\mathbf{q}(t))$.
	Since $\mathbf{q}'\mapsto -V(\mathbf{q}')$ is Lipschitz on $\mathbf{q}_2+rB_n$ and $t\mapsto\mathbf{q}(t)$ is absolutely continuous on $[t_1,t_2]$, $\phi$ is absolutely continuous on $[t_1,t_2]$.
	Since $\dot{\mathbf{q}}(t) = \mathbf{u}(\mathbf{x}(t),\mathsf{S})$ for all $t\in(t_1,t_2)$, we have
	\begin{align*}
	\dot{\phi}(t) 
	&= -\langle \mathbf{B}(\mathbf{q}(t) - \mathbf{q}_2),\mathbf{u}(\mathbf{x}(t))\rangle \\
	&= -\langle (\mathbf{R}^T\mathsf{D}(c_{aa}\hat{\mathbf{j}}_a^T))(\mathbf{x}_2)(\mathbf{x}(t) - \mathbf{x}_2), \mathbf{u}(\mathbf{x}(t))\rangle \\
	&= -\langle (\mathbf{R}^Tc_{aa}\hat{\mathbf{j}}_a^T - \mathbf{h}')(\mathbf{x}(t)), \mathbf{u}(\mathbf{x}(t))\rangle \\
	&= -(c_{aa}\hat{\mathbf{j}}_a\mathbf{v}_b)(\mathbf{x}(t)) + \langle \mathbf{h}'(\mathbf{x}(t)), \mathbf{u}(\mathbf{x}(t))\rangle
	\end{align*}
	for almost all $t\in[t_1,t_2]$ where $\mathbf{v}_b = \hat{\mathbf{J}}_b^T\mathbf{C}_{bb}^T\sum_{i=1}^b\mathbf{L}_{bi}\mathbf{r}_i'$.
	
	We find various upper bounds.
	Since $\mathbf{J}(t,\cdot) = \mathbf{J}(0,\cdot)$ and $\mathbf{R}(t,\cdot) = \mathbf{R}(0,\cdot)$ for all $t\in\mathbb{R}$, we have
	\begin{align*}
		\|(c_{aa}\hat{\mathbf{j}}_a^T)(\mathbf{x}')\| 
		&\le \|\mathsf{D}(c_{aa}\hat{\mathbf{j}}_a^T)(\mathbf{x}_2)(\mathbf{x}'-\mathbf{x}_2)\| + \|\mathbf{h}(\mathbf{x}')\| \\
		&\le \|\mathsf{D}_q(c_{aa}\hat{\mathbf{j}}_a^T)(\mathbf{x}_2)\|\|\mathbf{q}'-\mathbf{q}_2\| + \|\mathbf{h}(t_2,\mathbf{q}')\| \\
		&\le (\|\mathsf{D}_q(c_{aa}\hat{\mathbf{j}}_a^T)(\mathbf{x}_2)\| + rL)\|\mathbf{q}'-\mathbf{q}_2\| \\
		&= M_1\|\mathbf{q}'-\mathbf{q}_2\|
	\end{align*}
	and
	\begin{align*}
	\|\mathbf{h}'(\mathbf{x}')\|
	&\le \|\mathbf{R}^T(t_2,\mathbf{q}') - \mathbf{R}^T(t_2,\mathbf{q}_2)\|\|(c_{aa}\hat{\mathbf{j}}_a^T)(\mathbf{x}')\| \\
	&\quad + \|\mathbf{R}^T(\mathbf{x}_2)\|\|\mathbf{h}(t_2,\mathbf{q}')\| \\
	&\le L(M_1 + \|\mathbf{R}^T(\mathbf{x}_2)\|)\|\mathbf{q}'-\mathbf{q}_2\|^2 \\
	&= M_2\|\mathbf{q}' - \mathbf{q}_2\|^2
	\end{align*}
	for all $\mathbf{x}'\in\mathbf{x}_2+rB_X$. 
	Since $\mathbf{F}_q$, $\mathbf{R}^{-1}$, and $\mathbf{L}$ are bounded and $\mathbf{r}'\in C^0$, there exists $M_3\in(0,\infty)$ satisfying $\max\{\|\mathbf{J}(\mathbf{x}')\|_F,\|\mathbf{u}(\mathbf{x}')\|\} \le M_3$ for all $\mathbf{x}'\in\mathbf{x}_2+rB_X$. 
	Then,
	\begin{equation*}
		\|\mathbf{v}_b(\mathbf{x}')\| 
			\le \|\mathbf{J}(\mathbf{x}')\|_F\left\|\sum_{i=1}^b(\mathbf{L}_{bi}\mathbf{r}_i')(\mathbf{x}')\right\| 
			\le LM_3\|\mathbf{q}' - \mathbf{q}_2\|
	\end{equation*}
	for all $\mathbf{x}'\in\mathbf{x}_2+rB_X$.
	
	We can find the upper bound of $\dot{\phi}(t)$ by using the above inequalities as
	\begin{align*}
		\dot{\phi}(t)
			&\le \|(c_{aa}\hat{\mathbf{j}}_a)(\mathbf{x}(t))\|\|\mathbf{v}_b(\mathbf{x}(t))\| + \|\mathbf{h}'(\mathbf{x}(t))\|\|\mathbf{u}(\mathbf{x}(t))\| \\
			&\le M_3(LM_1+M_2)\|\mathbf{q}(t)-\mathbf{q}_2\|^2 \\
			&\le \frac{2M_3(LM_1+M_2)}{\rho}V(\mathbf{q}(t)) \\
			&\le -\alpha\phi(t)
	\end{align*}
	for almost all $t\in[t_1,t_2]$ where $\alpha\in(0,\infty)$. 
	By the Gronwall's inequality, we find the same contradiction \eqref{eqn:contradiction_of_V}.
	
	If $\mathbf{B}+\mathbf{B}^T<0$, we define $V(\mathbf{q}') = -\frac{1}{2}\langle\mathbf{q}'-\mathbf{q}_2,\mathbf{B}(\mathbf{q}'-\mathbf{q}_2)\rangle$ and reach to the same contradiction \eqref{eqn:contradiction_of_V}.
	
	Let $\bullet\in\{L,1\}$ and assume that there are two classical solutions $\mathbf{q}_1$ and $\mathbf{q}_2$ of \eqref{eqn:differential_equation} with the initial value $\mathbf{q}_1(t_0) = \mathbf{q}_2(t_0) = \mathbf{q}_0$. 
	Fix $T\in(0,\infty)$. 
	Since $(t,\mathbf{q}_1(t)),(t,\mathbf{q}_2(t))\in\mathrm{int}(\mathcal{G}_\mathsf{S}^\bullet)$ for all $t\in[t_0,\infty)$, there exists a compact set $\Omega \subset \mathrm{int}(\mathcal{G}_\mathsf{S}^\bullet)\subset\mathrm{int}(\mathcal{G}_\mathsf{S}^L)$ satisfying $\{(t,\mathbf{q}_i(t))\mid i\in\{1,2\},\,t\in[t_0,t_0+T]\}\subset\Omega$. 
	Since $\mathbf{u}$ is locally Lipschitz on the compact set $\Omega$, $\mathbf{u}$ is Lipschitz on $\Omega$ with a Lipschitz constant $L\in[0,\infty)$. 
	Then,
	\begin{align*}
		\|\mathbf{q}_1(t) - \mathbf{q}_2(t)\| 
			&\le \int_{t_0}^t\|\mathbf{u}(t,\mathbf{q}_1(s)) - \mathbf{u}(t,\mathbf{q}_2(s))\|ds \\
			&\le L\int_{t_0}^t\|\mathbf{q}_1(s) - \mathbf{q}_2(s)\|ds
	\end{align*}
	for all $t\in[t_0,t_0+T]$. 
	By the Gronwall's inequality, $\mathbf{q}_1(t) = \mathbf{q}_2(t)$ for all $t\in[t_0,t_0+T]$. 
	Applying the same procedure repeatedly with the interval $[t_0+iT,t_0+(i+1)T]$ for $i=1,2,3,\dots$ proves that $\mathbf{q}_1(t) = \mathbf{q}_2(t)$ for all $t\in[t_0,\infty)$.
\end{proof}

\begin{corollary}
	\label{cor:existence_of_joint_trajectory}
	Let $\alpha\in\{1,4\}$, $\mathbf{u}$ be the $\boldsymbol{\pi}_\alpha$-PIK solution of $[\mathsf{S}]^\bullet$ with the damping functions given by \eqref{eqn:damping_function}, and $\mathsf{S}\in[\mathsf{S}]^\bullet$.
	Assume that
	\begin{enumerate}
		\item if $\alpha\in\{1,3\}$, then $\mu_1,\dots,\mu_l\in(0,\infty)$ and $\nu\in\mathbb{N}\cup\{0\}$;
		\item $\mathbf{r}'$ is linearly bounded;
		\item $\mathbf{F}_q$ and $\mathbf{R}^{-1}$ are bounded;
		\item $\mathbf{F}_q(t,\cdot) = \mathbf{F}_q(0,\cdot)$ and $\mathbf{R}(t,\cdot) = \mathbf{R}(0,\cdot)$ for all $t\in\mathbb{R}$;
		\item $\mathbf{R}\in C^{1_p}$;
		\item if $X\setminus\mathcal{G}_\mathsf{S}^\bullet\neq\emptyset$, then for every $\mathbf{x}=(t,\mathbf{q})\in X\setminus\mathcal{G}_\mathsf{S}^\bullet$ one of the followings holds:
		\begin{enumerate}
			\item there exist $r,L\in(0,\infty)$ such that $\|\mathbf{r}'(\mathbf{x}')\| \le L\|\mathbf{q}'-\mathbf{q}\|$ for all $\mathbf{x}' = (t',\mathbf{q}')\in\mathbf{x} + rB_X$;
			\item there exist $a\in\overline{1,m}$ and $r,L\in(0,\infty)$ satisfying \eqref{eqn:priority_singularity_condition1} to \eqref{eqn:priority_singularity_condition3} such that $\mathbf{R}^T(\mathbf{x})\mathsf{D}_q(c_{aa}\hat{\mathbf{j}}_a^T)(\mathbf{x}) + (\mathsf{D}_q(c_{aa}\hat{\mathbf{j}}_a^T)(\mathbf{x}))^T\mathbf{R}(\mathbf{x})$ is either positive or negative definite and 
			\begin{enumerate}
				\item if $\alpha\in\overline{1,3}$, then $\nu\in\mathbb{N}$;
				\item if $\alpha = 4$, then $\|\mathbf{r}_b'(\mathbf{x}')\| \le L\|\mathbf{q}'-\mathbf{q}\|$ for all $\mathbf{x}'\in\mathbf{x}+rB_X$ and $\hat{\mathbf{j}}_a\hat{\mathbf{J}}_b^T\neq\mathbf{0}$.
			\end{enumerate}
		\end{enumerate}
		\label{cor:existence_of_joint_trajectory_condition_for_priority_singularity}
	\end{enumerate}
	Then, for each $(t_0,\mathbf{q}_0)\in\mathrm{int}(\mathcal{G}_\mathsf{S}^\bullet)$ there exists a classical solution $\mathbf{q}:[t_0,\infty)\to\mathrm{int}(\mathcal{G}_\mathsf{S}^\bullet)$ of \eqref{eqn:differential_equation} satisfying $\mathbf{q}(t_0) = \mathbf{q}_0$. 
	If $\bullet\in\{L,1\}$, then the solution is unique.
\end{corollary}
\begin{proof}
	We showed in the proof of Corollary \ref{cor:existence_of_krasovskii_solution} that $(\mathbf{I}_m + \mathbf{C}_L\mathbf{C}_D^\circledast)^{-1}$ and $\mathbf{L}$ are bounded.
	Thus, \ref{cor:existence_of_joint_trajectory_condition_for_priority_singularity}-a) implies \ref{thm:existence_of_joint_trajectory_condition_for_priority_singularity}-a) of Theorem \ref{thm:existence_of_joint_trajectory}.
	We can show $\mathbf{D}_a,\mathbf{H}_a\in C_{\mathcal{G}^\bullet(\mathsf{S})}^\bullet$ similarly as in the proof of Lemma \ref{lem:smoothness_of_extended_damped_pseudoinverse}.
	It follows that $\mathbf{D},\mathbf{H},\mathbf{C}_L\mathbf{C}_D^\circledast\in C_{\mathcal{G}^\bullet(\mathsf{S})}^\bullet$ and $\mathbf{L}\in C_{\mathcal{G}^\bullet(\mathsf{S})}^\bullet$.
	We will show that the assumption \ref{cor:existence_of_joint_trajectory_condition_for_priority_singularity}-b) implies that for every $\mathbf{x}=(t,\mathbf{q})\in X\setminus\mathcal{G}_\mathsf{S}^\bullet\neq\emptyset$ there exist $r' \in(0,r]$ and $L' \in [L,\infty)$ such that $\|(\sum_{i=1}^b\mathbf{L}_{bi}\mathbf{r}_i')(\mathbf{x}')\|\le L'\|\mathbf{q}'-\mathbf{q}\|$ for all $\mathbf{x}'\in\mathbf{x}+r'B_X$ and $\hat{\mathbf{j}}_a\hat{\mathbf{J}}_b^T \neq\mathbf{0}$.
	Then, the assumption \ref{cor:existence_of_joint_trajectory_condition_for_priority_singularity}-b) will hold with $r'$ and $L'$ and the proof will be completed by Theorem \ref{thm:existence_of_joint_trajectory}.
	Assume that there exist $a\in\overline{1,m}$ and $r,L\in(0,\infty)$ satisfying \eqref{eqn:priority_singularity_condition1} to \eqref{eqn:priority_singularity_condition3} for $\mathbf{x}=(t,\mathbf{q})\in X\setminus\mathcal{G}_\mathsf{S}^\bullet\neq\emptyset$.
	Let $r'=\min\{1,r\}$.
	Since $(\mathbf{I}_m + \mathbf{C}_L\mathbf{C}_D^\circledast)^{-1}$ and $\mathbf{J}$ are bounded and $\mathbf{r}'\in C^0$, there exists $M_1\in(0,\infty)$ satisfying 
	\begin{equation*}
	\max\{\|(\mathbf{I}_m + \mathbf{C}_L\mathbf{C}_D^\circledast)^{-1}(\mathbf{x}')\|,\,\|\mathbf{J}(\mathbf{x}')\|_F,\,\|\mathbf{r}'(\mathbf{x}')\|\} \le M_1
	\end{equation*}
	for all $\mathbf{x}'\in\mathbf{x}+r'B_X$.
	We showed in the proof of Theorem \ref{thm:existence_of_joint_trajectory} that there exists $M_2\in(0,\infty)$ satisfying $\|(c_{aa}\hat{\mathbf{j}}_a^T)(\mathbf{x}')\| \le M_2\|\mathbf{q}' - \mathbf{q}\|$ for all $\mathbf{x}' \in \mathbf{x} + r'B_X$.
	Let $b\in\overline{1,l}$ be such that $\hat{\mathbf{j}}_a\hat{\mathbf{J}}_b^T\neq\mathbf{0}$.
	Let $\sigma_i(\mathbf{C}_{bb}(\mathbf{x}'))$ and $\mathrm{diag}_i(\mathbf{C}_{bb}(\mathbf{x}'))$ be the $i$-th singular value and the $i$-th diagonal entry of $\mathbf{C}_{bb}(\mathbf{x}')$, respectively. 
	By the Weyl's product inequality \cite[Problems 7.3.P17]{Horn2013},
	\begin{align*}
		&\max\{\|\mathbf{D}_b(\mathbf{x}')\|,\,\|\mathbf{H}_b(\mathbf{x}')\|\} \\
		&\quad \le \frac{1}{\mu_b^2}|(\mathbf{C}_{bb}\mathbf{C}_{bb})(\mathbf{x}')|^\nu \\
		&\quad = \frac{1}{\mu_b^2}\prod_{i=1}^{m_b}\sigma_i^{2\nu}(\mathbf{C}_{bb}(\mathbf{x}')) \\
		&\quad = \frac{1}{\mu_b^2}\prod_{i=1}^{m_b}\mathrm{diag}_i^{2\nu}(\mathbf{C}_{bb}(\mathbf{x}')) \\
		&\quad \le \frac{\|\mathbf{J}(\mathbf{x}')\|_F^{2\nu(m_b-1)}}{\mu_b^2}\|(c_{aa}\hat{\mathbf{j}}_a^T)(\mathbf{x}')\|^{2\nu} \\
		&\quad \le \frac{M_1^{2\nu(m_b-1)}M_2^{2\nu}}{\mu_b^2}\|\mathbf{q}' - \mathbf{q}\| \\
		&\quad = M_3\|\mathbf{q}' - \mathbf{q}\|
	\end{align*}
	for all $\mathbf{x}'\in \mathbf{x}+r'B_X$.
	Let $\mathbf{x}'\in\mathbf{x}+rB_X$.
	If $\alpha = 1$, let $\mathbf{A} = \begin{bmatrix} \mathbf{0} & \mathbf{D}_b & \mathbf{0}\end{bmatrix}:X\to\mathbb{R}^{m_b\times m}$ be the block of $\mathbf{D}$ containing $\mathbf{D}_b$.
	Then, $\|(\sum_{i=1}^b\mathbf{L}_{bi}\mathbf{r}_i')(\mathbf{x}')\| = \|(\mathbf{A}(\mathbf{I}_m + \mathbf{C}_L\mathbf{C}_D^\circledast)^{-1}\mathbf{r}')(\mathbf{x}')\| \le M_1^2\|\mathbf{D}_b(\mathbf{x}')\| \le M_1^2M_3\|\mathbf{q}' - \mathbf{q}\|$.
	If $\alpha\in\{2,3\}$, then $\|(\sum_{i=1}^b\mathbf{L}_{bi}\mathbf{r}'_i)(\mathbf{x}')\| = \|(\mathbf{L}_{bb}\mathbf{r}'_b)(\mathbf{x}')\|\le M_1M_3\|\mathbf{q}' - \mathbf{q}\|$.
	If $\alpha = 4$, $\|(\sum_{i=1}^b\mathbf{L}_{bi}\mathbf{r}_i')(\mathbf{x}')\| = \|\mathbf{r}_b'(\mathbf{x}')\| \le L\|\mathbf{q}' - \mathbf{q}\|$.
	Let $L' = \max\{M_1^2M_3,\,M_1M_3,\,L\}$.
\end{proof}

\begin{remark}
	If we only assume \ref{thm:existence_of_joint_trajectory_condition_for_priority_singularity}-a) in Theorem \ref{thm:existence_of_joint_trajectory} and \ref{cor:existence_of_joint_trajectory_condition_for_priority_singularity}-a) in Corollary \ref{cor:existence_of_joint_trajectory}, then we can easily check that $\mathbf{u}(\cdot,\mathsf{S})$ becomes continuous on $X$ and $\mathbf{u}(\mathbf{x},\mathsf{S}) = \mathbf{0}$ for all $\mathbf{x}\in X\setminus\mathcal{G}_\mathsf{S}^\bullet$.
	It implies that the joint trajectory moves slowly in the vicinity of every $\mathbf{x}\in X\setminus\mathcal{G}_\mathsf{S}^\bullet$.
	On the other hand, the assumptions \ref{thm:existence_of_joint_trajectory_condition_for_priority_singularity}-b) in Theorem \ref{thm:existence_of_joint_trajectory} and \ref{cor:existence_of_joint_trajectory_condition_for_priority_singularity}-b) in Corollary \ref{cor:existence_of_joint_trajectory} allow us $\mathbf{u}(\mathbf{x},\mathsf{S}) \neq \mathbf{0}$ for some $\mathbf{x}\in X\setminus\mathcal{G}_\mathsf{S}^\bullet$.
	Specifically, if \ref{thm:existence_of_joint_trajectory_condition_for_priority_singularity}-b) in Theorem \ref{thm:existence_of_joint_trajectory} or \ref{cor:existence_of_joint_trajectory_condition_for_priority_singularity}-b) in Corollary \ref{cor:existence_of_joint_trajectory} holds for some $\mathbf{x}\in X\setminus\mathcal{G}_\mathsf{S}^\bullet$, then it is possible that $\mathbf{r}_i'(\mathbf{x})\neq\mathbf{0}$ for all $i\in\{j\in\overline{1,l}\mid\hat{\mathbf{j}}_a\hat{\mathbf{J}}_j^T =\mathbf{0}\}$; thus $\mathbf{u}(\mathbf{x},\mathsf{S}) \neq\mathbf{0}$ is possible.
	It gives us the fast movement of the joint trajectory in the vicinity of some $\mathbf{x}\in X\setminus\mathcal{G}_\mathsf{S}^\bullet$, which is a great advantage in many practical applications.
	We will show an example that satisfies assumptions \ref{thm:existence_of_joint_trajectory_condition_for_priority_singularity}-b) of Theorem \ref{thm:existence_of_joint_trajectory} and \ref{cor:existence_of_joint_trajectory_condition_for_priority_singularity}-b) of Corollary \ref{cor:existence_of_joint_trajectory} in Section \ref{sec:example}.
\end{remark}

\section{Task Convergence}
\label{sec:task_convergence}

In many practical cases, a kinematic system is given as
\begin{equation}
	\label{eqn:task_convergence_kinematic_system}
	\mathsf{S} = (l,\mathbf{m},n,\mathsf{D}\mathbf{f},\mathbf{R},\boldsymbol{\Psi}(\dot{\mathbf{p}} + \mathbf{K}(\mathbf{p} - \mathbf{f})))\in\mathbb{S}^\bullet
\end{equation}
where $\mathbf{f} = (\mathbf{f}_1,\dots,\mathbf{f}_l)\in C^{1_p}(X,\mathbb{R}^m)$ is the forward kinematic function, $\mathbf{p} = (\mathbf{p}_1,\dots,\mathbf{p}_l)\in C^{1_p}(\mathbb{R},\mathbb{R}^m)$ satisfying $\dot{\mathbf{p}}\in C^\bullet$ is the trajectory for the task position $\mathbf{f}(\mathbf{x})$ to be desired to follow, $\mathbf{K} = \mathrm{diag}(k_1\mathbf{I}_{m_1},\dots,k_l\mathbf{I}_{m_l})\in\mathbb{R}^{m\times m}$ with $k_a\in(0,\infty)$ is the feedback gain matrix, and $\boldsymbol{\Psi} = \mathrm{diag}(\psi_1\mathbf{I}_{m_1},\dots,\psi_l\mathbf{I}_{m_l})\in C^\bullet(X,\mathbb{R}^{m\times m})$ with $\psi_a(\mathbf{x})\in[0,1]$ is the activation function that can be used to activate or deactivate the term $\dot{\mathbf{p}} + \mathbf{K}(\mathbf{p} - \mathbf{f})$ \cite{Chiacchio1991}. 
Then, a PIK solution of $\mathsf{S}$ can be considered as an output tracking control law of the dynamical system
\begin{align*}
	\dot{\mathbf{q}} &= \mathbf{u} \\
	\mathbf{p}_a &= \mathbf{f}_a(t,\mathbf{q}),\quad a\in\overline{1,l} \\
	\mathsf{T}_1&\prec\cdots\prec\mathsf{T}_l
\end{align*}
where $\mathbf{q}\in\mathbb{R}^n$ is the state, $\mathbf{u}\in\mathbb{R}^n$ is the control input, and $\mathsf{T}_1\prec\cdots\prec\mathsf{T}_l$ represents the priority relations between multiple outputs $\mathbf{p}_a\in\mathbb{R}^{m_a}$ for $a\in\overline{1,l}$ in this case.
Let $\mathbf{u}$ be a PIK solution of $[\mathsf{S}]^\bullet$ in the form of \eqref{eqn:class_of_pik_solutions}.
Since the reference is fixed to $\mathbf{r} = \boldsymbol{\Psi}(\dot{\mathbf{p}} + \mathbf{K}(\mathbf{p}-\mathbf{f}))$, we may write $\mathbf{u}(t,\mathbf{q}) = \mathbf{u}(t,\mathbf{q},\mathsf{S})$.
Let $\mathbf{q}\in\mathrm{AC}([t_0,\infty),\mathbb{R}^n)$ be a Carath\'{e}odory solution of \eqref{eqn:differential_equation} with an initial value $(t_0,\mathbf{q}_0)\in\mathbb{R}\times\mathbb{R}^n$.
The existence condition can be given by \cite[\S1]{Filippov1988} or Theorem \ref{thm:existence_of_joint_trajectory}. 
Denote $\mathbf{x}(t) = (t,\mathbf{q}(t))$.
Obviously, $\mathbf{x}\in\mathrm{AC}([t_0,\infty),X)$.

The $a$-th reference $\mathbf{r}_a$ contains a feedforward $\dot{\mathbf{p}}_a$ and a feedback $k_a\mathbf{e}_a$ where $\mathbf{e}_a(t,\mathbf{q}) = \mathbf{p}_a(t) - \mathbf{f}_a(t,\mathbf{q})$ is the $a$-th task error or tracking error.
$\mathbf{u}$ minimizes the residuals $\mathbf{e}_a^\mathrm{res} = \mathbf{r}_a - \dot{\mathbf{f}}_a = \mathbf{r}_a' - \mathbf{J}_a\mathbf{R}\dot{\mathbf{q}}$ for $a\in\overline{1,l}$ in some sense under the priority relations.
So, we may expect 
\begin{equation}
\label{eqn:convergence_of_the_task_error}
\lim_{t\to\infty}\|\mathbf{e}_a(\mathbf{x}(t))\| = 0
\end{equation}
and find out conditions for \eqref{eqn:convergence_of_the_task_error}. 
However, demanding \eqref{eqn:convergence_of_the_task_error} is too restrictive in the general cases by the following reasons:
\begin{itemize}
	\item $\mathbf{p}_a(t)$ is not always located in $\mathbf{f}_a(t,\mathbb{R}^n)$, so that $\liminf_{t\to\infty}\inf_{\mathbf{q}'\in\mathbb{R}^n}\|\mathbf{e}_a(t,\mathbf{q}')\|>0$ is possible. 
	\item Even if $\mathbf{p}_a(t)\in\mathbf{f}_a(t,\mathbb{R}^n)$ for all $t\in[t_0,\infty)$, $\mathbf{q}(t)$ may converge to a singularity in which the $a$-th task loses DOF necessary for achieving \eqref{eqn:convergence_of_the_task_error}.
	\item Even if $\mathbf{p}_a(t)\in\mathbf{f}_a(t,\mathbb{R}^n)$ and $\mathrm{rank}(\mathbf{J}_a(\mathbf{x}(t))) = m_a$ for all $(a,t)\in\overline{1,l}\times[t_0,\infty)$, $\mathbf{q}(t)$ may converge to an algorithmic singularity in which there is a conflict between the $a$-th and $b$-th tasks in achieving both $\lim_{t\to\infty}\|\mathbf{e}_a(\mathbf{x}(t))\| = 0$ and $\lim_{t\to\infty}\|\mathbf{e}_b(\mathbf{x}(t))\| = 0$.
\end{itemize}
Therefore, we need to determine alternative convergence criteria instead of \eqref{eqn:convergence_of_the_task_error} that can be used in the general cases and find out conditions for those criteria.

Denote $\mathbf{A}_{ab} = \sum_{i=b}^a\mathbf{C}_{ai}\mathbf{C}_{ii}^T\mathbf{L}_{ib}$ and $\mathbf{b}_a = \dot{\mathbf{p}}_a - \mathbf{f}_{ta} - \sum_{b=1}^a\mathbf{A}_{ab}(\psi_b\dot{\mathbf{p}}_b - \mathbf{f}_{tb}) - \sum_{b=1}^{a-1}k_b\psi_b\mathbf{A}_{ab}\mathbf{e}_b$ for $1\le b\le a \le l$.
By differentiating $\mathbf{e}_a$ with respect to $t$, we can formulate the $a$-th error dynamics as
\begin{equation*}
	\dot{\mathbf{e}}_a + k_a\psi_a\mathbf{A}_{aa}\mathbf{e}_a = \mathbf{b}_a.
\end{equation*}
Define $\phi_a,\eta_a,\rho_a,\gamma_a:[t_0,\infty)\to\mathbb{R}$ for $a\in\overline{1,l}$ as 
\begin{align*}
	\phi_a(t) &= \|\mathbf{e}_a(\mathbf{x}(t))\| \\
	\eta_a(t) &= \|(\psi_a\mathbf{C}_{aa}^T\mathbf{L}_{aa}\mathbf{e}_a)(\mathbf{x}(t))\| \\ 
	\rho_a(t) &= k_a\psi_a(\mathbf{x}(t))\phi_a^{+2}(t)\langle \mathbf{e}_a(\mathbf{x}(t)),(\mathbf{A}_{aa}\mathbf{e}_a)(\mathbf{x}(t))\rangle \\
	\gamma_a(t) &= \phi_a^+(t)\langle \mathbf{e}_a(\mathbf{x}(t)),\mathbf{b}_a(\mathbf{x}(t))\rangle
\end{align*}
where $\phi_a^+(t) = 0$ if $\phi_a(t) = 0$ and $\phi_a^+(t) = 1/\phi_a(t)$ if $\phi_a(t)\neq 0$. 
We will need following assumptions:
\begin{description}
	\item[(A1)] $\mathbf{C}$ is bounded and $\mathbf{C}(\mathbf{x}(\cdot))$ is measurable in $[t_0,\infty)$;
	\item[(A2)] $\mathbf{L}$ is bounded and $\mathbf{L}(\mathbf{x}(\cdot))$ is measurable in $[t_0,\infty)$;
	\item[(A3)] $\mathbf{f}_{ta}$ and $\dot{\mathbf{p}}_a$ are bounded, $\int_{t_0}^\infty\|\mathbf{f}_{ta}(\mathbf{x}(t))\|dt <\infty$, and $\int_{t_0}^\infty\|\dot{\mathbf{p}}_a(t)\|dt <\infty$ for all $a\in\overline{1,l}$;
	\item[(A4)] $\mathbf{L}_{aa}(\mathbf{x}) = \mathbf{L}_{aa}^T(\mathbf{x})\ge 0$ for all $a\in\overline{1,l}$ and $\mathbf{x}\in X$;
	\item[(A5)] $\mathbf{C}_{aa}\mathbf{C}_{aa}^T\mathbf{L}_{aa} = \mathbf{L}_{aa}\mathbf{C}_{aa}\mathbf{C}_{aa}^T$ for all $a\in\overline{1,l}$;
	\item[(A6)] there exists $\mathbf{M}_{ab}:X\to\mathbb{R}^{m_a\times m_b}$ bounded and satisfying $\mathbf{C}_{aa}^T\mathbf{L}_{ab} = \mathbf{M}_{ab}\mathbf{C}_{bb}^T\mathbf{L}_{bb}$ for all $a,b\in\overline{1,l}$.
\end{description}
Note that if the trajectory existence is guaranteed by Theorem \ref{thm:existence_of_joint_trajectory}, then (A1) and (A2) are met.
One can easily verify that (A4) and (A5) imply that $(\mathbf{C}_{aa}\mathbf{C}_{aa}^T\mathbf{L}_{aa})(\mathbf{x})$ is symmetric and positive semidefinite and $(\mathbf{C}_{aa}\mathbf{C}_{aa}^T)^{1/2}\mathbf{L}_{aa}^{1/2} = \mathbf{L}_{aa}^{1/2}(\mathbf{C}_{aa}\mathbf{C}_{aa}^T)^{1/2}$ for all $a\in\overline{1,l}$ and $\mathbf{x}\in X$.

\begin{lemma}
	\label{lem:differential_inequality}
	If (A1) and (A2) hold, then for every $a\in\overline{1,l}$ and $t_0\le t_1 < t_2 < \infty$, $\phi_a$ is absolutely continuous on $[t_1,t_2]$; $\dot{\phi}_a(t) = -\rho_a(t)\phi_a(t) + \gamma_a(t)$ for almost all $t\in[t_1,t_2]$; and $\eta_a$, $\rho_a$, and $\gamma_a$ are integrable on $[t_1,t_2]$.
\end{lemma}
\begin{proof}
	Let $[t_0,\infty)$ be a metric space with a distance $d(t_1,t_2) = |t_1-t_2|$.
	Fix $a\in\overline{1,l}$ and $t_0\le t_1 < t_2 < \infty$. 
	$\phi_a$ is absolutely continuous on $[t_1,t_2]$ because $\mathbf{x}\in\mathrm{AC}([t_0,\infty),X)$ and $\|\mathbf{e}_a\|\in C^L(X,\mathbb{R})$.
	Let $t\in[t_1,t_2]$ be such that $\phi_a$ and $\mathbf{e}_a(\mathbf{x}(\cdot))$ are differentiable at $t$.
	If $\phi_a(t) = 0$, then $\dot{\phi}_a(t) = 0 = -\rho_a(t)\phi_a(t) + \gamma_a(t)$ because $\phi_a([t_0,\infty))\subset [0,\infty)$.
	If $\phi_a(t) > 0$, then $\dot{\phi}_a(t) = \phi_a^+(t)\langle \mathbf{e}_a(\mathbf{x}(t)),\dot{\mathbf{e}}_a(\mathbf{x}(t))\rangle = -\rho_a(t)\phi_a(t) + \gamma_a(t)$.
	Since $\dot{\mathbf{p}}_b$, $\mathbf{f}_{tb}(\mathbf{x}(\cdot))$, $\mathbf{e}_b(\mathbf{x}(\cdot))$, and $\psi_b(\mathbf{x}(\cdot))$ are continuous on the compact set $[t_1,t_2]$ for all $b\in\overline{1,a}$ and $\mathbf{C}$ and $\mathbf{L}$ are bounded, $\eta_a$, $\rho_a$, and $\gamma_a$ are bounded on $[t_1,t_2]$.

	Since $\phi_a$ is continuous on $[t_0,\infty)$ and $\phi_a(t) = 0$ if and only if $\phi_a^+(t) = 0$, $\phi_a^{-1}(0) = (\phi_a^+)^{-1}(0)$ is closed \cite[Corollary 4.8]{Rudin1964}, $I = [t_0,\infty)\setminus (\phi_a^+)^{-1}(0)$ is open, and $\phi_a^+$ is continuous at every $t\in I$. 
	Let $A\subset \mathbb{R}$ be an arbitrary open set.
	For every $t\in (\phi_a^+)^{-1}(A\setminus\{0\})\subset I$ there exists a neighborhood $N\subset[t_0,\infty)$ of $t$ satisfying $\phi_a^+(N) \subset A\setminus\{0\}$ because $A\setminus\{0\}$ is open and $\phi_a^+$ is continuous at $t$.
	Thus, $(\phi_a^+)^{-1}(A\setminus\{0\})$ is open.
	Since Borel sets are Lebesgue measurable, $(\phi_a^+)^{-1}(A) = (\phi_a^+)^{-1}(0) \cup (\phi_a^+)^{-1}(A\setminus\{0\})$ is measurable. 
	So, $\phi_a^+$ is measurable in $[t_0,\infty)$.
	It follows that $\eta_a$, $\rho_a$, and $\gamma_a$ are measurable in $[t_0,\infty)$ because those are continuous functions of measurable functions.
	Therefore, $\eta_a$, $\rho_a$, and $\gamma_a$ are integrable on $[t_1,t_2]$.
\end{proof}

\begin{theorem}
	\label{thm:task_convergence}
	Let $\mathsf{S}\in\mathbb{S}^\bullet$ be as in \eqref{eqn:task_convergence_kinematic_system}, $\mathbf{u}$ be a PIK solution of $[\mathsf{S}]^\bullet$ in the form of \eqref{eqn:class_of_pik_solutions}, and $\mathbf{q}:[t_0,\infty)\to\mathbb{R}^n$ be a Carath\'{e}odory solution of \eqref{eqn:differential_equation} with an initial value $\mathbf{q}(t_0) = \mathbf{q}_0$.
	Define $\eta_0:[t_0,\infty)\to\mathbb{R}$ as $\eta_0(t) = 0$.
	Assume (A1) to (A6).
	Then, for every $a\in\overline{1,l}$
	\begin{itemize}
		\item if $\sum_{b=0}^{a-1}\int_{t_0}^\infty\eta_b(t)dt <\infty$ or $\mathbf{C}_{ab} = \mathbf{L}_{ab} = \mathbf{0}$ for all $1\le b<a$, then $\int_{t_0}^\infty\eta_a^2(t)dt<\infty$;
		\item if additionally $\inf_{t\in[t_0,\infty)}\sigma_{\min}((\psi_a\mathbf{A}_{aa})(\mathbf{x}(t)))>0$ holds, then $\int_{t_0}^\infty\eta_a(t)dt<\infty$, $\int_{t_0}^\infty\phi_a(t)dt<\infty$, and $\lim_{t\to\infty}\eta_a(t) = \lim_{t\to\infty}\phi_a(t) = 0$.
	\end{itemize}
	If $\sum_{a=0}^l\int_{t_0}^\infty\eta_a(t)dt<\infty$ and $\mathbf{R}^{-1}$ is bounded, then $\int_{t_0}^\infty\|\mathbf{u}(\mathbf{x}(t))\|dt < \infty$.
\end{theorem}
\begin{proof}
	Construct $\mathbf{A} = [\mathbf{A}_{ij}]$ and $\mathbf{M} = [\mathbf{M}_{ij}]$ by letting $\mathbf{A}_{ab} = \mathbf{M}_{ab} = \mathbf{0}$ for $1\le a < b\le l$.
	By the assumptions, there exists $M\in[1,\infty)$ satisfying 
	\begin{equation*}
		\max\{\|\mathbf{C}(\mathbf{x})\|_F,\|\mathbf{L}(\mathbf{x})\|_F,\|\mathbf{A}(\mathbf{x})\|_F,\|\mathbf{M}(\mathbf{x})\|_F,\|\mathbf{K}\|_F\} \le M
	\end{equation*}
	for all $\mathbf{x}\in X$.
	By (A4) to (A6), we have $\mathbf{A}_{ab}\mathbf{e}_b = \sum_{i=b}^a\mathbf{C}_{ai}\mathbf{M}_{ib}\mathbf{C}_{bb}^T\mathbf{L}_{bb}\mathbf{e}_b$ and $\gamma_a \le (1+M)\sum_{b=1}^a(\|\dot{\mathbf{p}}_b\| + \|\mathbf{f}_{tb}\|) + aM^3\sum_{b=1}^{a-1}\eta_b$.
	If $\mathbf{C}_{ab} = \mathbf{L}_{ab} = \mathbf{0}$ for all $1\le b < a$, then $\mathbf{A}_{ab} = \mathbf{0}$ for all $1\le b<a$ and $\gamma_a \le (1+M)(\|\dot{\mathbf{p}}_a\|+\|\mathbf{f}_{ta}\|)$.
	
	Fix $a\in\overline{1,l}$ and assume $\sum_{b=0}^{a-1}\int_{t_0}^\infty\eta_b(t)dt<\infty$ or $\mathbf{C}_{ab} = \mathbf{L}_{ab} = \mathbf{0}$ for all $1\le b < a$.
	If $\sum_{b=0}^{a-1}\int_{t_0}^\infty\eta_b(t)dt<\infty$, then
	\begin{align*}
		\phi_a(t) &= \phi_a(t_0)e^{-\int_{t_0}^t\rho_a(s)ds} + \int_{t_0}^t\gamma_a(s)e^{-\int_s^t\rho_a(r)dr}ds \\
			&\le \phi_a(t_0) + (1+M)\sum_{b=1}^a\int_{t_0}^\infty (\|\dot{\mathbf{p}}_b\|+\|\mathbf{f}_{tb}\|)(\mathbf{x}(s))ds \\
			&\quad + aM^3\sum_{b=1}^{a-1}\int_{t_0}^\infty\eta_b(s)ds \\
			&= M_1 < \infty
	\end{align*}
	and
	\begin{align*}
		\int_{t_0}^t\eta_a^2(s)ds 
			&\le \frac{MM_1}{k_a}\int_{t_0}^t \rho_a(s)\phi_a(s)ds \\
			&\le \frac{MM_1}{k_a}\left(\int_{t_0}^t\gamma_a(s)ds + \phi_a(t) + \phi_a(t_0) \right) \\
			&\le \frac{2MM_1^2}{k_a} < \infty
	\end{align*}
	for all $t\in[t_0,\infty)$ by Lemma \ref{lem:differential_inequality}.
	If $\mathbf{C}_{ab} = \mathbf{L}_{ab} = \mathbf{0}$ for all $1\le b < a$, then $\int_{t_0}^\infty\eta_a^2(t)dt<\infty$ follows from
	\begin{equation*}
		\phi_a(t) \le \phi_a(t_0) + (1+M)\int_{t_0}^\infty (\|\dot{\mathbf{p}}_a\|+\|\mathbf{f}_{ta}\|)(\mathbf{x}(s))ds < \infty.
	\end{equation*}
	
	Assume $\sigma = \inf_{t\in[t_0,\infty)}\sigma_{\min}((\psi_a\mathbf{A}_{aa})(\mathbf{x}(t)))>0$ additionally.
	Then, $0 < \sqrt{k_a\sigma} \le \sqrt{\rho_a(t)}$,
	\begin{equation*}
		\int_{t_0}^t\eta_a(s)ds = \frac{1}{\sqrt{k_a}}\int_{t_0}^t \sqrt{\rho_a(s)}\phi_a(s)ds \le \frac{2M_1}{k_a\sqrt{\sigma}} < \infty,
	\end{equation*}
	and
	\begin{equation*}
		\int_{t_0}^t\phi_a(s)ds \le \frac{1}{\sqrt{\sigma}}\int_{t_0}^t\eta_a(s)ds \le \frac{2M_1}{k_a\sigma} <\infty
	\end{equation*}
	for all $t\in[t_0,\infty)$.
	We can find, similarly as before, that if $\sum_{b=0}^{a-1}\int_{t_0}^\infty\eta_b(t)dt < \infty$, then $\phi_b$ and $\eta_b$ are bounded on $[t_0,\infty)$ for all $b\in\overline{1,a}$.
	Therefore, there exists $L\in(0,\infty)$ satisfying $|\dot{\phi}_a(t)| \le \rho_a(t)\phi_a(t) + \gamma_a(t) \le L$ for all $t\in[t_0,\infty)$.
	Then, $\phi_a$ is Lipschitz on $[t_0,\infty)$ with the Lipschitz constant $L$.
	Suppose that there exists $\epsilon>0$ such that for every $T\in[t_0,\infty)$ there exists $t\ge T$ satisfying $\phi_a(t)\ge \epsilon$.
	Fix $T\in(0,\infty)$ and let $t_0\le t_1<t_2<\cdots$ satisfying $t_{i+1} - t_i \ge T$ and $\phi_a(t_i)\ge \epsilon$ for all $i\in\mathbb{N}$.
	Let $0<\delta\le\min\{\epsilon/(2L),T/2\}$.
	Then, $\phi_a(t) \ge \phi_a(t_i) - |\phi_a(t) - \phi_a(t_i)| \ge \phi_a(t_i) - L|t - t_i| \ge \epsilon/2$ for every $|t-t_i|\le \delta$.
	Thus, we find a contradiction
	\begin{equation*}
		\infty > \int_{t_0}^\infty \phi_a(t)dt \ge \sum_{i=1}^\infty\int_{t_i-\delta}^{t_i+\delta}\phi_a(t)dt \ge \sum_{i=1}^\infty\epsilon\delta = \infty.
	\end{equation*}
	Therefore, $\eta_a(t) \le \sqrt{M}\phi_a(t) \to 0$ as $t\to\infty$.
	
	Assume $\sum_{a=0}^l\int_{t_0}^\infty\eta_a(t)dt<\infty$ and $\sup_{\mathbf{x}\in X}\|\mathbf{R}^{-1}(\mathbf{x})\| = M_2<\infty$.
	Let $\mathbf{L}_D = \mathrm{diag}(\mathbf{L}_{11},\dots,\mathbf{L}_{ll})$.
	Then, we have $\mathbf{u} = \mathbf{R}^{-1}\hat{\mathbf{J}}^T\mathbf{M}\mathbf{C}_D^T\mathbf{L}_D(\boldsymbol{\Psi}(\dot{\mathbf{p}} + \mathbf{K}\mathbf{e}) - \mathbf{f}_t)$ and $\int_{t_0}^\infty\|\mathbf{u}(\mathbf{x}(t))\|dt \le M_2M^3\sum_{a=1}^l\int_{t_0}^\infty(\|\dot{\mathbf{p}}_a(t)\| + \|\mathbf{f}_{ta}(\mathbf{x}(t))\| + \eta_a(t))dt
	< \infty$.
\end{proof}

\begin{corollary}
	\label{cor:task_convergence}
	Let $\mathsf{S}\in \mathbb{S}^\bullet$ be as in \eqref{eqn:task_convergence_kinematic_system}, $\alpha\in\overline{1,4}$, $\mathbf{u}$ be the $\boldsymbol{\pi}_\alpha$-PIK solution of $[\mathsf{S}]^\bullet$ with the damping functions given by \eqref{eqn:damping_function}, and $\mathbf{q}:[t_0,\infty)\to\mathbb{R}^n$ be a Carath\'{e}odory solution of \eqref{eqn:differential_equation} with an initial value $\mathbf{q}(t_0) = \mathbf{q}_0$.
	Define $\eta_0:[t_0,\infty)\to\mathbb{R}$ as $\eta_0(t) = 0$.
	Assume $\mu_1,\dots,\mu_l\in(0,\infty)$ and $\nu\in\mathbb{N}\cup\{0\}$ if $\alpha\in\overline{1,3}$; $\mathbf{F}_q$ and $\mathbf{R}^{-1}$ are bounded; $\mathbf{C}(\mathbf{x}(\cdot))$ is measurable in $[t_0,\infty)$; (A3); and (A5) if $\alpha = 2$.
	Then, for every $a\in\overline{1,l}$
	\begin{itemize}
		\item if $\sum_{b=0}^{a-1}\int_{t_0}^\infty\eta_b(t)dt <\infty$ or $\mathbf{C}_{ab} = \mathbf{0}$ for all $1\le b<a$, then $\int_{t_0}^\infty\eta_a^2(t)dt<\infty$;
		\item if additionally $\inf_{t\in[t_0,\infty)}\sigma_{\min}((\psi_a\mathbf{C}_{aa})(\mathbf{x}(t)))>0$ holds, then $\int_{t_0}^\infty\eta_a(t)dt<\infty$, $\int_{t_0}^\infty\phi_a(t)dt<\infty$, and $\lim_{t\to\infty}\eta_a(t) = \lim_{t\to\infty}\phi_a(t) = 0$.
	\end{itemize}
	If $\sum_{a=0}^l\int_{t_0}^\infty\eta_a(t)dt<\infty$ and $\mathbf{R}^{-1}$ is bounded, then $\int_{t_0}^\infty\|\mathbf{u}(\mathbf{x}(t))\|dt < \infty$.
\end{corollary}
\begin{proof}
	Since $\mathbf{F}_q$ and $\mathbf{R}^{-1}$ are bounded, $\mathbf{J} = \mathbf{F}_q\mathbf{R}^{-1}$ and $\mathbf{C} = \mathbf{J}\hat{\mathbf{J}}^T$ are bounded.
	We showed in the proof of Corollary \ref{cor:existence_of_krasovskii_solution} that $\mathbf{L}$ is bounded.
	We show that $\mathbf{L}(\mathbf{x}(\cdot))$ is measurable in $[t_0,\infty)$.
	Define $\mathbf{G}_1:\mathbb{R}^{m_a\times m_a}\to\mathbb{R}^{m_a\times m_a}$ and $\mathbf{G}_2:\mathbb{R}^{m_a\times m_a}\times\mathbb{R}^{m_a\times n}\to\mathbb{R}^{m_a\times m_a}$ as $\mathbf{G}_1(\mathbf{X}) = |\mathbf{X}\mathbf{X}^T|^\nu(|\mathbf{X}\mathbf{X}^T|^\nu\mathbf{X}\mathbf{X}^T + \mu_a^2\mathbf{I}_{m_a})^{-1}$ and $\mathbf{G}_2(\mathbf{X},\mathbf{Y}) = |\mathbf{X}\mathbf{X}^T|^\nu(|\mathbf{X}\mathbf{X}^T|^\nu\mathbf{Y}\mathbf{Y}^T + \mu_a^2\mathbf{I}_{m_a})^{-1}$. 
	$\mathbf{G}_1$ and $\mathbf{G}_2$ are continuous on $\mathbb{R}^{m_a\times m_a}$ and $\mathbb{R}^{m_a\times m_a}\times\mathbb{R}^{m_a\times n}$, respectively, because $|\mathbf{X}\mathbf{X}^T|$ can be written as a polynomial of entries of $\mathbf{X}$ and $\mathrm{rank}(|\mathbf{X}\mathbf{X}^T|^\nu\mathbf{X}\mathbf{X}^T + \mu_a^2\mathbf{I}_{m_a}) = \mathrm{rank}(|\mathbf{X}\mathbf{X}^T|^\nu\mathbf{Y}\mathbf{Y}^T + \mu_a^2\mathbf{I}_{m_a}) = m_a$ for all $(\mathbf{X},\mathbf{Y})\in\mathbb{R}^{m_a\times m_a}\times\mathbb{R}^{m_a\times n}$. 
	Since $\mathbf{C}(\mathbf{x}(\cdot))$ is measurable in $[t_0,\infty)$ and $\mathbf{J}$ is continuous on $X$, $\mathbf{D}_a(\mathbf{x}(\cdot)) = \mathbf{G}_1(\mathbf{C}_{aa}(\mathbf{x}(\cdot)))$, $\mathbf{H}_a(\mathbf{x}(\cdot)) = \mathbf{G}_2(\mathbf{C}_{aa}(\mathbf{x}(\cdot)),\mathbf{J}_a(\mathbf{x}(\cdot)))$, and $\mathbf{C}_{aa}^*(\mathbf{x}(\cdot)) = (\mathbf{C}_{aa}^T\mathbf{D}_a)(\mathbf{x}(\cdot))$ are measurable in $[t_0,\infty)$ \cite[Theorem 1.7, Theorem 1.8, Exercises 1.3]{Rudin1987}. 
	If follows that $\mathbf{D}(\mathbf{x}(\cdot))$, $\mathbf{H}(\mathbf{x}(\cdot))$, and $(\mathbf{D}(\mathbf{I}_m + \mathbf{C}_L\mathbf{C}_D^\circledast)^{-1})(\mathbf{x}(\cdot)) = (\mathbf{D}(\mathbf{I}_m - \mathbf{C}_L\mathbf{C}_D^\circledast + \cdots + (-\mathbf{C}_L\mathbf{C}_D^\circledast)^{l-1}))(\mathbf{x}(\cdot))$ are measurable in $[t_0,\infty)$.
	Therefore, (A1) and (A2) hold.
	Since $\mathbf{D}_a(\mathbf{x}) = \mathbf{D}_a^T(\mathbf{x})\ge 0$ and $\mathbf{H}_a(\mathbf{x}) = \mathbf{H}_a^T(\mathbf{x})\ge0$ for all $a\in\overline{1,l}$ and $\mathbf{x}\in X$, (A4) holds for all $\alpha\in\overline{1,4}$.
	Since $\mathbf{C}_{aa}\mathbf{C}_{aa}^T\mathbf{D}_a = \mathbf{C}_{aa}\mathbf{C}_{aa}^* = \mathbf{D}_a\mathbf{C}_{aa}\mathbf{C}_{aa}^T$, (A5) holds for all $\alpha\in\{1,3,4\}$.
	It is obvious that (A6) holds for $\alpha\in\overline{2,4}$ because $\mathbf{L}$ is block diagonal.
	If $\alpha = 1$, (A6) follows from
	\begin{align*}
	\mathbf{C}_D^T\mathbf{L} 
	&= \mathbf{C}_D^\circledast(\mathbf{I}_m - \mathbf{C}_L\mathbf{C}_D^\circledast + \cdots + (-\mathbf{C}_L\mathbf{C}_D^\circledast)^{l-1}) \\
	&= (\mathbf{I}_m - \mathbf{C}_D^\circledast\mathbf{C}_L + \cdots + (-\mathbf{C}_D^\circledast\mathbf{C}_L)^{l-1})\mathbf{C}_D^\circledast \\
	&= (\mathbf{I}_m + \mathbf{C}_D^\circledast\mathbf{C}_L)^{-1}\mathbf{C}_D^T\mathbf{D}.
	\end{align*}
	One can easily check that $\mathbf{C}_{ab} = \mathbf{0}$ for all $1\le b < a$ implies $\mathbf{L}_{ab} = \mathbf{0}$ for all $1\le b<a$ from the above equation.
	Assume $\inf_{t\in[t_0,\infty)}\sigma_{\min}(\mathbf{C}_{aa}(\mathbf{x}(t))) \ge \inf_{t\in[t_0,\infty)}\sigma_{\min}((\psi_a\mathbf{C}_{aa})(\mathbf{x}(t))) = \sigma > 0$.
	Then, $|(\mathbf{C}_{aa}\mathbf{C}_{aa}^T)(\mathbf{x}(t))|^\nu = \prod_{i=1}^{m_a}\sigma_i^{2\nu}(\mathbf{C}_{aa}(\mathbf{x}(t))) \ge \sigma^{2m_a\nu}$ and $\lambda_a^2(\mathbf{x}(t)) \le \mu_a^2/\sigma^{2m_a\nu}$.
	Let $M = \sup_{\mathbf{x}\in X}\|\mathbf{J}(\mathbf{x})\|_F<\infty$.
	Then, $\sigma_{\min}(\mathbf{D}_a(\mathbf{x}(t))) = (\|\mathbf{C}_{aa}\|^2 + \lambda_a^2)^{-1}(\mathbf{x}(t)) \ge M_1 = (M^2 + \mu_a^2/\sigma^{2m_a\nu})^{-1}$ and $\sigma_{\min}(\mathbf{H}_a(\mathbf{x}(t))) = (\|\mathbf{J}_a\|^2 + \lambda_a^2)^{-1}(\mathbf{x}(t)) \ge M_1$.
	By \cite{Merikoski2004}, $\sigma_{\min}((\psi_a\mathbf{A}_{aa})(\mathbf{x}(t))) \ge \sigma_{\min}^2((\psi_a\mathbf{C}_{aa})(\mathbf{x}(t)))\sigma_{\min}(\mathbf{L}_{aa}(\mathbf{x}(t)))\ge \sigma^2M_1 > 0$
	for all $t\in[t_0,\infty)$.
	The proof is completed by Theorem \ref{thm:task_convergence}.
\end{proof}

\begin{lemma}
	\label{lem:joint_convergence}
	If $\int_{t_0}^\infty\|\mathbf{u}(\mathbf{x}(t))\|dt < \infty$, then there exists $\mathbf{q}_\infty\in\mathbb{R}^n$ such that $\lim_{t\to\infty}\|\mathbf{q}(t) - \mathbf{q}_\infty\| = 0$.
	If additionally there exists $t_\infty\in[t_0,\infty)$ such that $\mathbf{u}$ is continuous at $(t_\infty,\mathbf{q}_\infty)$ and $\mathbf{u}(t,\cdot) = \mathbf{u}(t_\infty,\cdot)$ for all $t\in[t_\infty,\infty)$, then $\lim_{t\to\infty}\|\mathbf{u}(\mathbf{x}(t))\| = \|\mathbf{u}(t_\infty,\mathbf{q}_\infty)\| = 0$.
\end{lemma}
\begin{proof}
	Assume $u_\infty = \int_{t_0}^\infty\|\mathbf{u}(\mathbf{x}(t))\|dt < \infty$.
	Let $t_1,t_2,t_3,\dots$ be a divergent sequence in $[t_0,\infty)$.
	Define $u_i = \int_{t_0}^{t_i}\|\mathbf{u}(\mathbf{x}(t))\|dt$ for $i = 1,2,\dots$.
	Let $\epsilon>0$ be arbitrary.
	Since $\lim_{i\to\infty}u_i = u_\infty$, there exists $N\in\mathbb{N}$ such that $|u_i - u_\infty| < \epsilon/2$ for all $i > N$.
	Then, for every $i,j\in\mathbb{N}\setminus\overline{1,N}$ we have $\|\mathbf{q}(t_i) - \mathbf{q}(t_j)\| \le \left|\int_{t_i}^{t_j}\|\mathbf{u}(t,\mathbf{q}(t))\|dt\right| = |u_i - u_j| \le |u_i - u_\infty| + |u_j - u_\infty| < \epsilon$.
	Therefore, $\{\mathbf{q}(t_i)\}$ converges in $\mathbb{R}^n$ because it is a Cauchy sequence \cite[Theorem 3.11]{Rudin1964}.
	Since it holds for every divergnet sequence $\{t_i\}_{i=1}^\infty$, $\mathbf{q}(t)$ converges to a point $\mathbf{q}_\infty\in\mathbb{R}^n$.
	Assume additionally that there exists $t_\infty\in[t_0,\infty)$ such that $\mathbf{u}\in C_{(t_\infty,\mathbf{q}_\infty)}^0$ and $\mathbf{u}(t,\cdot) = \mathbf{u}(t_\infty,\cdot)$ for all $t\in[t_\infty,\infty)$.
	Since $\lim_{t\to\infty}\|\mathbf{q}(t) - \mathbf{q}_\infty\| = 0$, we have $\lim_{t\to\infty}\|\mathbf{u}(t,\mathbf{q}(t)) - \mathbf{u}(t_\infty,\mathbf{q}_\infty)\| = \lim_{\substack{t\to\infty\\t\in[t_\infty,\infty)}}\|\mathbf{u}(t_\infty,\mathbf{q}(t)) - \mathbf{u}(t_\infty,\mathbf{q}_\infty)\| = 0$.
	Suppose $c = \|\mathbf{u}(t_\infty,\mathbf{q}_\infty)\| > 0$.
	Then, there exists $T\in[t_\infty,\infty)$ satisfying $\|\mathbf{u}(t,\mathbf{q}(t))\| \ge c/2$ for all $t>T$.
	It follows that $\infty>\int_{t_0}^\infty\|\mathbf{u}(t,\mathbf{q}(t))\|dt \ge \int_T^\infty\|\mathbf{u}(t,\mathbf{q}(t))\|dt \ge \int_T^\infty c/2dt = \infty$, a contradiction.
	Therefore, $\lim_{t\to\infty}\|\mathbf{u}(\mathbf{x}(t))\| = \|\mathbf{u}(t_\infty,\mathbf{q}_\infty)\| = 0$.
\end{proof}

\begin{remark}
	A practically useful result we can get from Theorem \ref{thm:task_convergence} and Corollary \ref{cor:task_convergence} is that if the assumption $\inf_{t\in[t_0,\infty)}\sigma_{\min}((\psi_a\mathbf{A}_{aa})(\mathbf{x}(t)))>0$ holds for all $a\in\overline{1,l}$, then all task errors converge to zero and also the joint trajectory converges to a point in $\mathbb{R}^n$.
	However, we will need an extra work to find conditions on the desired task trajectory $\mathbf{p}$, the initial value $(t_0,\mathbf{q}_0)$, the feedback gain matrix $\mathbf{K}$, and the activation function $\boldsymbol{\Psi}$ in order to guarantee that assumption.
	It would be a meaningful work to find such conditions for the practical applications, but in this paper we rather show in Section \ref{sec:example} that we can still analyze the task convergence in the general case that the joint trajectory converges to or passes through singularity.
\end{remark}

\section{Stability}
\label{sec:stability}

In some practical cases, a kinematic system is given as
\begin{equation}
	\label{eqn:stability_kinematic_system}
	\mathsf{S} = (l,\mathbf{m},n,\mathsf{D}\mathbf{f},\mathbf{R},\boldsymbol{\Psi}\mathbf{K}(\mathbf{p}-\mathbf{f}))\in\mathbb{S}^\bullet
\end{equation}
that satisfies $\mathbf{f}(t,\cdot) = \mathbf{f}(0,\cdot)$ and $\mathbf{R}(t,\cdot) = \mathbf{R}(0,\cdot)$ for all $t\in\mathbb{R}$ where $\mathbf{p} = (\mathbf{p}_1,\dots,\mathbf{p}_l)\in\mathbb{R}^m$ is the point for the task position $\mathbf{f}(t,\mathbf{q})$ to be desired to reach, $\mathbf{K} = \mathrm{diag}(k_1\mathbf{I}_{m_1},\dots,k_l\mathbf{I}_{m_l})\in\mathbb{R}^{m\times m}$ with $k_a\in(0,\infty)$ is the feedback gain matrix, and $\boldsymbol{\Psi} = \mathrm{diag}(\psi_1\mathbf{I}_{m_1},\dots,\psi_l\mathbf{I}_{m_l})\in C^\bullet(X,\mathbb{R}^{m\times m})$ with $\psi_a(t,\mathbf{q}) = \psi_a(0,\mathbf{q})\in[0,1]$ for all $(a,\mathbf{x})$ is the activation function that can be used to activate or deactivate the term $\mathbf{K}(\mathbf{p}-\mathbf{f})$.
Let $\mathbf{u}$ be a PIK solution of $[\mathsf{S}]^\bullet$ in the form of \eqref{eqn:class_of_pik_solutions} satisfying $\mathbf{L}(t,\cdot) = \mathbf{L}(0,\cdot)$ for all $t$.
For the sake of simplicity in the notation, we may write $\mathbf{u}(\mathbf{q}) = \mathbf{u}(t,\mathbf{q})$ and other functions too.
In this section, we study stability of the autonomous system
\begin{equation}
\label{eqn:autonomous_system}
\dot{\mathbf{q}} = \mathbf{u}(\mathbf{q}).
\end{equation}
Define $S(\mathbf{q}_0)$ as the set of all Carath\'{e}odory solutions $\mathbf{q} \in \mathrm{AC}([0,\infty),\mathbb{R}^n)$ of \eqref{eqn:autonomous_system} with the initial value $\mathbf{q}(0) = \mathbf{q}_0$.
There are various notions of stability.
An equilibrium point $\mathbf{q}_\infty\in\mathbf{u}^{-1}(\mathbf{0}) = \{\mathbf{q}'\in\mathbb{R}^n\mid\mathbf{u}(\mathbf{q}') = \mathbf{0}\}$ is said to be
\begin{itemize}
	\item \textit{(Lyapunov) stable} if for every $\epsilon>0$ there exists $\delta>0$ such that for every $\mathbf{q}_0\in\mathbf{q}_\infty + \delta B_n$, $\mathbf{q}\in S(\mathbf{q}_0)\neq\emptyset$, and $t\in[0,\infty)$ we have $\|\mathbf{q}(t) - \mathbf{q}_\infty\| < \epsilon$;
	\item \textit{semistable} if $\mathbf{q}_\infty$ is stable and there exists $\delta>0$ such that for every $\mathbf{q}_0\in\mathbf{q}_\infty+\delta B_n$ and $\mathbf{q}\in S(\mathbf{q}_0)\neq\emptyset$ there exists a stable equilibrium point $\mathbf{q}_\infty'\in\mathbf{u}^{-1}(\mathbf{0})$ satisfying $\lim_{t\to\infty}\|\mathbf{q}(t)-\mathbf{q}_\infty'\| = 0$;
	\item \textit{asymptotically stable} if $\mathbf{q}_\infty$ is stable and there exists $\delta>0$ such that for every $\mathbf{q}_0\in\mathbf{q}_\infty+\delta B_n$ and $\mathbf{q}\in S(\mathbf{q}_0)\neq\emptyset$ we have $\lim_{t\to\infty}\|\mathbf{q}(t) - \mathbf{q}_\infty\| = 0$.
\end{itemize}
Note that the definition of stability includes existence of Carath\'{e}odory solutions in the vicinity of the equilibrium point.
A motivation of introducing semistability is to handle continuum of equilibria \cite{Hui2008}\cite{Hui2009}.
If $m<n$, then $\mathbf{f}^{-1}(\mathbf{p})\subset\mathbf{u}^{-1}(\mathbf{0})$ might form a continuum of equilibruim points such that any  $\mathbf{q}_\infty\in\mathbf{f}^{-1}(\mathbf{p})$ is not asymptotically stable.
If $\mathbf{q}_\infty\in\mathbf{f}^{-1}(\mathbf{p})$ is semistable, then we can guarantee that every joint trajectory starting from a certain neighborhood of $\mathbf{q}_\infty$ will stay in the vicinity of $\mathbf{q}_\infty$ and converge to a stable equilibrium point $\mathbf{q}_\infty'\in\mathbf{u}^{-1}(\mathbf{0})$, while if $\mathbf{q}_\infty$ is only stable, then there could be endless joint motions such as peoriodic motions.
If $\mathbf{p}_\infty$ is an isolated point of $\mathbf{u}^{-1}(\mathbf{0})$, then semistability coinsides with asymptotic stability.
Define $\mathcal{H}(\mathsf{S}) = \mathcal{H}_\mathsf{S} = \{\mathbf{q}\in\mathbb{R}^n\mid\mathrm{rank}(\mathbf{J}(\mathbf{q})) = m\}$.
Since $\mathbf{J}\in C^0$, $\mathcal{H}_\mathsf{S}$ is open.

\begin{theorem}
	\label{thm:stability}
	Let $\mathsf{S}\in\mathbb{S}^\bullet$ be as in \eqref{eqn:stability_kinematic_system} and $\mathbf{u}$ be a PIK solution of $[\mathsf{S}]^\bullet$ in the form of \eqref{eqn:class_of_pik_solutions}.
	Assume (A4), (A5), $\mathbf{L}\in C_{\mathcal{H}(\mathsf{S})}^0$, $\mathbf{q}_\infty\in\mathbf{f}^{-1}(\mathbf{p})$, and $\mathrm{rank}((\mathbf{C}_D^T\mathbf{L}\boldsymbol{\Psi})(\mathbf{q}_\infty)) = m$.
	Then, the equilibrium point $\mathbf{q}_\infty$ of \eqref{eqn:autonomous_system} is semistable. If $m = n$, then $\mathbf{q}_\infty$ is asymptotically stable.
\end{theorem}
\begin{proof}
	We first prove that $\mathbf{q}_\infty$ is stable by contradiction.
	Suppose that there exists $\epsilon_1>0$ such that for every $\delta>0$ there exists $\mathbf{q}_0\in\mathbf{q}_\infty + \delta B_n$ such that either $S(\mathbf{q}_0) = \emptyset$ or there exists $\mathbf{q}\in S(\mathbf{q}_0) \neq\emptyset$ and $T\in[0,\infty)$ satisfying $\|\mathbf{q}(T) - \mathbf{q}_\infty\| \ge \epsilon_1$.
	Since $\mathrm{rank}((\mathbf{C}_D^T\mathbf{L}\boldsymbol{\Psi})(\mathbf{q}_\infty)) = m$ and $\mathbf{C}_D^T\mathbf{L}\boldsymbol{\Psi}\in C_{\mathcal{H}(\mathsf{S})}^0$, there exists $\epsilon_2>0$ satisfying $\mathrm{rank}((\mathbf{C}_D^T\mathbf{L}\boldsymbol{\Psi})(\mathbf{q})) = m$ for all $\mathbf{q}\in\mathbf{q}_\infty+\epsilon_2B_n$ such that $\mathbf{u}\in C_{\mathbf{q}_\infty+\epsilon_2B_n}^0$.
	Thus, for every $\mathbf{q}_0\in\mathbf{q}_\infty+\epsilon_2B_n$ either $S(\mathbf{q}_0) \neq \emptyset$ or there exists $T\in[0,\infty)$ and $\mathbf{q}\in C_{(0,T)}^1([0,T],\mathbb{R}^n)$ satisfying $\mathbf{q}(0) = \mathbf{q}_0$, $\|\mathbf{q}(T) - \mathbf{q}_0\| = \epsilon_2$, and $\dot{\mathbf{q}}(t) = \mathbf{u}(\mathbf{q}(t))$ for all $t\in(0,T)$.
	Let $\epsilon_0=\min\{\epsilon_1,\epsilon_2\} > \delta_1 > \delta_2 > \cdots > 0$ be such that $\delta_i\to0$ as $i\to\infty$.
	Then, for every $i\in\mathbb{N}$ there exists $t_i\in(0,\infty)$ and $\mathbf{q}_i\in C_{(0,t_i)}^1([0,t_i],\mathbb{R}^n)$ satisfying $\|\mathbf{q}_i(0) - \mathbf{q}_\infty\|\le\delta_i$, $\|\mathbf{q}_i(t) - \mathbf{q}_\infty\|<\epsilon_0$ for all $t\in[0,t_i)$, $\|\mathbf{q}_i(t_i) - \mathbf{q}_\infty\| = \epsilon_0$, and $\dot{\mathbf{q}}_i(t) = \mathbf{u}(\mathbf{q}_i(t))$ for all $t\in(0,t_i)$.
	
	Let $\mathbf{P} = \mathrm{diag}(p_1\mathbf{I}_{m_1},\dots,p_l\mathbf{I}_{m_l})\in\mathbb{R}^{m\times m}$ be arbitrary and $\mathbf{M} = [\mathbf{M}_{ij}] = \mathbf{C}\mathbf{C}_D^T\mathbf{L}\boldsymbol{\Psi}\mathbf{K}$ where $\mathbf{M}_{ab}:\mathbb{R}^n\to\mathbb{R}^{m_a\times m_b}$ is the $(a,b)$-th block of $\mathbf{M}$ for $a,b\in\overline{1,l}$.
	By the assumptions, $\mathbf{M}_{aa}(\mathbf{q}) = k_a(\psi_a\mathbf{C}_{aa}\mathbf{C}_{aa}^T\mathbf{L}_{aa})(\mathbf{q}) = \mathbf{M}_{aa}^T(\mathbf{q}) > 0$ for all $a\in\overline{1,l}$ and $\mathbf{q}\in\mathbf{q}_\infty+\epsilon_0B_n$.
	Since $\mathbf{M}_{ab}\in C_{\mathbf{q}_\infty+\epsilon_0B_n}^0$ for all $a,b\in\overline{1,l}$, there exist $\phi_{aa} = \min\{\sigma_{\min}(\mathbf{M}_{aa}(\mathbf{q}))\mid\mathbf{q}\in\mathbf{q}_\infty+\epsilon_0B_n\} \in (0,\infty)$ for $a\in\overline{1,l}$ and $\phi_{ab} = \frac{1}{2}\max\{\sigma_{\max}(\mathbf{M}_{ab}(\mathbf{q}))\mid \mathbf{q}\in\mathbf{q}_\infty+\epsilon_0B_n\}\in[0,\infty)$ for $1\le b<a\le l$.
	Define $\mathbf{Q} = [q_{ij}] \in\mathbb{R}^{l\times l}$ as $q_{aa} = p_a\phi_{aa}$ for $a\in\overline{1,l}$ and $q_{ab} = q_{ba} = -p_b\phi_{ba}$ for $1\le a < b\le l$.
	The symmetric matrix $\mathbf{Q}$ is positive definite if and only if there exists a lower triangular matrix $\mathbf{X} = [x_{ij}]\in\mathbb{R}^{l\times l}$ with positive diagonals such that $\mathbf{Q} = \mathbf{X}\mathbf{X}^T$ \cite[Corollary 7.2.9]{Horn2013}.
	By comparing entries of $\mathbf{Q} = \mathbf{X}\mathbf{X}^T$, we can find $\mathbf{X}$ as
	\begin{equation*}
		x_{aa} = \left( p_a\phi_{aa} - \sum_{b=1}^{a-1}x_{ab}^2\right)^{1/2}
	\end{equation*}
	and
	\begin{equation*}
		x_{ab} = -\frac{1}{x_{bb}}\left( p_a\phi_{ab} + \sum_{i=1}^{b-1}x_{bi}x_{ai} \right), \quad a\in\overline{b+1,l}
	\end{equation*}
	under the condition $p_1>0$ and $p_a > \sum_{b=1}^{a-1}x_{ab}^2/\phi_{aa}$ for $a\in\overline{2,l}$.
	Fix $p_1,\dots,p_l\in(0,\infty)$ such that $\mathbf{Q}$ is positive definite.
	
	Define $V:\mathbb{R}^n\to[0,\infty)$ as
	\begin{equation*}
		V(\mathbf{q}) = \frac{1}{2}\langle\mathbf{e}(\mathbf{q}),\mathbf{P}\mathbf{e}(\mathbf{q})\rangle = \frac{1}{2}\sum_{a=1}^lp_a\|\mathbf{e}_a(\mathbf{q})\|^2
	\end{equation*}
	where $\mathbf{e}_a = \mathbf{p}_a - \mathbf{f}_a$.
	Let $\rho_1 = \min\{p_1,\dots,p_l\}$, $\rho_2 = \max\{p_1,\dots,p_l\}$, $\rho_3 = \sigma_{\min}(\mathbf{Q})$, and $\rho=\rho_3/(2\rho_2)$.
	Then, $2\rho_1\|\mathbf{e}(\mathbf{q})\|^2 \le V(\mathbf{q}) \le 2\rho_2\|\mathbf{e}(\mathbf{q})\|^2$ for all $\mathbf{q}\in\mathbb{R}^n$ and
	\begin{align*}
		\dot{V}(\mathbf{q}_i(t)) &= -\langle\mathbf{e}(\mathbf{q}_i(t)), \mathbf{P}\mathbf{M}(\mathbf{q}_i(t))\mathbf{e}(\mathbf{q}_i(t))\rangle \\
			&\le -\left\langle \begin{bmatrix} \|\mathbf{e}_1(\mathbf{q}_i(t))\| \\ \vdots \\ \|\mathbf{e}_l(\mathbf{q}_i(t))\| \end{bmatrix}, \mathbf{Q}\begin{bmatrix} \|\mathbf{e}_1(\mathbf{q}_i(t))\| \\ \vdots \\ \|\mathbf{e}_l(\mathbf{q}_i(t))\| \end{bmatrix} \right\rangle \\
			&\le -\rho_3\sum_{a=1}^l\|\mathbf{e}_a(\mathbf{q}_i(t))\|^2 \\
			&\le -\rho V(\mathbf{q}_i(t))
	\end{align*}
	for all $i\in\mathbb{N}$ and $t\in(0,t_i)$.
	By the Gronwall's inequality, $V(\mathbf{q}_i(t))\le V(\mathbf{q}_i(0))e^{-\rho t}$ for all $i\in\mathbb{N}$ and $t\in[0,t_i]$.
	Since $\mathsf{D}\mathbf{f}$ and $\mathbf{R}^{-1}(\mathbf{C}_D\hat{\mathbf{J}})^T\mathbf{L}\boldsymbol{\Psi}\mathbf{K}$ are continuous on $\mathbf{q}_\infty+\epsilon_0B_n$, there exist $L,M\in[0,\infty)$ satisfying
	\begin{align}
		\|\mathbf{f}(\mathbf{q}) - \mathbf{f}(\mathbf{q}_\infty)\| &\le L\|\mathbf{q} - \mathbf{q}_\infty\| \label{eqn:theorem_stability_lipschitz} \\
		\|(\mathbf{R}^{-1}\hat{\mathbf{J}}^T\mathbf{C}_D^T\mathbf{L}\boldsymbol{\Psi}\mathbf{K})(\mathbf{q})\| &\le M \label{eqn:theorem_stability_boundedness}
	\end{align}
	for all $\mathbf{q}\in\mathbf{q}_\infty+\epsilon_0B_n$.
	Then, we can derive
	\begin{align*}
		\epsilon_0 
			&\le \|\mathbf{q}_i(t_i) - \mathbf{q}_i(0)\| + \|\mathbf{q}_i(0) - \mathbf{q}_\infty\| \\
			&\le \int_0^{t_i}\|\mathbf{u}(\mathbf{q}_i(t))\|dt + \delta_i \\
			&\le M\sqrt{\frac{V(\mathbf{q}_i(0))}{2\rho_1}}\int_0^{t_i}e^{-\rho t/2}dt + \delta_i \\
			&\le \frac{2M}{\rho}\sqrt{\frac{\rho_2}{\rho_1}}\|\mathbf{f}(\mathbf{q}_i(0)) - \mathbf{f}(\mathbf{q}_\infty)\| + \delta_i \\
			&\le \left(1 + \frac{2LM}{\rho}\sqrt{\frac{\rho_2}{\rho_1}}\right)\delta_i
	\end{align*}
	for all $i\in\mathbb{N}$.
	Since $\lim_{i\to\infty}\delta_i = 0$, there exists $N\in\mathbb{N}$ such that
	\begin{equation*}
	\epsilon_0 \le \left(1 + \frac{2LM}{\rho}\sqrt{\frac{\rho_2}{\rho_1}}\right)\delta_i < \epsilon_0
	\end{equation*}
	for all $i>N$, a contradiction.
	Therefore, $\mathbf{q}_\infty$ is stable.
	
	We prove that $\mathbf{q}_\infty$ is semistable.
	Let $\epsilon\in(0,\infty)$ be such that $\mathrm{rank}((\mathbf{C}_D^T\mathbf{L}\boldsymbol{\Psi})(\mathbf{q})) = m$ for all $\mathbf{q}\in\mathbf{q}_\infty+\epsilon B_n$.
	Since $\mathbf{q}_\infty$ is stable, there exists $\delta>0$ such that $\|\mathbf{q}(t) - \mathbf{q}_\infty\|<\epsilon$ for all $\mathbf{q}_0\in\mathbf{q}_\infty + \delta B_n$, $\mathbf{q}\in S(\mathbf{q}_0) \neq \emptyset$, and $t\in[0,\infty)$.
	Fix $\mathbf{q}_0\in\mathbf{q}_\infty + \delta B_n$ and $\mathbf{q}\in S(\mathbf{q}_0)$.
	Let $0\le t_1<t_2<\cdots$ be an arbitrary divergent sequence and $\mathbf{q}_i = \mathbf{q}(t_i)$.
	There exist $L,M\in[0,\infty)$ satisfying \eqref{eqn:theorem_stability_lipschitz} and \eqref{eqn:theorem_stability_boundedness} on $\mathbf{q}_\infty + \epsilon B_n$.
	Then,
	\begin{align*}
	\|\mathbf{q}_i - \mathbf{q}_j\| &\le \left|\int_{t_i}^{t_j}\|\mathbf{u}(\mathbf{q}(t))\|dt\right| \\
	&\le M\left|\int_{t_i}^{t_j}\|\mathbf{e}(\mathbf{q}(t))\|dt\right| \\
	&\le M\sqrt{\frac{V(\mathbf{q}(0))}{2\rho_1}}\left|\int_{t_i}^{t_j}e^{-\rho t/2}dt\right| \\
	&\le \frac{2\delta LM}{\rho}\sqrt{\frac{\rho_2}{\rho_1}}(e^{-\rho t_i/2} + e^{-\rho t_j/2})
	\end{align*}
	for all $i,j\in\mathbb{N}$.
	For all $\epsilon'>0$ there exists $N\in\mathbb{N}$ such that $\|\mathbf{q}_i-\mathbf{q}_j\|<\epsilon'$ if $i,j>N$.
	So, $\{\mathbf{q}_i\}$ is Cauchy and converges to a point in $\mathbf{q}_0+\epsilon B_n$.
	Since it holds for an arbitrary divergent sequence $\{t_i\}$, $\mathbf{q}(t)$ converges to a point $\mathbf{q}_\infty' \in \mathbf{q}_\infty + \epsilon B_n$.
	Since $\mathbf{f}^{-1}(\mathbf{p})$ is closed and
	\begin{equation*}
		\lim_{t\to\infty}\|\mathbf{p} - \mathbf{f}(\mathbf{q}(t))\| 
			\le \sqrt{\frac{V(\mathbf{q}(0))}{2\rho_1}}\lim_{t\to\infty}e^{-\rho t/2} 
			= 0,
	\end{equation*}
	we have $\mathbf{q}_\infty'\in\mathbf{f}^{-1}(\mathbf{p})$ and $\mathrm{rank}((\mathbf{C}_D^T\mathbf{L}\boldsymbol{\Psi})(\mathbf{q}_\infty')) = m$.
	By the first part of the proof, we see that $\mathbf{q}_\infty'$ is a stable equilibrium point.
	Therefore, $\mathbf{q}_\infty$ is semistable.
	
	If $m = n$, then $\mathbf{f}^{-1}(\mathbf{p}) = \{\mathbf{q}_\infty\}$ by the inverse function theorem \cite[Theorem 9.24]{Rudin1964}, so semistability coincides with asymptotic stability.
\end{proof}

\begin{corollary}
	\label{cor:stability}
	Let $\mathsf{S}\in\mathbb{S}^\bullet$ be as in \eqref{eqn:stability_kinematic_system}, $\alpha\in\overline{1,4}$, $\mathbf{u}$ be the $\boldsymbol{\pi}_\alpha$-PIK solution of $[\mathsf{S}]^\bullet$ with the damping functions given by \eqref{eqn:damping_function}.
	Assume $\mu_1,\dots,\mu_l,\nu\in[0,\infty)$ if $\alpha\in\overline{1,3}$, (A5) if $\alpha = 2$, $\mathbf{q}_\infty\in\mathbf{f}^{-1}(\mathbf{p})\cap\mathcal{H}_\mathsf{S}$, and $\mathrm{rank}(\boldsymbol{\Psi}(\mathbf{q}_\infty)) = m$.
	Then, the equilibrium point $\mathbf{q}_\infty$ of \eqref{eqn:autonomous_system} is semistable. If $m = n$, then $\mathbf{q}_\infty$ is asymptotically stable.
\end{corollary}
\begin{proof}
	Since $\mathbf{D}_a(\mathbf{q}) = \mathbf{D}_a^T(\mathbf{q})\ge 0$ and $\mathbf{H}_a(\mathbf{q}) = \mathbf{H}_a^T(\mathbf{q})\ge0$ for all $a\in\overline{1,l}$ and $\mathbf{q}\in \mathbb{R}^n$, (A4) holds for all $\alpha\in\overline{1,4}$.
	Since $\mathbf{C}_{aa}\mathbf{C}_{aa}^T\mathbf{D}_a = \mathbf{C}_{aa}\mathbf{C}_{aa}^* = \mathbf{D}_a\mathbf{C}_{aa}\mathbf{C}_{aa}^T$, (A5) holds for all $\alpha\in\{1,3,4\}$.
	Since $\mathbf{C}_{aa}\hat{\mathbf{J}}_a\in C_{\mathcal{H}(\mathsf{S})}^0$ and $\mathrm{rank}(\mathbf{C}_{aa}(\mathbf{q})) = m_a$ for all $\mathbf{q}\in\mathcal{H}_\mathcal{S}$, we have $\lambda_a^2 = \mu_a^2/|\mathbf{C}_{aa}\mathbf{C}_{aa}^T|^\nu\in C_{\mathcal{H}(\mathsf{S})}^0$ and $\mathbf{D}_a,\mathbf{H}_a\in C_{\mathcal{H}(\mathsf{S})}^0$.
	It follows that $\mathbf{L}\in C_{\mathcal{H}(\mathsf{S})}^0$ and $\mathcal{H}_\mathsf{S} \subset \{\mathbf{q}\in\mathbb{R}^n\mid \mathrm{rank}(\mathbf{L}(\mathbf{q})) = m\}$.
	The proof is completed by Theorem \ref{thm:stability}.
\end{proof}

\section{Example}
\label{sec:example}

A minimal example that shows $\bullet$-discontinuity of PIK solutions is a two-link manipulator whose forward kinematic function is given as
\begin{equation*}
\mathbf{f}(t,\mathbf{q}) = \begin{bmatrix} f_1(t,\mathbf{q}) \\ f_2(t,\mathbf{q}) \end{bmatrix} = \begin{bmatrix} \mathsf{L}_1\cos(q_1) + \mathsf{L}_2\cos(q_1+q_2) \\ \mathsf{L}_1\sin(q_1) + \mathsf{L}_2\sin(q_1+q_2) \end{bmatrix}
\end{equation*}
where $\mathsf{L}_1$ and $\mathsf{L}_2$ are link lengths, $q_1$ and $q_2$ are joint angles, $\mathbf{q} = (q_1,q_2) \in \mathbb{R}^2$, and $(x,y) = (f_1(t,\mathbf{q}),f_2(t,\mathbf{q}))\in\mathbb{R}^2$ is the position of the end-effector in the $xy$-plane. 
Let $\bullet\in\mathcal{I}$ and denote $\mathsf{L} = \mathsf{L}_1+\mathsf{L}_2$, $\mathsf{L}'=|\mathsf{L}_1-\mathsf{L}_2|$, and $\mathbf{J} = \begin{bmatrix} \mathbf{j}_1^T & \mathbf{j}_2^T\end{bmatrix}^T = \mathsf{D}_q\mathbf{f}$.
Assign priority to the $x$-directional motion over the $y$-directional motion of the end-effector. 
We can find the QR decomposition of $\mathbf{J}^T$ given by Lemma \ref{lem:orthogonalization_of_J} as
\begin{align*}
	\mathbf{J}
		&= \begin{bmatrix} c_{11} & 0 \\ c_{21} & c_{22}\end{bmatrix} \begin{bmatrix} \hat{\mathbf{j}}_1 \\ \hat{\mathbf{j}}_2\end{bmatrix} \\
		&= \begin{dcases*}
		\begin{bmatrix} 0 & 0 \\ 0 & \sqrt{\mathbf{j}_2\mathbf{j}_2^T} \end{bmatrix} \begin{bmatrix} \hat{\mathbf{j}}_1 \\ \frac{\mathbf{j}_2}{\sqrt{\mathbf{j}_2\mathbf{j}_2^T}} \end{bmatrix}, & $\begin{array}{c}q_1=0\\q_2=0\end{array}$
		\\
		\begin{bmatrix} \sqrt{\mathbf{j}_1\mathbf{j}_1^T} & 0 \\ \frac{\mathbf{j}_1\mathbf{j}_2^T}{\sqrt{\mathbf{j}_1\mathbf{j}_1^T}} & 0 \end{bmatrix} \begin{bmatrix} \frac{\mathbf{j}_1}{\sqrt{\mathbf{j}_1\mathbf{j}_1^T}} \\ \hat{\mathbf{j}}_2 \end{bmatrix}, & $\begin{array}{c}q_1\neq0\\q_2=0\end{array}$
		\\
		\begin{bmatrix} \sqrt{\mathbf{j}_1\mathbf{j}_1^T} & 0 \\ \frac{\mathbf{j}_1\mathbf{j}_2^T}{\sqrt{\mathbf{j}_1\mathbf{j}_1^T}} & \sqrt{\mathbf{j}_2\mathbf{N}_1\mathbf{j}_2^T} \end{bmatrix} \begin{bmatrix} \frac{\mathbf{j}_1}{\sqrt{\mathbf{j}_1\mathbf{j}_1^T}} \\ \frac{\mathbf{j}_2\mathbf{N}_1}{\sqrt{\mathbf{j}_2\mathbf{N}_1\mathbf{j}_2^T}} \end{bmatrix}, & $\begin{array}{c}q_2\neq0\end{array}$
		\end{dcases*}
\end{align*}
for all $\mathbf{q}\in\Omega = [-\frac{\pi}{2},\frac{\pi}{2}]^2$ where $\mathbf{N}_1 = \mathbf{I}_2 - (\mathbf{j}_1^T\mathbf{j}_1)/(\mathbf{j}_1\mathbf{j}_1^T)$. 
Note that $\hat{\mathbf{j}}_1^T(t,\mathbf{0})$ and $\hat{\mathbf{j}}_2^T(t,\mathbf{q})$ for $\mathbf{q}\in\{(q_1,q_2)\in\Omega\mid q_1\neq0,\,q_2=0\}$ should be chosen from $\mathcal{N}(\mathbf{J}(t,\mathbf{0}))$ and $\mathcal{N}(\mathbf{J}(t,\mathbf{q}))$, respectively. 
If $\mathbf{q}_0\in\Omega\setminus\{\mathbf{0}\}$, then there exists $\epsilon>0$ such that $\mathbf{j}_1(\mathbf{x})\neq\mathbf{0}$ and $\hat{\mathbf{j}}_1(\mathbf{x}) = (\mathbf{j}_1/\sqrt{\mathbf{j}_1\mathbf{j}_1^T})(\mathbf{x})$ for every $\mathbf{x}\in\mathbb{R}\times(\mathbf{q}_0+\epsilon B_2)$. 
So, $\hat{\mathbf{j}}_1\in C_{(t,\mathbf{q}_0)}^\bullet$ and $\mathcal{B}_{(t,\mathbf{q}_0)}^\bullet = \{\{\hat{\mathbf{j}}_1\},\{\hat{\mathbf{j}}_2\}\}$ for all $t$ by Proposition \ref{prp:properties_of_MCBS}. 
Let $t_0\in\mathbb{R}$ be arbitrary and $\mathbf{x}_0 = (t_0,\mathbf{0})$.
We find that $\mathcal{B}_{\mathbf{x}_0}^\bullet = \{\{\hat{\mathbf{j}}_1,\hat{\mathbf{j}}_2\}\}$ from
\begin{equation*}
(c_{22}\hat{\mathbf{j}}_2^T)(\mathbf{x}) = \begin{dcases*}
(\mathsf{L}_1+\mathsf{L}_2, \mathsf{L}_2), & $q_1 = q_2 = 0$ \\
(0, 0), & $q_1\neq0,\,q_2=0$ \\
(\mathsf{L}_1/2,-\mathsf{L}_1/2), & $q_1=0,\,q_2\neq0$ \\
(0,\mathsf{L}_2), & $q_1 = -q_2 \neq 0$.
\end{dcases*}
\end{equation*}
Observe that $\mathbf{P}(\{\hat{\mathbf{j}}_2\})$ is purely $\bullet$-discontinuous at $\mathbf{x}_0$ and $[\mathbf{P}(\{\hat{\mathbf{j}}_2\})]_{\mathbf{x}_0}^\bullet(\mathbf{x}_0) = \mathbf{P}(\{\hat{\mathbf{j}}_2\},\mathbf{x}_0)$. 
Since $\mathcal{F}(\{\hat{\mathbf{j}}_2\},\mathbf{x}_0) = \{\hat{\mathbf{j}}_2\}$, we have $[\mathbf{P}(\{\hat{\mathbf{j}}_2\})]_{\mathbf{x}_0}^\bullet(\mathbf{x}_0) = \mathbf{P}(\mathcal{F}(\{\hat{\mathbf{j}}_2\},\mathbf{x}_0),\mathbf{x}_0)$. 
Therefore, there does not exist a $\bullet$-continuous SPIK solution of the equivalence class of the kinematic system $\mathsf{S}_0 = (2,(1,1),2,\mathsf{D}\mathbf{f},\mathbf{I}_2,\mathbf{0})$ by Theorem \ref{thm:smoothness_of_the_PIK_solution}. 

We showed $(t,\mathbf{0})\not\in\mathcal{G}_{\mathsf{S}_0}^\bullet$ for all $t\in\mathbb{R}$ and $\bullet\in\mathcal{I}$.
Observe $\mathsf{D}_qf_1(t,\mathbf{0}) = (c_{11}\hat{\mathbf{j}}_1)(t,\mathbf{0}) = \mathbf{0}$.
Indeed, $X\setminus\mathcal{G}_{\mathsf{S}_0}^\bullet = \{\mathbf{x}\in X\mid c_{11}(\mathbf{x}) = 0\} = \mathbb{R}\times\pi\mathbb{Z}^2$ and $\mathcal{G}_{\mathsf{S}_0}^\bullet = \mathrm{int}(\mathcal{G}_{\mathsf{S}_0}^\bullet) = \mathbb{R}\times(\mathbb{R}^2\setminus\pi\mathbb{Z}^2)$ for all $\bullet\in\mathcal{I}$.
One can easily check that $f_1$ has its maximum value at $\mathbf{q}\in2\pi\mathbb{Z}^2$ and its minimum value at $\mathbf{q}\in(\pi,0) + 2\pi\mathbb{Z}^2$ from the Hessian matrix of $f_1$ at each $\mathbf{q}\in\pi\mathbb{Z}^2$
\begin{equation*}
\mathsf{D}_q(c_{11}\hat{\mathbf{j}}_1^T)(\mathbf{x})
= \begin{dcases*}
\begin{bmatrix} -\mathsf{L}_1-\mathsf{L}_2 & -\mathsf{L}_2 \\ -\mathsf{L}_2 & -\mathsf{L}_2\end{bmatrix}, & $\mathbf{q}\in\begin{bmatrix}0\\0\end{bmatrix}+2\pi\mathbb{Z}^2$ \\
\begin{bmatrix} \mathsf{L}_1+\mathsf{L}_2 & \mathsf{L}_2 \\ \mathsf{L}_2 & \mathsf{L}_2\end{bmatrix}, & $\mathbf{q}\in\begin{bmatrix}\pi\\0\end{bmatrix}+2\pi\mathbb{Z}^2$ \\
\begin{bmatrix} -\mathsf{L}_1+\mathsf{L}_2 & \mathsf{L}_2 \\ \mathsf{L}_2 & \mathsf{L}_2 \end{bmatrix}, & $\mathbf{q}\in\begin{bmatrix}0\\\pi\end{bmatrix}+2\pi\mathbb{Z}^2$ \\
\begin{bmatrix} \mathsf{L}_1-\mathsf{L}_2 & -\mathsf{L}_2 \\ -\mathsf{L}_2 & -\mathsf{L}_2 \end{bmatrix}, & $\mathbf{q}\in\begin{bmatrix}\pi\\\pi\end{bmatrix}+2\pi\mathbb{Z}^2$.
\end{dcases*}
\end{equation*}
Let $Y_1 = \pi\mathbb{Z}\times 2\pi\mathbb{Z}$ and $Y_2 = \pi\mathbb{Z}\times(\pi + 2\pi\mathbb{Z})$.
$f_1(X) = [-\mathsf{L},\mathsf{L}]$, $f_1^{-1}(\{\mathsf{L},-\mathsf{L}\}) = \mathbb{R}\times Y_1$, and $f_1^{-1}(\{\mathsf{L}',-\mathsf{L}'\} \supset \mathbb{R}\times Y_2$.
$\mathsf{D}_q(c_{11}\hat{\mathbf{j}}_1^T)$ is symmetric and positive or negative definite on $\mathbb{R}\times Y_1$.
Positive or negative definiteness of $\mathsf{D}_q(c_{11}\hat{\mathbf{j}}_1^T)$ on $\mathbb{R}\times Y_2$ depends on the values of $\mathsf{L}_1$ and $\mathsf{L}_2$.
Let $\mathbf{x}\in X\setminus\mathcal{G}_{\mathsf{S}_0}^\bullet$.
Since $\mathbf{J}\in C_\mathbf{x}^1$ and $\mathsf{D}\mathbf{J}\in C_\mathbf{x}^{L_p}$, there exist $a_\mathbf{x}\in\overline{1,2}$ and $r_\mathbf{x},L_\mathbf{x}\in(0,\infty)$ satisfying \eqref{eqn:priority_singularity_condition1} to \eqref{eqn:priority_singularity_condition3} by Lemma \ref{lem:Lipschitz_continuity_of_linear_approximation}.
$c_{11}(\mathbf{x}) = c_{21}(\mathbf{x}) = 0$ by Lemma \ref{lem:orthogonalization_of_J}.
Since $\mathbf{J}(\mathbf{x}) \neq \mathbf{0}$, we have $c_{22}(\mathbf{x}) \neq 0$ and $a_\mathbf{x} = 1$.
Indeed, if $a_\mathbf{x} = 1$, then we can choose any $r_\mathbf{x},L_\mathbf{x}\in(0,\infty)$ provided $\|\mathsf{D}\mathbf{j}_1^T(\mathbf{x}') - \mathsf{D}\mathbf{j}_1^T(\mathbf{x})\| \le n^{-3/2}L_\mathbf{x}\|\mathbf{x}' - \mathbf{x}\|$ for all $\mathbf{x}' \in \mathbf{x} + r_\mathbf{x} B_X$; see the proof of Lemma \ref{lem:Lipschitz_continuity_of_linear_approximation}.
Since $\mathbf{j}_1$ is periodic, we can let $r_1 = \min\{r_\mathbf{x}\mid \mathbf{x}\in X\setminus\mathcal{G}_{\mathsf{S}_0}^\bullet\}$ and $L_1 = \max\{L_\mathbf{x}\mid\mathbf{x}\in X\setminus\mathcal{G}_{\mathsf{S}_0}^\bullet\}$.
Let $\alpha \in\overline{1,4}$ and $\mathbf{u}$ be the $\boldsymbol{\pi}_\alpha$-PIK solution of $[\mathsf{S}_0]^\bullet$ with the damping functions given by \eqref{eqn:damping_function}.
Let $\mu_1,\mu_2\in(0,\infty)$ and $\nu\in\mathbb{N}$ if $\alpha\in\overline{1,3}$.
We construct a desired end-effector trajectory $\mathbf{p} = (p_1,p_2)\in C^{1_p}(\mathbb{R},\mathbb{R}^2)$ under the conditions that $\dot{\mathbf{p}}$ is bounded and $\int_{-\infty}^\infty\|\dot{\mathbf{p}}(t)\|dt<\infty$ and select the feedback gain matrix $\mathbf{K} = \mathrm{diag}(k_1,k_2)\in\mathbb{R}^{2\times 2}$ satisfying $k_1,k_2\in(0,\infty)$.
Since $\mathbf{f}$ is bounded, there exists $L_2\in(0,\infty)$ such that $\|\dot{\mathbf{p}}(t) + \mathbf{K}(\mathbf{p}(t) - \mathbf{f}(\mathbf{x}))\| \le \|\dot{\mathbf{p}}(t)\| + \|\mathbf{K}\|\left(\|\mathbf{p}(t_0)\| + \int_{-\infty}^\infty\|\dot{\mathbf{p}}(s)\|ds + \|\mathbf{f}(\mathbf{x})\|\right) \le L_2$ for all $\mathbf{x}\in X$.
Let $r_2\in(0,\infty)$ be arbitrary and design the activation function $\boldsymbol{\Psi} = \mathrm{diag}(\psi_1,\psi_2)\in C^\bullet(X,\mathbb{R}^{2\times 2})$ with $\psi_1,\psi_2:X\to[0,1]$ satisfying
\begin{itemize}
	\item $\psi_a(\mathbf{x}) = \psi_a(0,\mathbf{q}+(2\pi,2\pi))>0$ for all $a\in\{1,2\}$ and $\mathbf{x}\in \mathcal{G}_{\mathsf{S}_0}^\bullet$;
	\item $\psi_2(t',\mathbf{q}') \le \|\mathbf{q}'-\mathbf{q}\|$ for every $\mathbf{x}\in\mathbb{R}\times Y_2$ and $(t',\mathbf{q}')\in\mathbf{x}+r_2B_X$;
	\item if $\alpha = 4$, then $\psi_1(t',\mathbf{q}') \le \|\mathbf{q}'-\mathbf{q}\|$ for every $\mathbf{x}\in\mathbb{R}\times Y_1$ and $(t',\mathbf{q}')\in\mathbf{x}+r_2B_X$.
\end{itemize}
Let $r = \min\{r_1,r_2\}$, $L = \max\{L_1,L_2\}$, and $\mathbf{r} = \boldsymbol{\Psi}(\dot{\mathbf{p}}+\mathbf{K}(\mathbf{p}-\mathbf{f}))$.
Then, we see that the kinematic system $\mathsf{S} = (2,(1,1),2,\mathsf{D}\mathbf{f},\mathbf{I}_2,\mathbf{r})\in[\mathsf{S}_0]^\bullet$ satisfies the assumptions of Corollary \ref{cor:existence_of_joint_trajectory}.
Therefore, for each $(t_0,\mathbf{q}_0)\in\mathcal{G}_\mathsf{S}^\bullet$ there exists a classical solution $\mathbf{q}:[t_0,\infty)\to\mathcal{G}_\mathsf{S}^\bullet$ of \eqref{eqn:differential_equation} satisfying $\mathbf{q}(t_0) = \mathbf{q}_0$.
If $\dot{\mathbf{p}}\in C^L$, then the solution is unique.
Now, we are ready to investigate the task convergence of the $\boldsymbol{\pi}_\alpha$-PIK solution of $\mathsf{S}$.

Fix $(t_0,\mathbf{q}_0)\in\mathcal{G}_{\mathsf{S}_0}^\bullet$ and let $\mathbf{q}:[t_0,\infty)\to\mathbb{R}^2$ be a classical solution of \eqref{eqn:differential_equation} satisfying $\mathbf{q}(t_0) = \mathbf{q}_0$ and $(t,\mathbf{q}(t))\in\mathcal{G}_{\mathsf{S}_0}^\bullet$ for all $t\in[t_0,\infty)$.
Define $I_1 = \dot{\phi}_1^{-1}((0,\infty))$ and $I_2 = [t_0,\infty)\setminus I_1$.
Since $\dot{\phi}_1$ is measurable in $[t_0,\infty)$ and $(0,\infty)$ is open, $I_1$ and $I_2$ are measurable.
Then, $\int_{I_1}\dot{\phi}_1(t)dt \le \int_{t_0}^\infty|\gamma_1(t)|dt <\infty$ and $\int_{I_2}\dot{\phi}_1(t)dt = \int_{t_0}^\infty\dot{\phi}_1(t)dt - \int_{I_1}\dot{\phi}_1(t)dt\le \int_{t_0}^\infty|\gamma_1(t)|dt - \int_{I_1}\dot{\phi}_1(t)dt<\infty$.
Thus, $\int_{t_0}^\infty|\dot{\phi}_1(t)|dt = \int_{I_1}\dot{\phi}_1(t)dt - \int_{I_2}\dot{\phi}_1(t)dt<\infty$.
Define $I_3 = \{t\in[t_0,\infty)\mid e_1(\mathbf{x}(t))\neq 0\}$ and $I_4 = [t_0,\infty)\setminus I_3$.
Since $|\dot{\phi}_1(t)| = |\phi_1^+(t)e_1(\mathbf{x}(t))\dot{e}_1(\mathbf{x}(t))| = |\dot{e}_1(\mathbf{x}(t))|$ for all $t\in I_3$ and $\int_{I_4}|\dot{e}_1(\mathbf{x}(t))|dt = 0$, we have $\int_{t_0}^\infty|\dot{e}_1(\mathbf{x}(t))|dt = \int_{I_3}|\dot{\phi}_1(t)|dt < \infty$.
Let $t_0\le t_1<t_2<\cdots$ be a divergent sequence and $\epsilon>0$.
Since $\int_{t_0}^\infty|\dot{e}_1(\mathbf{x}(t))|dt<\infty$, there exists $N$ such that $|e_1(\mathbf{x}(t_i)) - e_1(\mathbf{x}(t_j))|\le\int_{t_N}^\infty|\dot{e}_1(\mathbf{x}(t))|dt<\epsilon$ for all $i,j\ge N$.
Thus, $\{e_1(\mathbf{x}(t_i))\}$ is Cauchy and $e_1(\mathbf{x}(t_i))\to e_1^*\in\mathbb{R}$ as $i\to\infty$.
Since it holds for every divergent sequence $\{t_i\}$, $\lim_{t\to\infty}e_1(\mathbf{x}(t)) = e_1^*$.
Since $\int_{t_0}^\infty|\dot{p}_a(t)|dt<\infty$, we also have $\lim_{t\to\infty}p_a(t) = p_a^*\in\mathbb{R}$.
It follows that $\lim_{t\to\infty}f_1(\mathbf{x}(t)) = f_1^*\in [-\mathsf{L},\mathsf{L}]$.
Since the kinematic system $\mathsf{S}$ satisfies all the assumptions of Corollary \ref{cor:task_convergence}, we have $\int_{t_0}^\infty(\psi_1c_{11}l_{11}e_1)^2(\mathbf{x}(t))dt<\infty$.

\textit{1)} 
Assume $f_1^*\in\{\mathsf{L},-\mathsf{L}\}$. 
Since $f_1^{-1}(\{\mathsf{L},-\mathsf{L}\}) = \mathbb{R}\times Y_1$, $\lim_{t\to\infty}\mathbf{q}(t) = \mathbf{q}_\infty\in Y_1$ and $\lim_{t\to\infty}f_2(\mathbf{x}(t)) = 0$.

\textit{2)}
Assume $f_1^*\in(-\mathsf{L},\mathsf{L})\setminus\{\mathsf{L}',-\mathsf{L}'\}$.
Since $c_{11}$ and $\psi_1$ are periodic, there exist $T\in[t_0,\infty)$ and $\sigma>0$ satisfying $(\psi_1c_{11})(\mathbf{x}(t))\ge \sigma$ for all $t\in[T,\infty)$.
Let $l_{ab}$ be the $(a,b)$-th entry of $\mathbf{L}$.
By Corollary \ref{cor:task_convergence}, we have $\int_{t_0}^\infty|e_1(\mathbf{x}(t))|dt < \infty$, $\lim_{t\to\infty}e_1(\mathbf{x}(t)) = 0$, and $\int_{t_0}^\infty(\psi_2c_{22}l_{22}e_2)^2(\mathbf{x}(t))dt<\infty$.
By the same manner, $\lim_{t\to\infty}f_2(\mathbf{x}(t)) = f_2^*\in[-\mathsf{L},\mathsf{L}]$ and $\lim_{t\to\infty}\mathbf{q}(t) = \mathbf{q}_\infty\in\mathbf{f}^{-1}(f_1^*,f_2^*)$.
If $\mathsf{L}'^2<(f_1^*)^2 + (f_2^*)^2<\mathsf{L}^2$, then there exist $T\in[t_0,\infty)$ and $\sigma>0$ satisfying $(\psi_2c_{22})(\mathbf{x}(t))\ge\sigma$ for all $t\in[T,\infty)$.
By Corollary \ref{cor:task_convergence}, $\int_{t_0}^\infty|e_2(\mathbf{x}(t))|dt<\infty$ and $\lim_{t\to\infty}e_2(\mathbf{x}(t)) = 0$.
If $(f_1^*)^2+(f_2^*)^2 \in \{\mathsf{L}'^2,\mathsf{L}^2\}$, then we can only guarantee $\lim_{t\to\infty}\eta_2(t) = 0$.

\textit{3)} 
Assume $f_1^*\in\{\mathsf{L}',-\mathsf{L}'\}$.
If there exist $r>0$ and $T\in[t_0,\infty)$ such that $|f_2(\mathbf{x}(t))|\ge r$ for all $t\in[T,\infty)$, then we have the same results of the case 2).
Assume that there exists $t_0\le t_1<t_2<\cdots$ such that $t_i\to\infty$ and $f_2(\mathbf{x}(t_i))\to0$ as $i\to\infty$.
We prove $\lim_{t\to\infty}f_2(\mathbf{x}(t)) = 0$ by contradiction.
Suppose that there exists $r_0>0$ such that for every $T\in[t_0,\infty)$ there exists $t\in[T,\infty)$ satisfying $f_2(\mathbf{x}(t)) \ge r_0$; the case $f_2(\mathbf{x}(t)) \le -r_0$ can be proven similarly.
Let $r\in(0,\min\{r_0,\mathsf{L}/2\}]$ be arbitrary.
There exists $t_0\le t_1'<t_2'<\cdots$ satisfying $t_i'\to\infty$ as $i\to\infty$ and $f_2(\mathbf{x}(t_i')) = r$ for all $i\in\mathbb{N}$.
Without loss of generality, $t_i < t_i' < t_{i+1}$ and $f_2(\mathbf{x}(t_i)) < r/2$ for all $i\in\mathbb{N}$.
Since $\mathbf{u}$ is bounded, there exists $\delta_0\in(0,\infty)$ such that $|f_2(\mathbf{x}(t+\delta_0)) - f_2(\mathbf{x}(t))|<r/2$ for all $t\in[t_0,\infty)$.
Observe that $\int_{t_0}^\infty\eta_1^2(t)dt<\infty$ implies $\lim_{t\to\infty}\int_t^{t+\delta}\eta_1(s)ds = 0$ for all $\delta\in(0,\infty)$.
Let $\delta\in(0,\delta_0]$ be arbitrary.
There exist $t_i^-\in[t_i,t_i']$ and $t_i^+\in[t_i',t_{i+1}]$ satisfying $f_2(\mathbf{x}([t_i^-,t_i^-+\delta])),f_2(\mathbf{x}([t_i^+,t_i^++\delta]))\subset[r/2,r]$, $f_2(\mathbf{x}(t_i^-+\delta)) - f_2(\mathbf{x}(t_i^-)) > 0$, and $f_2(\mathbf{x}(t_i^++\delta)) - f_2(\mathbf{x}(t_i^+)) < 0$.
Denote $\alpha = k_2\psi_2c_{22}^2l_{22}$ and $\beta = \alpha(p_2 - p_2^*) + \psi_1(c_{21}c_{11}l_{11} + c_{22}^2l_{21})\dot{p}_1 + \psi_2c_{22}^2l_{22}\dot{p}_2 + k_1(c_{21} + c_{22}m_{21})\eta_1$ where $m_{21}$ is as in (A6).
Then, $df_2/dt = \mathsf{D}_qf_2\mathbf{u} = \alpha (p_2^* - f_2) + \beta$ and $\lim_{t\to\infty}\int_t^{t+\delta}\beta(\mathbf{x}(s))ds = 0$.
Since $f_1(\mathbf{x}(t)) \to f_1^*\in\{\mathsf{L}',-\mathsf{L}'\}$ as $t\to\infty$, there exists $N\in\mathbb{N}$ such that $0<\alpha_1 \le \alpha(\mathbf{x}(t)) \le \alpha_2<\infty$ for all $t\in[t_i^-,t_i^-+\delta]\cup[t_i^+,t_i^++\delta]$ and $i\ge N$.
It follows that $\lim_{i\to\infty}\int_{t_i^-}^{t_i^-+\delta}(p_2^* - f_2(\mathbf{x}(t)))dt \ge 0$ and $\lim_{i\to\infty}\int_{t_i^+}^{t_i^++\delta}(p_2^* - f_2(\mathbf{x}(t)))dt \le 0$.
So, we find a contradiction $r/2\le p_2^*\le r$ that $r\in(0,\min\{r_0,\mathsf{L}/2\}]$ is arbitrary.
Therefore, $\lim_{t\to\infty}f_2(\mathbf{x}(t)) = 0$ and if $\mathsf{L}_1\neq\mathsf{L}_2$, then $\lim_{t\to\infty}\mathbf{q}(t) = \mathbf{q}_\infty\in Y_2$.

In the analysis of the task convergence, we observed that every task trajectory $\mathbf{p}(t)$ and joint trajectory $\mathbf{q}(t)$ satisfying aforementioned conditions converge to points $\mathbf{p}_\infty\in\mathbb{R}^2$ and $\mathbf{q}_\infty\in\mathbb{R}^2$, respectively.
However, the convergence of the task error $\mathbf{e}(\mathbf{x}(t))$ to zero is guaranteed only when $\mathbf{q}_\infty = \mathbf{f}^{-1}(\mathbf{p}_\infty)\in \mathcal{H}_{\mathsf{S}_0}$.
Since the kinematic system $\mathsf{S}_1 = (2,(1,1),2,\mathsf{D}\mathbf{f},\mathbf{I}_2,\boldsymbol{\Psi}\mathbf{K}(\mathbf{p}_\infty - \mathbf{f}))$ satisfies all the assumptions of Corollary \ref{cor:stability}, we see that $\mathbf{q}_\infty\in\mathcal{H}_{\mathsf{S}_0}$ is an asymptotically stable equilibrium point of \eqref{eqn:autonomous_system}.

\section{Conclusion}
\label{sec:conclusion}

We have presented various theoretical properties of a class of PIK solutions related to nonsmoothness, trajectory existence, task convergence, and stability.
We found a sufficient condition in which for every strongly proper objective function there exists a smooth reference such that the PIK solution is nonsmooth.
For nonsmooth PIK solutions, we constructed an alternative existence and uniqueness theorem of a joint trajectory by using structural information of PIK solutions.
We found a few task convergence properties of PIK solutions when the tasks are designed to follow some desired task trajectories.
We analyzed stability of the differential equation whose right hand side is a PIK solution when the tasks are designed to reach some desired task positions.
We applied our findings to a two-link manipulator in order to show how a PIK solution can be designed to guarantee trajectory existence, task convergence, and stability in the existence of nonsmoothness.

\ifCLASSOPTIONcaptionsoff
  \newpage
\fi

% Generated by IEEEtran.bst, version: 1.13 (2008/09/30)

\begin{IEEEbiography}[{\includegraphics[width=1in,height=1.25in,clip,keepaspectratio]{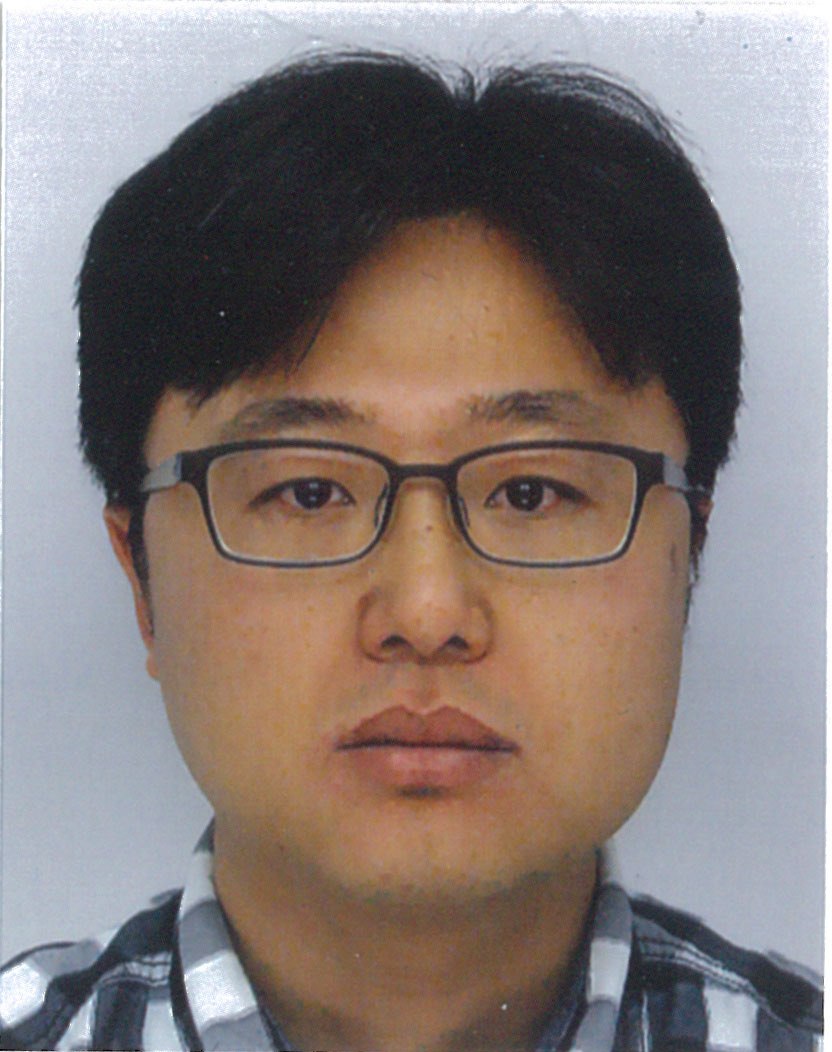}}]{Sang-ik An}
	received the B.S. in both mechanical and electronic engineerings from Korea Aerospace University, Korea in 2008 and the M.S. in mechanical engineering from Korea Advanced Institute of Science and Technology, Korea in 2010. He is currently working towards the Ph.D. degree in electrical and computer engineering at Technical University of Munich, Germany.
\end{IEEEbiography}
\begin{IEEEbiography}[{\includegraphics[width=1in,height=1.25in,clip,keepaspectratio]{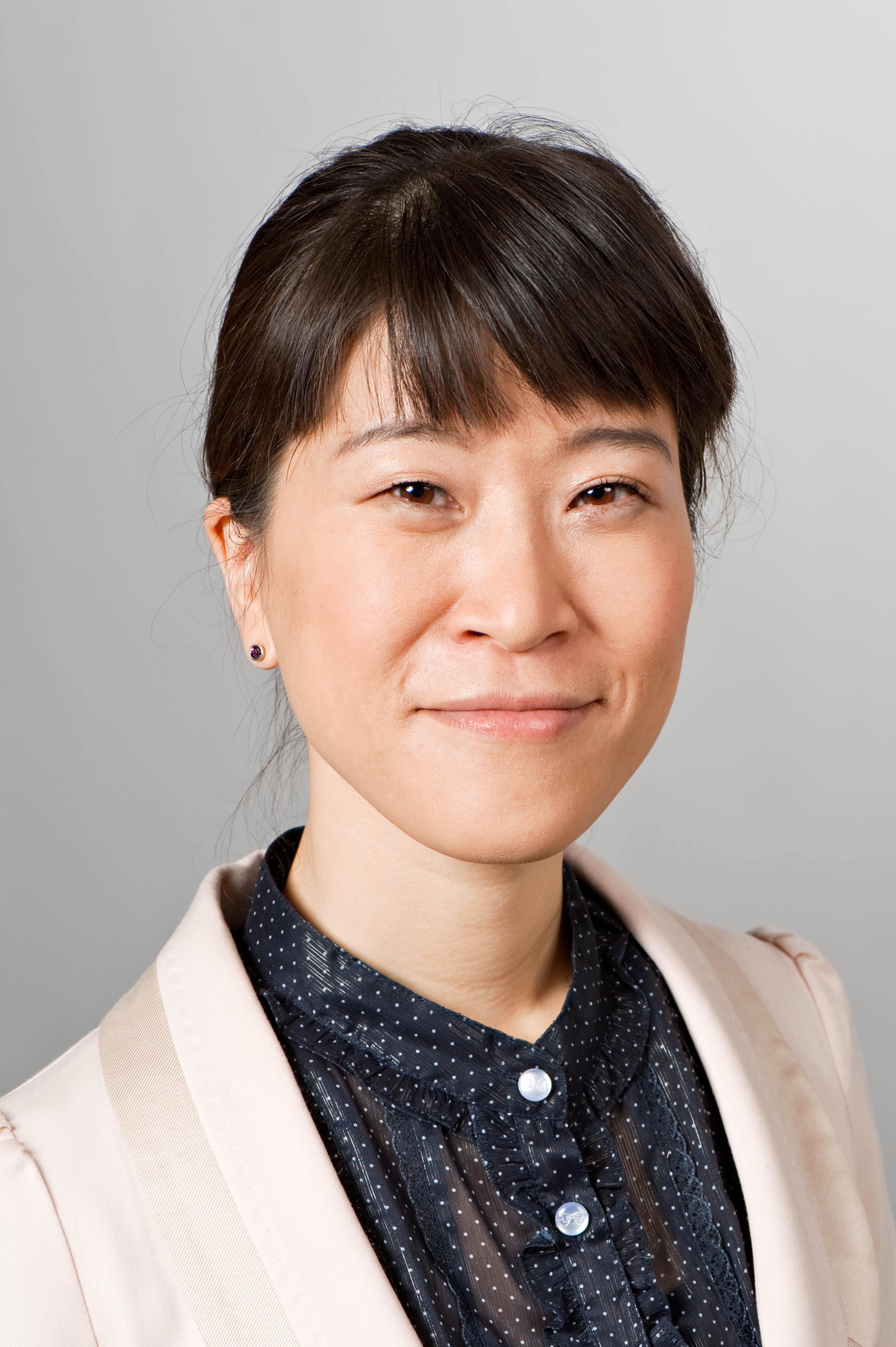}}]{Dongheui Lee}
	is Associate Professor of Human-centered Assistive Robotics at the TUM Department of Electrical and Computer Engineering. She is also director of a Human-centered assistive robotics group at the German Aerospace Center (DLR). Her research interests include human motion understanding, human robot interaction, machine learning in robotics, and assistive robotics.
	
	Prior to her appointment as Associate Professor, she was an Assistant Professor at TUM (2009-2017), Project Assistant Professor at the University of Tokyo (2007-2009), and a research scientist at the Korea Institute of Science and Technology (KIST) (2001-2004). After completing her B.S. (2001) and M.S. (2003) degrees in mechanical engineering at Kyung Hee University, Korea, she went on to obtain a PhD degree from the department of Mechano-Informatics, University of Tokyo, Japan in 2007. She was awarded a Carl von Linde Fellowship at the TUM Institute for Advanced Study (2011) and a Helmholtz professorship prize (2015). She is coordinator of both the euRobotics Topic Group on physical Human Robot Interaction and of the TUM Center of Competence Robotics, Autonomy and Interaction.
\end{IEEEbiography}
\vfill

\end{document}